\documentclass[12pt,english]{article}
\usepackage[T1]{fontenc}
\usepackage[utf8]{inputenc}
\usepackage{geometry}
\usepackage{xcolor}
\definecolor{linkmuted}{RGB}{58,78,104}
\definecolor{citesky}{RGB}{0,0,139}
\definecolor{eqred}{RGB}{180,30,45}
\definecolor{weblight}{RGB}{110,150,190}
\usepackage{babel}
\usepackage{array}
\usepackage{float}
\usepackage{booktabs}
\usepackage{mathtools}
\usepackage{url}
\usepackage{bm}
\usepackage{multirow}
\usepackage{amsmath}
\allowdisplaybreaks
\usepackage{amsthm}
\usepackage{amssymb}
\usepackage{graphicx}
\graphicspath{{figures/}}
\usepackage{setspace}
\usepackage[authoryear]{natbib}
\usepackage{bibunits}
\makeatletter
\let\NAT@fullfalse\NAT@fulltrue
\makeatother
\usepackage[
 hypertexnames=false,
 bookmarks=false,bookmarksnumbered=false,bookmarksopen=false,
 breaklinks=false,pdfborder={0 0 0},pdfborderstyle={},backref=false,colorlinks=true
]{hyperref}
\hypersetup{
 citecolor = citesky,
 linkcolor = eqred,
 urlcolor = weblight,
 filecolor = weblight}

\makeatletter

\providecommand{\tabularnewline}{\\}

\newtheoremstyle{remarkbf}%
  {3pt}%
  {3pt}%
  {\normalfont}%
  {}%
  {\normalfont\bfseries}%
  {.}%
  {.5em}%
  {}%
\theoremstyle{remarkbf}
\newtheorem{rem}{\protect\remarkname}
\theoremstyle{plain}
\newtheorem{assumption}{\protect\assumptionname}
\theoremstyle{plain}
\newtheorem{lem}{\protect\lemmaname}
\theoremstyle{plain}
\newtheorem{thm}{\protect\theoremname}

\usepackage{mathrsfs}
\usepackage{dsfont}
\usepackage{indentfirst}
\usepackage{enumitem}
\usepackage{threeparttable}
\usepackage{bookmark}
\theoremstyle{plain}
\newtheorem{prop}{\protect\propositionname}
\theoremstyle{definition}
\newtheorem{Module}{Module}
\theoremstyle{plain}
\allowdisplaybreaks

\g@addto@macro\normalsize{%
 \abovedisplayskip=7pt plus 1pt minus 2pt
 \abovedisplayshortskip=5.5pt plus 1pt minus 2pt
 \belowdisplayskip=7pt plus 1pt minus 2pt
 \belowdisplayshortskip=6.5pt plus 1pt minus 2pt
}{}{}

\definecolor{pAlgae}{RGB}{0, 0, 102}

\date{}
\defaultbibliographystyle{chicago}
\defaultbibliography{NetForm}

\ifdefined\showcaptionsetup
 \PassOptionsToPackage{caption=false}{subfig}
\fi
\usepackage{subfig}
\usepackage{tikz}
\usetikzlibrary{arrows.meta,positioning,shapes.geometric,fit}
\makeatother

\providecommand{\assumptionname}{Assumption}
\providecommand{\lemmaname}{Lemma}
\providecommand{\remarkname}{Remark}
\providecommand{\theoremname}{Theorem}
\providecommand{\propositionname}{Proposition}
\renewcommand{\hat}[1]{\widehat{#1}}
\renewcommand{\tilde}[1]{\widetilde{#1}}

\begin{document}
 
\global\long\def\qed{\ \qedsymbol}%
 \makeatletter \renewenvironment{proof}[1][\proofname]{\par\pushQED{\qed}\normalfont\topsep6\p@\@plus6\p@\relax\trivlist\item[\hskip\labelsep\normalfont\bfseries #1\@addpunct{.}]\ignorespaces}{\popQED\endtrivlist\@endpefalse}
\makeatother

\global\long\def\a{\alpha}%
\global\long\def\b{\beta}%
\global\long\def\g{\gamma}%
\global\long\def\d{\delta}%
\global\long\def\e{\epsilon}%
\global\long\def\l{\lambda}%
\global\long\def\t{\theta}%
\global\long\def\o{\omega}%
\global\long\def\s{\sigma}%
\global\long\def\G{\Gamma}%
\global\long\def\D{\Delta}%
\global\long\def\L{\Lambda}%
\global\long\def\T{\Theta}%
\global\long\def\O{\Omega}%
\global\long\def\R{\mathbb{R}}%
\global\long\def\N{\mathbb{N}}%
\global\long\def\Q{\mathbb{Q}}%
\global\long\def\I{\mathbb{I}}%
\global\long\def\P{P}%
\global\long\def\E{\mathbb{E}}%
\global\long\def\B{\mathbb{B}}%
\global\long\def\S{\mathbb{S}}%
\global\long\def\V{\mathbb{V}\text{ar}}%
\global\long\def\X{{\bf X}}%
\global\long\def\cX{\mathscr{X}}%
\global\long\def\cY{\mathscr{Y}}%
\global\long\def\cA{\mathscr{A}}%
\global\long\def\cB{\mathscr{B}}%
\global\long\def\cM{\mathscr{M}}%
\global\long\def\cN{\mathcal{N}}%
\global\long\def\cG{\mathcal{G}}%
\global\long\def\cC{\mathcal{C}}%
\global\long\def\sp{\,}%
\global\long\def\es{\emptyset}%
\global\long\def\mc#1{\mathscr{#1}}%
\global\long\def\ind{\mathbf{\mathds{1}}}%
\global\long\def\indep{\perp}%
\global\long\def\any{\forall}%
\global\long\def\ex{\exists}%
\global\long\def\p{\partial}%
\global\long\def\cd{\cdot}%
\global\long\def\Dif{\nabla}%
\global\long\def\imp{\Rightarrow}%
\global\long\def\iff{\Leftrightarrow}%
\global\long\def\up{\uparrow}%
\global\long\def\down{\downarrow}%
\global\long\def\arrow{\rightarrow}%
\global\long\def\rlarrow{\leftrightarrow}%
\global\long\def\lrarrow{\leftrightarrow}%
\global\long\def\abs#1{\left|#1\right|}%
\global\long\def\norm#1{\left\Vert #1\right\Vert }%
\global\long\def\rest#1{\left.#1\right|}%
\global\long\def\bracket#1#2{\left\langle #1\middle\vert#2\right\rangle }%
\global\long\def\sandvich#1#2#3{\left\langle #1\middle\vert#2\middle\vert#3\right\rangle }%
\global\long\def\third#1{\frac{#1}{3}}%
\global\long\def\ellipsis{\textellipsis}%
\global\long\def\sand#1{\left\lceil #1\right\vert }%
\global\long\def\wich#1{\left\vert #1\right\rfloor }%
\global\long\def\sandwich#1#2#3{\left\lceil #1\middle\vert#2\middle\vert#3\right\rfloor }%
\global\long\def\inprod#1{\left\langle #1\right\rangle }%
\global\long\def\ol#1{\overline{#1}}%
\global\long\def\ul#1{\underline{#1}}%
\global\long\def\td#1{\tilde{#1}}%
\global\long\def\bs#1{\boldsymbol{#1}}%
\global\long\def\upto{\nearrow}%
\global\long\def\downto{\searrow}%
\global\long\def\dto{\stackrel{d}{\to}}%
\global\long\def\asto{\rightarrow_{a.s.}}%
\global\long\def\gto{\rightarrow}%
\global\long\def\fto{\Rightarrow}%

\title{\vspace{-1.5cm}Bagging the Network\thanks{We thank Donald Andrews, Xiaohong Chen, Chih-Sheng Hsieh, Bryan Graham,
Matthew O. Jackson, Yuichi Kitamura, Oliver Linton, Shuyang Sheng, Martin Weidner, Weichen
Wang, Jun Yu, Yichong Zhang, and seminar participants at the University of Sydney, University of Macau, Singapore Management University, Singapore Workshop on Econometrics and Data Sciences, Asia Meeting of the Econometric Society (NYU Abu Dhabi), Cowles Conference on Econometrics (Yale), and Econometric Society World Congress for comments.}}
\author{Ming Li \qquad Zhentao Shi \qquad Yapeng Zheng\thanks{Li: Department of Economics and Risk Management Institute, National University of Singapore, \texttt{mli@nus.edu.sg}. Shi (corresponding author): Department of Economics, The Chinese University of Hong Kong, \texttt{zhentao.shi@cuhk.edu.hk}. Zheng: Department of Economics, The Chinese University of Hong Kong, \texttt{yapengzheng@link.cuhk.edu.hk}.}}

\setlength{\skip\footins}{28pt plus 4pt minus 2pt}
\maketitle

\thispagestyle{empty}


\begin{abstract}
We develop a unified estimation and inference framework for dyadic network formation with individual fixed effects, covering both transferable-utility (TU) and nontransferable-utility (NTU) links under general link functions. Under NTU, bilateral consent makes the fixed effects non-additive and the log-likelihood non-concave in the high-dimensional fixed effects, so differencing and profile-likelihood methods fail. We combine a joint method-of-moments initial estimator, a Le Cam one-step refinement, and a split-network jackknife bagging step that removes the incidental parameter bias without inflating variance. The resulting homophily estimator is asymptotically normal, unbiased, and attains the Cram\'er--Rao lower bound without requiring the log-likelihood to be concave in the fixed effects; we extend the theory to average partial effects and establish robustness to link-function misspecification. Simulations under both TU and NTU designs confirm these predictions. Applied to Thai village networks (TU), kinship and wealth differences both increase linking; in the Nyakatoke risk-sharing network (NTU), wealth differences have no significant effect, mirroring the two regimes' distinct logics.

\bigskip{}

\end{abstract}
\clearpage
\setlength{\skip\footins}{10.8pt plus 4pt minus 2pt}

\begin{bibunit}

\section{Introduction}\label{sec:introduction}

Dyadic network formation models, which describe how links arise between pairs of agents, are central to understanding social and economic phenomena such as friendship networks, risk-sharing arrangements, and inter-firm alliances. In such models, individual unobserved heterogeneity is
pervasive: agents differ in sociability, ability, or other latent traits
that simultaneously affect their desirability as partners and correlate
with observed covariates. Treating this heterogeneity as individual
fixed effects is therefore essential, yet it poses econometric
challenges, most notably the incidental parameter problem, which leads
to asymptotic bias and substantial computational burdens when the
number of fixed effects grows with the network size. Whether fixed-effects inference for dyadic network-formation models can be pushed to Cram\'er--Rao efficiency under \emph{general} link functions, particularly for bilateral-consent settings where fixed effects enter non-additively, has remained an open question, which this paper closes.

Whether a link involves transferable or nontransferable utilities has
direct consequences for econometric analysis. Under transferable
utilities (TU), the total surplus from a link is additively separable
in the two agents' fixed effects, and a range of methods exploit this
separability and specific distributional assumptions on the idiosyncratic shocks to eliminate or profile out the fixed effects
\citep{chatterjee2011random,graham2017econometric,
dzemski2019empirical}. A method that accommodates general link functions is still lacking in the literature. Furthermore, many economically
important networks involve nontransferable utilities (NTU):
friendships form when both individuals are willing, risk-sharing
arrangements require mutual consent, and bilateral trade links are
sustained only when each party individually benefits. Under NTU, each
agent's latent type enters the linking probability through a separate
agreement condition, so the fixed
effects are no longer additively separable and the Bernoulli log-likelihood is no longer concave in them. This non-additivity and non-concavity simultaneously break the differencing and profile-likelihood methods developed for TU, all of which rely on either additive bias structure or global concavity of the objective in the nuisance parameters. Despite the empirical prevalence of
NTU settings, no existing
framework delivers inference for homophily parameters or consistent
estimation of fixed effects under NTU with a general link function.
Bilateral consent without side payments characterizes a significant proportion of networks studied in applied research, including friendships, risk-sharing, collaboration, and information exchange. This gap is therefore consequential: applied researchers currently have no available procedure to compute standard errors and conduct hypothesis tests for precisely the class of networks they routinely encounter.

This paper proposes a unified parametric framework for estimation and
inference under both TU and NTU with general link functions, based on a single large network with
observed pairwise covariates and individual fixed effects.\footnote{We provide a more in-depth comparison between TU and NTU in
Remark~\ref{rem:tu_vs_ntu}.} Our estimation
strategy proceeds in three steps. First, we construct a joint method-of-moments (JMM) estimator for the homophily parameters $\b_0$ and
the fixed effects $\bm{\alpha}_0$, and establish consistency for both and normality for the homophily parameter. Second, we refine the JMM estimator of  $\beta_0$  via
\citet{le1969theorie}'s one-step approximation to the maximum
likelihood estimator (MLE), which requires merely a single Newton-type
update and avoids the computational burden and numerical instability
of full maximum likelihood optimization with high-dimensional fixed
effects. Third, we debias the one-step estimator, which inherits
the asymptotic bias of the MLE, through a split-network jackknife
combined with bootstrap aggregating (bagging)
\citep*{breiman1996bagging,hirano2017forecasting}, a technique that is novel in the network formation literature.

Our contributions are threefold. \emph{Theoretically}, we deliver the first asymptotic inference and efficiency theory for homophily parameters under NTU, with $\ell_{\infty}$-consistent estimation of the individual fixed effects; the bagging estimator $\widehat{\b}_{\text{BG}}$ is asymptotically normal, asymptotically unbiased, attains the Cram\'er--Rao lower bound \citep[CRLB,][]{rao1992information}, and is robust to the choice of random splits. \emph{Methodologically}, our approach sidesteps log-likelihood concavity altogether: JMM solves a degree-matching moment condition rather than maximizing the likelihood, while the one-step and bagging refinements use the outer-product Fisher information, which is positive definite by construction. Concavity fails structurally under NTU and under TU with non-log-concave shock densities; our framework is the first to deliver Cram\'er--Rao-efficient inference in both cases. We also establish asymptotic normality for the average partial effects (APEs), robustness to link-function misspecification, and a multi-network extension. \emph{Empirically}, we provide the first comparative evidence of how the TU vs.\ NTU distinction shapes estimated homophily: wealth differences drive linking under TU (Thai villages) but not under NTU (Nyakatoke).

Efficiency and asymptotic unbiasedness matter here because they determine whether applied work can distinguish economically meaningful effects from noise. An inefficient estimator yields unnecessarily wide confidence intervals; a biased estimator---even if $\sqrt{N}$-consistent---invalidates plug-in hypothesis tests. Our bagging estimator is simultaneously unbiased and CRLB-efficient, yielding valid confidence intervals as narrow as the model allows.

Simulation results confirm that the proposed estimators for the homophily
parameters, individual fixed effects, and APEs perform as predicted
by the theory under both the NTU baseline and a TU counterpart using either
logistic or probit link functions. In particular, the non-concavity of the log-likelihood
renders MLE unreliable for estimating the high-dimensional fixed effects
$\bm{\alpha}_{0}$: MLE produces substantially
larger root mean squared errors (RMSEs) than JMM, with the failure
most pronounced under bimodal or hub--periphery designs
(Figure~\ref{fig:alpha_rmse} in the Supplemental Material).
We present two empirical examples that exhibit the TU and NTU
regimes. First, under TU, we apply our method to the Townsend Thai
village networks \citep{kinnan2024propagation}, pooling across 16
villages to estimate link formation for financial, operations, and
labor transactions. We find that kinship is a strong and broadly
positive predictor across all three networks, and that greater
absolute net-worth differences between households are also associated
with more frequent linking, consistent with gains-from-trade in
transactional relationships under TU. Second,
under NTU, we revisit the Nyakatoke risk-sharing network
\citep*{deweerdt2004risk}, where our method indicates that wealth
differences have no statistically significant effect on link
formation. The contrast between these two settings is not incidental: TU surplus-sharing predicts that wealth heterogeneity \emph{fuels} linking in financial, operations, and labor transactions (as in the Thai estimates), while the bilateral-consent logic of risk pooling predicts that wealth heterogeneity does \emph{not} drive linking when both sides must consent (as in the Nyakatoke estimate). The same econometric framework recovers both patterns without imposing the separability that distinguishes them. A code-and-data demonstration
of our proposed methods is available at the GitHub repository \url{https://github.com/YapengZheng/bagging_network}.

\paragraph{Literature Review.} Our paper contributes to the literature
on dyadic network formation in a single large network. Most existing
work studies TU, which allows individual fixed effects to be eliminated
by arithmetic differencing (\citealp{graham2017econometric};
see \citealp{graham2020network}, for a review). \citet{gao2023logical}
study a semiparametric model under NTU using logical differencing, but without
inference for homophily parameters or estimators of fixed effects.
We complement their work by establishing inference for homophily parameters,
delivering $\ell_{\infty}$-consistent estimators of fixed effects,
and developing asymptotic results for the APEs.

Our paper also builds on \citet{graham2017econometric}, who introduces
a tetrad logit estimator and a joint MLE under TU and logistic
link functions. These methods, as well as functional differencing
\citep{bonhomme2012functional}, do not extend
to NTU because the fixed effects enter the linking probability
non-additively (as $p(\alpha_i,\alpha_j,x_{ij}^\top\beta)$ rather than $p(\alpha_i+\alpha_j+x_{ij}^\top\beta)$). Recent contributions
also remain within TU: \citet{hughes2026estimating} develops a jackknife
bias correction, \citet{Qu2025} propose a projection approach for
directed networks, and \citet*{gao2020nonparametric} studies semiparametric estimation. In
contrast, our estimator applies under both TU and NTU. An
earlier working paper by \citet{shi2016structural} analyzes the MLE for this model under NTU. 
The present paper subsumes and supersedes it by establishing formal efficiency via the Le Cam one-step approximation, 
developing a stable procedure that avoids the non-concavity of direct MLE (see Section~\ref{sec:model_and_computation} and Figure~\ref{fig:alpha_rmse} in the Supplemental Material), and introducing bagging for debiasing.

Methodologically, our work relates to the large-$T$ panel literature
on nonlinear fixed-effect models \citep{hahn2004jackknife,dhaene2015split,fernandez2016individual,fernandez2018fixed,honore2021identification}.
These methods rely on concavity of log-likelihood functions and/or
sparsity assumptions on certain derivatives of functionals of the fixed effects
that are hard to verify in our setting. Instead, we adapt the sample-splitting
idea (see \citealp{mei2026nickell,liao2024nickell} for tackling Nickell-type
biases in panel predictive regressions using split-sample strategies)
and establish formally that bagging delivers unbiased and efficient
estimation. Related work on orthogonalized
estimators \citep{bonhomme2024neyman} requires
an additive bias structure that fails under NTU. The Bernoulli log-likelihood is non-additive because the $\log(1-p_{ij})$ term involves the product $F(\alpha_i+x_{ij}^\top\beta)\,F(\alpha_j+x_{ij}^\top\beta)$, which violates the additive bias structure required by their Assumption~3.

Finally, there
is a line of work on strategic network formation and empirical games
based on pairwise stability \citep[e.g.,][]{jackson1996strategic,de2018identifying}.
These models incorporate externalities but typically impose restrictions
on heterogeneity or the degree distribution and often require TU;
\citet{GaoLiXu2026TIS} recently provide tractable identification in strategic models with unobserved heterogeneity under TU.
Our fixed-effects approach, which accommodates both TU and NTU and
permits arbitrary correlation between observables and the fixed effects,
is therefore complementary to and methodologically distinct from the
existing literature (for a review of the two approaches,
see \citealp{de2020econometric}).

\paragraph{Organization.} Section~\ref{sec:model_and_computation} states the model and estimation algorithm; Section~\ref{sec:estimation_and_inference} develops the asymptotic theory; Section~\ref{sec:extensions} covers APEs and link-function misspecification; Sections~\ref{sec:simulation} and~\ref{sec:empirical_applications} report simulations and empirical applications; Section~\ref{sec:conclusion} concludes. Appendices collect notation, supporting lemmas, and proofs; the Supplemental Material contains proofs of supporting lemmas, the multi-network extension, extended simulations, and data descriptions.

\paragraph{Notation.} Let ``$\coloneqq$'' denote a definition, and let the
superscript ``$\top$'' denote the transpose of a vector or a matrix.
We use boldface for variables of increasing dimension with $n$.
For example, the true fixed effects $\bm{\alpha}_{0}=\left(\a_{i0}\right)_{1\leq i\leq n}$
are $n\times1$. For an $n\times1$ vector $\mathbf{a}=(a_{1},\dots,a_{n})^{\top}$,
its $\ell_{1}$ norm is $\Vert\mathbf{a}\Vert_{1} \coloneqq \sum_{i=1}^{n}|a_{i}|$,
$\ell_{2}$ norm is $\Vert\mathbf{a}\Vert_{2} \coloneqq (\sum_{i=1}^{n}a_{i}^{2})^{1/2}$,
and $\ell_{\infty}$ norm is $\Vert\mathbf{a}\Vert_{\infty} \coloneqq \max_{1\leq i\leq n}|a_{i}|$.
When $O(\cdot)$ (and other notation for order) is written for a vector
(or matrix), it means that each element in the vector (or matrix)
is of the order in $O(\cdot)$. Here, ``$\text{plim}$'' denotes
the probability limit, ``$\stackrel{p}{\to}$'' convergence in probability,
and ``$\stackrel{d}{\to}$'' convergence in distribution. Unless
otherwise noted, for all convergence results we pass $n\to\infty$.
For an $n\times n$ matrix $\mathbf{A}$, we write $\Vert\mathbf{A}\Vert_{1} \coloneqq \max_{1\leq i\leq n}\Vert\mathbf{A}_{\cdot i}\Vert_{1}$,
$\Vert\mathbf{A}\Vert_{\infty} \coloneqq \max_{1\leq i\leq n}\Vert\mathbf{A}_{i\cdot}\Vert_{1}$
and $\Vert\mathbf{A}\Vert_{\max} \coloneqq \max_{1\leq i,j\leq n}|\mathbf{A}_{ij}|$,
where $\mathbf{A}_{\cdot i}$ and $\mathbf{A}_{i\cdot}$ are the $i$th
column and row of $\mathbf{A}$, respectively. To simplify notation, we write $p_{ij}(\bm{\alpha},\beta) \coloneqq p(\alpha_{i},\alpha_j,x_{ij}^{\top}\beta)$. The bare form $p_{ij}$ denotes this function with the arguments $(\bm{\alpha},\beta)$ suppressed; we write $p_{ij,0}\coloneqq p_{ij}(\bm{\alpha}_{0},\beta_{0})$ for its value at the true parameters. Evaluation at any other specific arguments is written out explicitly. The same convention applies to other objects that are functions of $(\bm{\alpha},\beta)$. Finally, the abbreviation ``w.p.a.1'' stands for ``with probability approaching one.''

\section{Model and Computation}\label{sec:model_and_computation}

We consider an undirected network formed among agents $i\in\mathcal{I}_{n} \coloneqq \{1,\dots,n\}$.
Hence, there are $N\coloneqq\binom{n}{2}$ dyads to be linked. 
An observed link $Y_{ij}$ between $i$ and $j$ is formed with probability:
\begin{equation}
\Pr(Y_{ij}=1|X_{ij},\alpha_{i0},\alpha_{j0}) \coloneqq p(\alpha_{i0},\alpha_{j0},X_{ij}^\top\beta_0)\text{ for }1\leq i\neq j\leq n,\label{eq:observed_link}
\end{equation}
where $p(\cdot):\R^3 \to (0,1)$ is a user-specified symmetric function in $\alpha_i$ and $\alpha_j$.
We rule out self-loops, i.e., $Y_{ii}=0,\ i\in\mathcal{I}_{n}$. The link probability $p_{ij,0}$ is jointly determined by the two agent-specific fixed effects $(\alpha_{i0},\alpha_{j0})$ and the dyad-specific
index $X_{ij}^{\top}\beta_{0}$ that captures the homophily effect
in the observable characteristics of each pair $(i,j)$, where $X_{ij}\in\mathbb{R}^{K}$
denotes the symmetric dyad-level covariates for all $i\neq j$. The log-likelihood
is 
\begin{align*}
\ell_{n}(\bm{\alpha},\beta) \coloneqq & \sum_{i=1}^{n}\sum_{j>i}\{Y_{ij}\log p_{ij}(\bm{\alpha},\beta)+(1-Y_{ij})\log\left(1-p_{ij}(\bm{\alpha},\beta)\right)\}.
\end{align*}

\begin{rem}
\label{rem:tu_vs_ntu}Our network formation model (\ref{eq:observed_link})
covers both TU and NTU:
\begin{align}
Y_{ij}= & \,\ind\left\{ \alpha_{i0}+\alpha_{j0}+X_{ij}^{\top}\beta_{0}-\epsilon_{ij}>0\right\} \text{ and}\tag{TU}\\
Y_{ij}= & \,\ind\left\{ \alpha_{i0}+X_{ij}^{\top}\beta_{0}-\epsilon_{ij}>0\right\} \ind\left\{ \alpha_{j0}+X_{ji}^{\top}\beta_{0}-\epsilon_{ji}>0\right\},\tag{NTU}
\end{align}
where $\epsilon_{ij}$ is an idiosyncratic error with a known distribution.
The model (TU) essentially asserts that, if the joint
 surplus generated by a bilateral link $\alpha_{i0}+\alpha_{j0}+X_{ij}^{\top}\beta_{0}-\epsilon_{ij}$
is positive, then the link between $i$ and $j$ is formed. An important
assumption behind the model (TU) is that the link
surplus can be freely distributed between $i$ and $j$, and that
bargaining efficiency is always achieved, which is a strong assumption in many
networks (e.g., risk-sharing networks and friendship networks). The
model (NTU), on the other hand, requires that
the utility surplus from the link for both $i$ and $j$ be strictly
positive in order to form a link, which is arguably more realistic
in the aforementioned networks. Furthermore, the model (NTU)
reflects the fact that the party with relatively lower utility is
the pivotal one in link formation. Finally, it can be shown (see \citealp{gao2023logical})
that the model (NTU) can accommodate homophily
effects in both observable and unobservable covariates.
\end{rem}
Given the model, we introduce the algorithm to estimate the homophily
coefficient $\b_{0}$. There are three sequential modules---JMM, one-step (OS), and Bagging (BG)---that lead
to $\widehat{\b}_{\text{BG}}$. Specifically, Module JMM provides
an initial consistent estimator, which, however, does not reach the
CRLB and is biased. We refine the JMM estimator with the one-step
adjustment to achieve the CRLB. Finally, we apply the bagged split-network
jackknife to debias the one-step estimator while preserving its efficiency.

\medskip{}

We define a few objects before each module. Let ${\bf Y}=\left(Y_{ij}\right)_{1\leq i,j\leq n}$
be the $n\times n$ adjacency matrix and ${\bf X}=\left(X_{ij}\right)_{1\leq i,j\leq n}$
be the $n\times n\times K$ random tensor of covariates. Denote their
realizations by $\mathbf{y}=(y_{ij})_{1\leq i,j\leq n}$ and ${\bf x}=\left(x_{ij}\right)_{1\leq i,j\leq n}$,
respectively. The \emph{degree} $d_{i} \coloneqq \sum_{j\neq i}y_{ij}$ is
defined for each $i\in\mathcal{I}_{n}$ of the observed network $\mathbf{Y}$.
Define a vector of moment functions $\mathbf{m}(\bm{\alpha},\beta) \coloneqq (\mathbf{m}_{1}^{\top}(\bm{\alpha},\beta),m_{2}^{\top}(\bm{\alpha},\beta))^{\top}$,
where $\mathbf{m}_{1}(\bm{\alpha},\beta) \coloneqq (d_{i}-\sum_{j\neq i}p_{ij}(\bm{\alpha},\beta))_{i=1}^{n}$
is an $n$-dimensional function that concerns the average degree of
each $i$, and $m_{2}(\bm{\alpha},\beta) \coloneqq \sum\limits_{i=1}^{n}\sum\limits_{j>i}[y_{ij}-p_{ij}(\bm{\alpha},\beta)]x_{ij}$
is a $K$-dimensional function.

\begin{Module}[JMM]\label{alg:mod_jmm} The \textbf{JMM estimator}
$(\widehat{\bm{\alpha}},\widehat{\beta})$ is the solution
to the $(n+K)$-equation system $\mathbf{m}(\bm{\alpha},\beta)=0$.

\end{Module}
The JMM estimator is just-identified with $n+K$ moment conditions for $n+K$ unknowns, and we choose degree-based equations in $\mathbf{m}_{1}$ for their simplicity and stability \citep{graham2017econometric}.
To find the solution to $\mathbf{m}(\bm{\alpha},\beta)=0$, for each
$\b$ we let 
\begin{equation}
r_{i}(\bm{\alpha},\beta)=\alpha_{i}+(n-1)^{-1}\bigg(d_{i}-\sum_{j\neq i}p_{ij}(\bm{\alpha},\beta)\bigg),\qquad i\in\mathcal{I}_{n}\label{eq:iter}
\end{equation}
and $\mathbf{r}(\bm{\alpha},\beta)=\big(r_{1}(\bm{\alpha},\beta),\dots,r_{n}(\bm{\alpha},\beta)\big)^{\top}$.
The intuition is, for any $i$ when $d_{i}$ is strictly larger than
$\sum_{j\neq i}p_{ij}(\bm{\alpha},\beta)$, we would like to increase
$\a_{i}$ such that each $p_{ij}(\bm{\alpha},\beta)$ for $j\neq i$
is larger, and vice versa. Starting with an initial value $\bm{\alpha}^{0}$,
we iterate $\bm{\alpha}^{k+1}(\beta)=\mathbf{r}(\bm{\alpha}^{k}(\beta),\beta)$
until convergence to obtain $\widehat{\bm{\alpha}}\left(\b\right)$,
and then we solve the finite dimensional equations $m_{2}(\widehat{\bm{\alpha}}\left(\b\right),\beta)=0$.
Unlike existing methods (e.g., Theorem 1.5 of \citealp{chatterjee2011random}; equation (17) of \citealp{graham2017econometric}), we do not rely on a specific functional-form assumption on the idiosyncratic shocks (e.g., logit or probit): $\widehat{\bm\alpha}(\beta)$ is recovered purely through a contraction-type fixed-point argument on (\ref{eq:iter}).

\medskip{}

The OS module involves the score and information matrix. Define
the score of $\ell_{n}$ as $\mathbf{s}(\bm{\alpha},\beta)=\big(\mathbf{s}_{1}^{\top}(\bm{\alpha},\beta),s_{2}^{\top}(\bm{\alpha},\beta)\big)^{\top}=\left(\partial\ell_{n}/\partial\bm{\alpha}^{\top},\partial\ell_{n}/\partial\beta^{\top}\right)^{\top}$,
and partition the information matrix 
\begin{equation}
\mathbf{I}(\bm{\alpha},\beta)=\mathbb{E}[\mathbf{s}(\bm{\alpha},\beta)\mathbf{s}(\bm{\alpha},\beta)^{\top}|\mathbf{x},\bm{\alpha}]\eqqcolon\begin{pmatrix}\mathbf{I}_{11}(\bm{\alpha},\beta) & \mathbf{I}_{12}(\bm{\alpha},\beta)\\
\mathbf{I}_{12}^{\top}(\bm{\alpha},\beta) & \mathrm{I}_{22}(\bm{\alpha},\beta)
\end{pmatrix}\label{eq:info_mat}
\end{equation}
into four compatible blocks. Define the concentrated score function
and information matrix of $\beta$ as 
\begin{align}
s_{n}(\bm{\alpha},\beta) & =s_{2}(\bm{\alpha},\beta)-\mathbf{I}_{12}(\bm{\alpha},\beta)^{\top}\mathbf{I}_{11}(\bm{\alpha},\beta)^{-1}\mathbf{s}_{1}(\bm{\alpha},\beta)\text{ and}\nonumber\\
\mathrm{I}_{n}(\bm{\alpha},\beta) & =\mathrm{I}_{22}(\bm{\alpha},\beta)-\mathbf{I}_{12}(\bm{\alpha},\beta)^{\top}\mathbf{I}_{11}(\bm{\alpha},\beta)^{-1}\mathbf{I}_{12}(\bm{\alpha},\beta),\label{eq:concentrated_info}
\end{align}
respectively.
Here, $s_n(\bm{\alpha},\beta)$ and $\mathrm{I}_n(\bm{\alpha},\beta)$ are the concentrated counterparts of $s_{2}(\bm\alpha,\beta)$ and $\mathrm{I}_{22}(\bm\alpha,\beta)$
after concentrating out the influence from estimating $\bm{\alpha}$, respectively.

\begin{Module}[OS]\label{alg:mod_os} Substitute the JMM estimator
$(\widehat{\bm{\alpha}},\hat{\beta})$ into 
\begin{equation}
\hat{\beta}_{\mathrm{OS}} \coloneqq \hat{\beta}+\mathrm{I}_{n}(\widehat{\bm{\alpha}},\hat{\beta})^{-1}s_{n}(\widehat{\bm{\alpha}},\hat{\beta}).\label{eq:beta_os}
\end{equation}
\end{Module}

Module OS is \citet{le1969theorie}'s one-step approximation of the MLE: starting from the $\sqrt{N}$-consistent JMM estimate, a single Newton update along the efficient-score direction closes the efficiency gap and delivers the CRLB asymptotically. In (\ref{eq:beta_os}), $s_n$ is the slope of the log-likelihood in $\beta$ after profiling out $\bm\alpha$, and $\mathrm{I}_n^{-1}$ rescales the step by the inverse curvature, so the adjustment is larger when the likelihood is flat in $\beta$ and smaller when it is sharply curved. Module OS sidesteps the non-concavity and numerical instability of direct MLE optimization (see Remark~\ref{rem:os_no_concavity}).

\medskip{}

Finally, the bagging step introduces randomization. Assume an even integer
$n$ for convenience. Let $t=1,2,\ldots,\tilde{T}_{n}$, for some
$\tilde{T}_{n}\leq\binom{n}{n/2}$, index an equal-sized random partition
of $\mathcal{I}_{n}$ into $\mathcal{I}_{1,n}^{(t)}$ and $\mathcal{I}_{2,n}^{(t)}$
such that $\mathcal{I}_{1,n}^{(t)}\cup\mathcal{I}_{2,n}^{(t)}=\mathcal{I}_{n}$,
$\mathcal{I}_{1,n}^{(t)}\cap\mathcal{I}_{2,n}^{(t)}=\emptyset$,
and the splits are independent over $t$.

\begin{Module}[BG]\label{alg:mod_bg} For each $t=1,\dots,\tilde{T}_{n}$,
hold the full-sample JMM estimator $\hat{\beta}$ fixed and, on the
subnetwork indexed by $\mathcal{I}_{1,n}^{(t)}$, update only $\widehat{\bm{\alpha}}$
via the fixed-point iteration in (\ref{eq:iter}). Apply the OS step
on that split to obtain $\hat{\beta}_{\mathrm{OS,1}}^{(t)}$ using the new $\widehat{\bm\alpha}$ and full-sample $\widehat{\beta}$. Repeat
the same procedure on $\mathcal{I}_{2,n}^{(t)}$ to obtain $\hat{\beta}_{\mathrm{OS,2}}^{(t)}$.
Apply the bagged jackknife to obtain the \textbf{BG estimator} 
\begin{equation*}
\hat{\beta}_{\mathrm{BG}} \coloneqq 2\hat{\beta}_{\mathrm{OS}}-(2\tilde{T}_{n})^{-1}\sum_{t=1}^{\tilde{T}_{n}}(\hat{\beta}_{\mathrm{OS,1}}^{(t)}+\hat{\beta}_{\mathrm{OS,2}}^{(t)}).
\end{equation*}
\end{Module}

We employ split-network jackknife to debias $\hat{\beta}_{\mathrm{OS}}$.
Due to the equal splits of nodes, each of $\hat{\beta}_{\mathrm{OS,1}}^{(t)}$
and $\hat{\beta}_{\mathrm{OS,2}}^{(t)}$ incurs twice the leading
bias in the asymptotic expansion. If we apply the split-network jackknife
only once, then 
\[
\hat{\beta}_{\mathrm{OS-SJ}}^{\left(t\right)} \coloneqq 2\hat{\beta}_{\mathrm{OS}}-\dfrac{1}{2}\left(\hat{\beta}_{\mathrm{OS,1}}^{(t)}+\hat{\beta}_{\mathrm{OS,2}}^{(t)}\right)
\]
self-cancels the leading bias. However, the variance of $\hat{\beta}_{\mathrm{OS-SJ}}^{\left(t\right)}$
is doubled, since splitting the network in half causes the links between
nodes belonging to different subnetworks to be ignored. Furthermore,
splitting the whole network randomly makes the estimator computationally
unstable. To deal with these issues, we let $\tilde{T}_{n}\to\infty$
and indeed $\hat{\beta}_{\mathrm{BG}}=\tilde{T}_{n}^{-1}\sum_{t=1}^{\tilde{T}_{n}}\hat{\beta}_{\mathrm{OS-SJ}}^{\left(t\right)}$
averages $\hat{\beta}_{\mathrm{OS-SJ}}^{\left(t\right)}$ over $\tilde{T}_{n}$
independent splits.

\begin{rem}
It is asymptotically valid if we split the nodes into $G \geq 2$ equal parts and form a jackknife-type correction with appropriate weights, as discussed in \cite{dhaene2015split} for panel data. However, increasing $G$ reduces the number of nodes in each subsample, which degrades the quality of fixed-effect estimates and can cause numerical instability. Our approach instead fixes $G=2$, the minimal split needed for first-order bias correction, and aggregates over $\tilde{T}_n \to \infty$ independent random splits. This bagging strategy is crucial: a single split-network jackknife eliminates the leading bias but doubles the asymptotic variance, whereas averaging over many splits deflates the variance back to the Cram\'er--Rao lower bound (Theorem~\ref{thm:bagging_normality}).
\end{rem}

When computing $\hat{\beta}_{\mathrm{OS,1}}^{(t)}$ and $\hat{\beta}_{\mathrm{OS,2}}^{(t)}$
for each random split $t$, we do not recompute a split-specific initial
JMM estimate of $\beta$. Instead, the full-sample JMM estimator $\widehat{\b}$
is retained, $\widehat{\bm{\alpha}}$ is re-estimated on each split
via (\ref{eq:iter}), and the OS update then produces the split-specific
estimators $\hat{\beta}_{\mathrm{OS,1}}^{(t)}$ and $\hat{\beta}_{\mathrm{OS,2}}^{(t)}$.
The rationale is that the incidental parameter bias of the OS estimator
comes from the estimation of high-dimensional fixed effects: $\hat{\alpha}_i$
converges at rate $\sqrt{(\log n)/n}$, which is slower than $\sqrt{N}\asymp n$
for $\hat{\beta}$, and it is this slower rate that generates the
$O(N^{-1/2})$ bias requiring correction.
Consequently, the procedure remains
computationally efficient for moderate values of $\tilde{T}_{n}$,
such as 200 or 400 in our simulations.

\begin{figure}[htbp]
\centering
\includegraphics[width=\textwidth,height=\textheight,keepaspectratio]{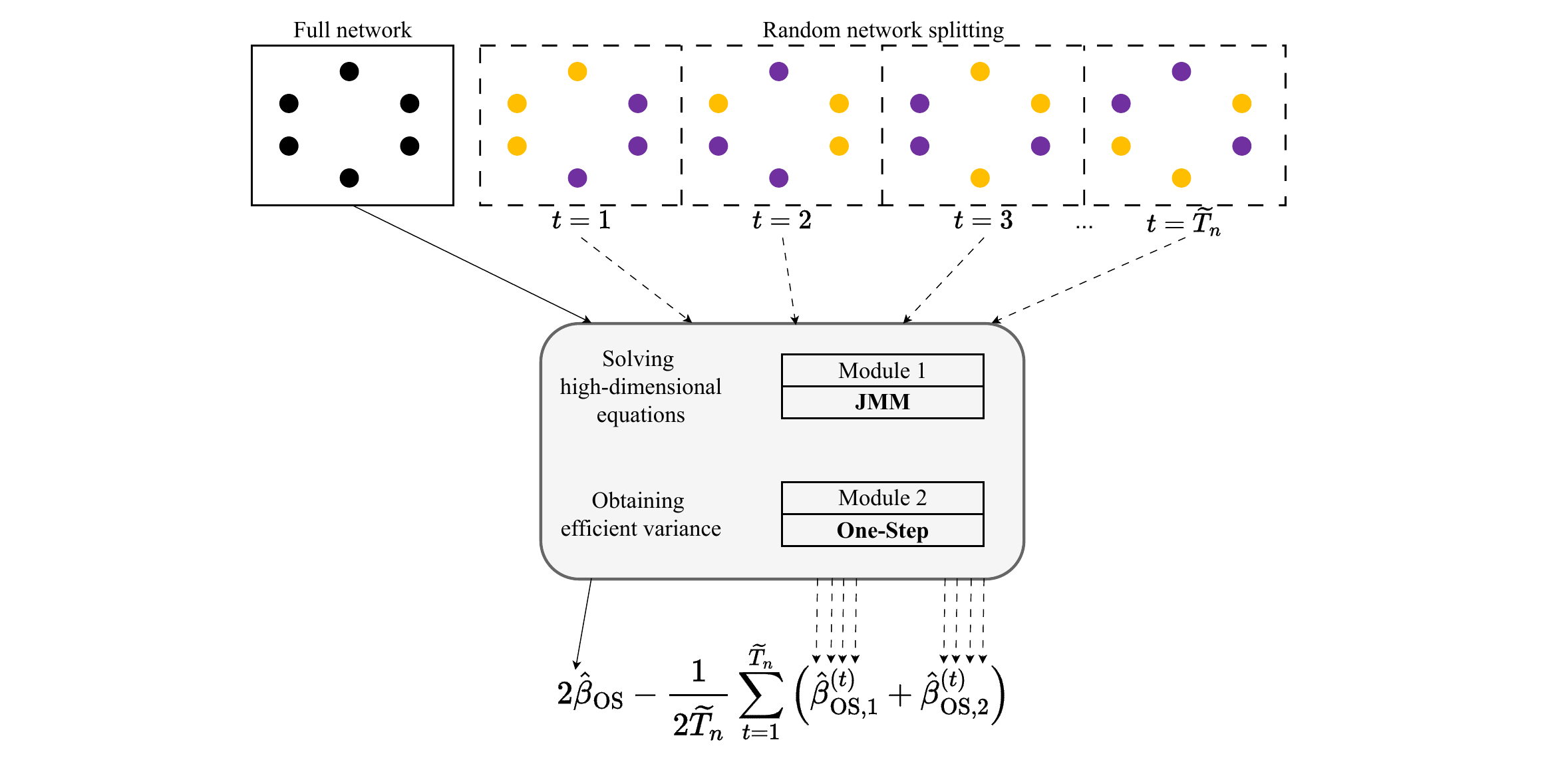}
\caption{Flowchart of the estimation procedure}
\label{fig:summary_of_bagging}
\end{figure}

Figure~\ref{fig:summary_of_bagging} summarizes the procedure. Modules JMM and OS are first applied to the full network (black dots) to obtain $\widehat{\b}_{\text{OS}}$. For each random split $t=1,\dots,\tilde{T}_{n}$, the nodes are divided into two halves (e.g., yellow and purple dots in split $t=1$). On each half, JMM re-estimates $\bm{\alpha}$ using the full-sample $\widehat{\b}$, and the OS step updates $\hat{\beta}$ to yield $\hat{\beta}_{\mathrm{OS,1}}^{(t)}$ and $\hat{\beta}_{\mathrm{OS,2}}^{(t)}$. Averaging the bagged jackknife over all splits produces $\hat{\beta}_{\mathrm{BG}}$.

\section{Large Sample Theory}\label{sec:estimation_and_inference}

This section proves that the bagging estimator $\widehat{\b}_{\mathrm{BG}}$ is asymptotically normal, unbiased, and CRLB-efficient. We reach this result in three steps. Section~\ref{subsec:jmm_estimator} shows that the JMM estimator is $\sqrt{N}$-consistent for $\b_{0}$ and that the plug-in $\widehat{\bm\alpha}(\widehat{\b})$ converges to $\bm\alpha_{0}$ in the $\ell_{\infty}$ norm, enough to enable the one-step correction. Section~\ref{subsec:one_step_estimator} shows that the one-step estimator $\widehat{\b}_{\mathrm{OS}}$ attains the CRLB but carries an $O(n^{-1})$ incidental-parameter bias. Section~\ref{subsec:bagging} then shows that bagging removes this bias without inflating the asymptotic variance.

\bigskip

We first state four baseline assumptions that underlie the theoretical
results.
\begin{assumption}[Correctly Specified Model]
\label{assu:model}The conditional likelihood of $\mathbf{Y}=\mathbf{y}$
given $\mathbf{X}=\mathbf{x}$ and $\bm{\alpha}=\bm{\alpha}_{0}$
is
\[
\Pr(\mathbf{Y}=\mathbf{y}|{\bf X}=\mathbf{x},\bm{\alpha}=\bm{\alpha}_{0})=\prod_{i=1}^{n}\prod_{j>i}p_{ij,0}^{y_{ij}}\left(1-p_{ij,0}\right)^{1-y_{ij}}.
\]
\end{assumption}
Assumption~\ref{assu:model} is similar to Assumption 1 of \citet{graham2017econometric},
except for two important differences. First, under potentially NTU $\alpha_{i0}$
and $\alpha_{j0}$ are not additively separable in the linking probability
between $i$ and $j$, and thus the tetrad logit estimator of \citet{graham2017econometric}
does not apply in our setting. Second, the functional form of $p\left(\cd\right)$
is general (subject to Assumption~\ref{assu:ass_f} below)
and includes the commonly used logistic (e.g., \citealp{chatterjee2011random,graham2017econometric,Qu2025})
and probit as special cases. The multiplicative form of the joint
likelihood implies that the links are formed independently of one another
conditional on the agent attributes. It is suitable in settings such
as risk-sharing networks, online friendships, trade networks, and
conflicts between nation-states. Conversely, settings with strategic link formation, where agents' linking decisions depend on other links in the network, fall outside this framework.

Assumption~\ref{assu:model} also requires the link function $p\left(\cd\right)$
to be correctly specified. It is well-known that, under regularity
conditions, the MLE converges to the parameter value that
minimizes the Kullback-Leibler divergence between the true and the
misspecified model. The issue is complicated
by the high-dimensional individual fixed effects and has not been investigated in the network formation literature. In the Supplemental Material, we discuss the impact of link function misspecification
on the theoretical results in Section~\ref{sec:sm_misspecification} 
and provide supporting simulation evidence in Section~\ref{sec:sm_simulations}.

\begin{assumption}[Bounded Support]
\label{assu:compact_support_and_sampling} $\mbox{}$

\begin{enumerate}[label=(\roman*), font=\upshape]
\item \label{assu:bounded_support:i} $\bm{\alpha}_{0}$ lies in the interior of a compact set $\mathbb{A}\subset\mathbb{R}^{n}$.
\item \label{assu:bounded_support:ii} $\beta_{0}$ lies in the interior of a compact set $\mathbb{B}\subset\mathbb{R}^{K}$.
\item \label{assu:bounded_support:iii} $X_{ij}$ satisfy $X_{ij}\in\mathbb{X}\subset\mathbb{R}^{K}$ for some compact set $\mathbb{X}$.
\end{enumerate}
\end{assumption}
Assumption~\ref{assu:compact_support_and_sampling}
collects and combines \citet{graham2017econometric}'s Assumptions
2 and 5(i). Together with Assumption~\ref{assu:ass_f} below, it implies that the
probability of a link forming between dyad $(i,j)$ is uniformly
bounded within $[\kappa,1-\kappa]$ for some $\kappa\in(0,1/2)$,
which requires the network to be dense.\footnote{Density of an undirected network is defined as $\rho_{n}=N^{-1}\sum_{i=1}^{n}\sum_{j>i}y_{ij}$. A network is dense if $\lim_{n\to\infty}\rho_{n}\in\left[c_{1},c_{2}\right]$
for some constant $0<c_{1}\leq c_{2}<1$.} The dense network makes it possible to estimate $\a_{i0}$ consistently
for each $i$.

Note that our theory in principle can allow the support of $X_{ij}$
to be unbounded; however, it would add little theoretical insight
but incur more technical complexity in the rates of convergence via
the Bernstein inequalities to bound the tail probabilities of random
variables. Assumption~\ref{assu:compact_support_and_sampling}\ref{assu:bounded_support:iii},
which is similar to \citet[ Assumption~2(ii)]{graham2017econometric},
allows us to focus on the main idea.

\begin{assumption}[Random Sampling]\label{assu:iid_sampling}
The sequence $\{(\alpha_{i0}, X_i)\}_{i=1}^n$ is an i.i.d.~sample
from a Borel probability measure $\nu$ on $\mathbb{A}\times\mathbb{X}_0$,
where $\mathbb{X}_0$ is a compact subset of a finite-dimensional
Euclidean space. The dyadic covariate satisfies $X_{ij}=h(X_i,X_j)$
for a fixed measurable symmetric function
$h:\mathbb{X}_0\times\mathbb{X}_0\to\mathbb{X}$, where $\mathbb{X}$
is the compact set in
Assumption~\ref{assu:compact_support_and_sampling}\ref{assu:bounded_support:iii}.
\end{assumption}
Assumption~\ref{assu:iid_sampling} parallels
\citet[Assumption~3]{graham2017econometric}: an agent is an i.i.d.~draw
from a population with unrestricted joint distribution over
$(\alpha_{i0}, X_i)$, accommodating arbitrary correlation between
unobserved heterogeneity and observed covariates. This is a
``fixed-effects'' treatment in the sense of \citet{chamberlain1984panel},
and provides the regularity needed for the limiting objects in
Lemma~\ref{lem:limits_exist} of Appendix~\ref{sec:supporting_lemmas}
to be well-defined.

For $r_{1},r_{2},r_{3}\in\mathbb{N}\cup\{0\}$, define
\[
p^{(r_{1},r_{2},r_{3})}(\alpha_{i},\alpha_{j},t)\coloneqq\frac{\partial^{r_{1}+r_{2}+r_{3}}p(\alpha_{i},\alpha_{j},t)}{\partial \alpha_{i}^{r_{1}}\partial \alpha_{j}^{r_{2}}\partial t^{r_{3}}}.
\]
Analogously to $p_{ij}(\bm{\alpha},\beta)$, we write $p_{ij}^{(r_{1},r_{2},r_{3})}(\bm{\alpha},\beta)\coloneqq p^{(r_{1},r_{2},r_{3})}(\alpha_{i},\alpha_{j},x_{ij}^{\top}\beta)$.
\begin{assumption}[Restrictions on $p(\cdot)$]
\label{assu:ass_f} The link-probability function $p(\alpha_{i},\alpha_{j},t)$ is symmetric in $(\alpha_i,\alpha_j)$ and three times continuously differentiable
in $(\alpha_{i},\alpha_{j},t)$. There exist constants $c_{1}\in(0,1/2]$
and $c_{2},c_{3}>0$ such that, for all $(\bm{\alpha},\beta)\in\mathbb{A}\times\mathbb{B}$,
$x_{ij}\in\mathbb{X}$, and $1\leq i\neq j\leq n$:
\[
\begin{aligned}
& p_{ij}\in[c_{1},1-c_{1}],\quad
p_{ij}^{(1,0,0)},\ p_{ij}^{(0,1,0)}\in[c_{2},1-c_{2}],\text{ and}\\
& \left|p_{ij}^{(r_{1},r_{2},r_{3})}\right|\leq c_{3}\text{ for any }1\leq r_{1}+r_{2}+r_{3}\leq 3.
\end{aligned}
\]
\end{assumption}
Assumption~\ref{assu:ass_f} lower-bounds the link probability and the
first derivatives with respect to the individual effects, while uniformly
upper-bounding all derivatives up to total order three.
In conjunction
with Assumption~\ref{assu:compact_support_and_sampling},
it is satisfied by common TU/NTU specifications with smooth link functions
(e.g., the cumulative distribution function of the logistic or the normal). This assumption is similar to \citet[Assumption 4.3(v)]{fernandez2016individual},
which regulates the smoothness of the likelihood functions.\footnote{The $p_{ij}\in[c_1,1-c_1]$ bound rules out, in population, nodes with zero or unit linking probabilities. In finite samples, a node $i$ may have degree $d_i=0$ (isolated) or $d_i=n-1$ (fully connected), causing the fixed-point iteration~(\ref{eq:iter}) to push $\hat{\alpha}_i$ to $\pm\infty$. We recommend removing such nodes before estimation. Under the dense-network assumption, the probability of any node being isolated or fully connected vanishes exponentially in $n$, so this trimming does not alter the asymptotic theory.}

\medskip{}

\begin{rem}[Role of symmetry in the proofs]
\label{rem:symmetry_proofs}
Our theoretical results are derived under the general symmetric link
function $p(\alpha_i,\alpha_j,x_{ij}^{\top}\beta)$ in
Assumption~\ref{assu:ass_f}, without distinguishing between TU and
NTU. The key structural property exploited throughout the proofs is
symmetry in the first two arguments (i.e.,
$p(\alpha_i,\alpha_j,t)=p(\alpha_j,\alpha_i,t)$), which holds under
both TU (where $p$ depends on $\alpha_i+\alpha_j$) and NTU (where
bilateral consent ensures the linking probability is symmetric). Thus, we do not rely on TU nor a specific link function to prove the theory.
\end{rem}

\subsection{JMM}\label{subsec:jmm_estimator}

Recall that the JMM module gives $(\widehat{\bm{\alpha}},\widehat{\b})$
to start with. The next lemma concerns the existence and uniqueness
of $\widehat{\bm{\alpha}}(\beta)$, as well as the convergence of
$\bm{\alpha}^{k}(\beta)$ to $\widehat{\bm{\alpha}}(\beta)$ via (\ref{eq:iter}).
\begin{lem}
\label{lem:alpha_conv_rate} If Assumptions \ref{assu:model}, \ref{assu:compact_support_and_sampling}, and \ref{assu:ass_f}
hold, then there exists a unique $\widehat{\bm{\alpha}}(\beta)$ w.p.a.1
for each $\beta\in\left\{ \beta\in\mathbb{B}|\left\Vert \beta-\beta_{0}\right\Vert _{2}<c\right\} $
in a neighborhood around $\beta_{0}$, where $c>0$ is small but fixed.
Moreover, uniformly across all $k$, we have 
\[
\Vert\bm{\alpha}^{k+2}(\beta)-\widehat{\bm{\alpha}}(\beta)\Vert_{1} \leq\delta\Vert\bm{\alpha}^{k}(\beta)-\widehat{\bm{\alpha}}(\beta)\Vert_{1},
\]
for some fixed constant $\delta\in(0,1)$.
\end{lem}
Lemma~\ref{lem:alpha_conv_rate} guarantees that
$\widehat{\bm{\alpha}}(\beta)=\lim\limits_{k\to\infty}\bm{\alpha}^{k}(\beta)$
and that the $\ell_{1}$-distance between $\widehat{\bm{\alpha}}(\beta)$
and $\bm{\alpha}^{k}(\beta)$ decreases geometrically after every
two iterations. Computing $\widehat{\bm{\alpha}}(\beta)$ is fast
in the simulations, which is another advantage of our iterative algorithm.
It is worth mentioning that we deviate from the existing methods (e.g.,
Theorem 1.5 of \citealp{chatterjee2011random} or the fixed point
equation (17) of \citealp{graham2017econometric}) in this step by
not requiring $p_{ij}$ to be logistic or
the link formation process to be TU. Instead, we use a gradient-descent-type
iterative algorithm (\ref{eq:iter}) to compute $\widehat{\bm{\alpha}}\left(\b\right)$
as a function of $\b$ and show that it is a contraction mapping.
As a result, it can accommodate general non-logistic link functions
and NTU.

Although $\widehat{\bm{\alpha}}(\beta)$ is unique by Lemma~\ref{lem:alpha_conv_rate}
for any $\b$ that lies within a distance of $c$ of $\b_{0}$, in
principle there could be multiple solutions to $m_{2}(\widehat{\bm{\alpha}}(\beta),\beta)=0$.
The next identification condition guarantees that any such $\widehat{\beta}$
is consistent for $\beta_{0}$. To state the assumption, we define
the concentrated moment equation for $\b$ as 
\begin{equation}
\label{eq:def_S_n_bar}
\bar{S}_{n}(\beta) \coloneqq N^{-1}\mathbb{E}[m_{2}(\bm{\alpha}(\beta),\beta)|\mathbf{x},\bm{\alpha}_{0}],
\end{equation}
where $\bm{\alpha}(\beta)$ is the unique solution to $\mathbb{E}[\mathbf{m}_{1}(\bm{\alpha},\beta)|\mathbf{x},\bm{\alpha}_{0}]=\bm{0}_{n}$,
a result from the proof of Lemma~\ref{lem:alpha_conv_rate}.
\begin{assumption}[Identification of $\b_{0}$]
\label{assu:mm_iden}Suppose for all $\delta>0$ and for $n$ large
enough,
\[
\inf_{\beta\in\mathbb{B}:\left\Vert \beta-\beta_{0}\right\Vert _{2}\geq\delta}\left\Vert \bar{S}_{n}(\beta)\right\Vert _{2}>0\text{ and }(\partial \bar{S}_n(\beta)/\partial\beta^\top)|_{\beta=\beta_0}\text{ has full rank.}
\]
\end{assumption}
Assumption~\ref{assu:mm_iden} identifies the low-dimensional parameter
$\b_{0}$, as discussed in \citet*{chen2014local} for nonlinear models
with high-dimensional nuisance parameters. The first part is the
standard ``unique minimizer'' condition that makes $\beta_{0}$ the
unique zero of $\bar{S}_{n}$ \citep[Page 45]{van2000asymptotic}; the
second part is a local rank condition---equivalently, the concentrated
Jacobian $\mathrm J_{22}-\mathbf J_{21}\mathbf J_{11}^{-1}\mathbf J_{12}$
is nonsingular at the truth. Substantively, Assumption~\ref{assu:mm_iden}
rules out cancellations in which the effect of $\beta$ on the
$x$-weighted link residuals can be exactly offset by adjusting the node
fixed effects while preserving all degree moments.

Assumption~\ref{assu:mm_iden} is the dyadic-network analogue of the
within-variation condition in panel models, but residualization is
against \emph{additive node effects} rather than within-individual means.
Let $e_{i}\in\mathbb{R}^{n}$ denote the $i$th standard basis vector,
$D$ the $N\times n$ dyad-incidence matrix whose row corresponding to
dyad $(i,j)$ equals $e_{i}+e_{j}$, $X$ the $N\times K$ matrix stacking
$x_{ij}^{\top}$ as rows, and $W_{0}=\mathrm{diag}\{p_{ij,0}(1-p_{ij,0})\}$.
Under the additive TU-logit specification $p_{ij}=\Lambda(\alpha_{i}+\alpha_{j}+x_{ij}^{\top}\beta)$,
differentiating the profiled moment at $\beta_{0}$ gives
\[
\left.\frac{\partial \bar{S}_{n}(\beta)}{\partial \beta^{\top}}\right|_{\beta_{0}}
=-N^{-1}(M_{D}X)^{\top} W_{0}\,(M_{D}X),
\]
where $M_{D}=I-D(D^{\top}W_{0}D)^{-1}D^{\top}W_{0}$, so the rank condition holds iff $X$ has nontrivial component orthogonal
to the additive-node subspace $\mathrm{col}(D)$. Non-additive covariates
such as $\lvert x_{i}-x_{j}\rvert$ or $x_{i}\cdot x_{j}$ pass this test
trivially; $x_{ij}=x_{i}+x_{j}$ is the canonical violation, since
$\beta$ then absorbs into relabeled fixed effects
$\tilde{\alpha}_{i}\coloneqq\alpha_{i}+x_{i}\beta$. The global
single-zero requirement is automatically satisfied under TU-logit from strict
concavity of the population log-likelihood. Under NTU, where this
concavity is unavailable, it reduces to parametric identifiability of the map
$(\bm\alpha,\beta)\mapsto\{F(\alpha_{i}+x_{ij}\beta)F(\alpha_{j}+x_{ij}\beta)\}_{i<j}$,
which holds under the same non-additivity. As a concrete check,
in a two-group NTU-logit model with $x_{ij}=1$ for cross-group dyads
and $x_{ij}=0$ within-group, the within-group subgraph pins down
$\bm\alpha_{0}$ via $\Lambda(\alpha_{i})\Lambda(\alpha_{j})$ and any cross-group
dyad then identifies $\beta_{0}$ by strict monotonicity of $\Lambda$. 

In the next theorem, we prove that $\hat{\beta}$ is consistent for
$\beta_{0}$ and that $\widehat{\bm{\alpha}}$ is uniformly consistent
for $\bm{\alpha}_{0}$ in the sup norm. Furthermore, we establish
asymptotic normality for $\hat{\beta}$. To state the result, define the concentrated Jacobian $\mathrm{J}_{n}(\bm{\alpha},\beta)$ in (\ref{eq:def_J_n}) and its probability limit $\mathrm{J}_{0} \coloneqq \mathrm{plim}_{n\to\infty}N^{-1}\mathrm{J}_{n,0}$. The bias vector $B_{0}=(B_{10},\dots,B_{K0})^{\top}$ is given in~(\ref{eq:bias_B0}), and the asymptotic sandwich variance $\Omega_{0}$ is given in~(\ref{limit_var_mm}); both are defined in Appendix~\ref{sec:appendix_defs}. Lemma~\ref{lem:limits_exist} in Appendix~\ref{sec:supporting_lemmas} establishes the existence and positive-definiteness of these limits.

\begin{thm}
\label{thm:jmm_alpha_beta_consistency}If
Assumptions \ref{assu:model}--\ref{assu:mm_iden} hold, then
\[
\hat{\beta}\stackrel{p}{\to}\beta_{0}\quad\text{ and }\quad\Vert\widehat{\bm{\alpha}}-\bm{\alpha}_{0}\Vert_{\infty}\stackrel{p}{\to}0.
\]
Furthermore, we have
\[
\sqrt{N}(\hat{\beta}-\beta_{0})-\mathrm{J}_{0}^{-1}B_{0}\stackrel{d}{\to}\mathcal{N}(0,\Omega_{0}).
\]
\end{thm}
Theorem~\ref{thm:jmm_alpha_beta_consistency} shows that $\bm{\widehat{\alpha}}$
is uniformly consistent for $\bm{\alpha_0}$ and the JMM estimator $\hat{\beta}$ is asymptotically normal. However, the
limiting distribution of  $\hat{\beta}$  does not center around $\beta_0$. The bias term $\mathrm{J}_{0}^{-1}B_{0}$
arises from estimating $\bm{\alpha}_{0}$. The incidental parameter
problem is common in the literature on nonlinear panel fixed effects
regression with large $N$ and $T$. Moreover, as $\Omega_0$ is generally greater than $\mathrm{I_0}^{-1}$  in the semi-definiteness sense,  $\hat{\beta}$  does not achieve the CRLB. We refine $\hat{\beta}$  by the following modules.

\subsection{One-Step Estimator}\label{subsec:one_step_estimator}

The first refinement on $\widehat{\beta}$ concerns achieving the CRLB. We follow \citet{le1969theorie}'s
one-step adjustment as specified in Module \ref{alg:mod_os}. Algebra shows
\[
\mathbb{E}\left[\frac{\partial s_{n}(\bm{\alpha},\beta_{0})}{\partial\bm{\alpha}}\Big|\mathbf{x},\bm{\alpha}_{0}\right]=\mathbf{0}_{K\times n}\text{ and }\mathbb{E}\left[\frac{\partial s_{n}(\bm{\alpha},\beta_{0})}{\partial\beta}\Big|\mathbf{x},\bm{\alpha}_{0}\right]=-\mathrm{I}_{n,0}.
\]
Therefore, a Taylor expansion on the right-hand side of (\ref{eq:beta_os})
yields 
\begin{equation}
\hat{\beta}_{\mathrm{OS}}-\b_{0}\approx\mathrm{I}_{n,0}^{-1}s_{n,0}\label{eq:taylor_expansion_beta_os}
\end{equation}
in large samples, where $s_{n}$ and $\mathrm{I}_{n}$ are the concentrated
score and information matrix defined in (\ref{eq:concentrated_info}).

\begin{rem}[OS and non-concavity in the fixed effects]\label{rem:os_no_concavity}
The normalizer in~(\ref{eq:beta_os}) is the outer-product Fisher
information $\mathbf{I}_{11}=\mathbb{E}[\mathbf{s}_1\mathbf{s}_1^{\top}\mid\mathbf{x},\bm{\alpha}]$,
not the negative concentrated Hessian
$-\partial^{2}\ell_{n}/\partial\bm{\alpha}\partial\bm{\alpha}^{\top}$.
The two agree in expectation at the truth by the information identity, but
only $\mathbf{I}_{11}$ is positive semidefinite for every
$(\bm{\alpha},\beta)$ by construction; its invertibility and diagonal approximation (Lemma~\ref{lem:I11_diag_approx}) rest on the positive
per-dyad Fisher information guaranteed by
Assumption~\ref{assu:ass_f}, not on concavity of $\ell_{n}$ in $\bm{\alpha}$.
Module~OS therefore remains well-defined and first-order efficient
under global non-concavity in the fixed effects, as arises structurally
under NTU and under TU whenever the shock density is not log-concave
(e.g., Student's $t$). A Newton update using
$-\partial^{2}\ell_{n}/\partial\bm{\alpha}\partial\bm{\alpha}^{\top}$
would be ill-conditioned or undefined here.
\end{rem}

To establish (\ref{eq:taylor_expansion_beta_os})
rigorously and hence the asymptotic normality of $\hat{\beta}_{\mathrm{OS}}$,
we impose an additional assumption on the conditioning of the
information matrix (\ref{eq:info_mat}). Define $\mathbf{D}(\bm{\alpha},\beta)\coloneqq\mathrm{diag}(\mathbf{I}_{11}(\bm{\alpha},\beta))$
and $\mathbf{Q}(\bm{\alpha},\beta)\coloneqq\mathbf{D}(\bm{\alpha},\beta)^{-1/2}\mathbf{I}_{11}(\bm{\alpha},\beta)\mathbf{D}(\bm{\alpha},\beta)^{-1/2}$.

\begin{assumption}[Information Matrix Conditioning]
\label{assu:res_order} 
\[
\sup_{(\bm{\alpha},\beta)\in\mathbb{A}\times\mathbb{B}}\left\Vert \mathbf{Q}(\bm{\alpha},\beta)^{-1}\right\Vert _{1}<c\text{ for some positive constant }c.
\]
\end{assumption}
The concentrated score $s_n$ is constructed by projecting the
fixed-effect scores out of the $\beta$ score. Define the profiling
weights
\begin{equation}
w_{ki}(\bm{\alpha},\beta) \coloneqq
\left[\mathbf{I}_{12}(\bm{\alpha},\beta)^{\top}
\mathbf{I}_{11}(\bm{\alpha},\beta)^{-1}\right]_{ki},
\end{equation}
which determine how much of node $i$'s fixed-effect score is
subtracted from the raw score for $\beta_{k}$. In the additive
TU-logit benchmark, the concentrated score takes the form
$s_{n,k}(\bm{\alpha},\beta)=\sum_{i<j}[y_{ij}-p_{ij}(\bm{\alpha},
\beta)](x_{ij,k}-w_{ki}-w_{kj})$, where $w_{ki}+w_{kj}$ is the best
weighted least-squares approximation to $x_{ij,k}$ with additive node
terms, a Frisch--Waugh--Lovell residualization.
Assumption~\ref{assu:res_order} is a stability condition on this
projection: the standardized matrix $\mathbf{Q}$ normalizes
$\mathbf{I}_{11}$ to unit diagonal, and bounding
$\lVert\mathbf{Q}^{-1}\rVert_{1}$ prevents the node-level information
contributions from becoming near-collinear as $n$ grows. For additive models
$p(\alpha_{i},\alpha_{j},x_{ij}^{\top}\beta)=F(\alpha_{i}+\alpha_{j}+x_{ij}^{\top}\beta)$
with smooth $F$, $p^{(1,0,0)}=p^{(0,1,0)}=F'$, which makes $\mathbf{I}_{11}$
diagonally dominant with off-diagonal entries of order $O(1)$.
In this case, $\mathbf{I}_{11}(\bm\alpha,\beta)^{-1}$ can be approximated by $\mathbf{D}(\bm\alpha,\beta)^{-1}$ with maximum entrywise error $O(n^{-2})$ by Lemma \ref{lem:appro_inverse}, then Assumption \ref{assu:res_order} is not needed (see Remark~\ref{rem:assu5_additive} in the Supplemental Material).

\begin{rem}
Lemma~\ref{prop:assu5_primitive} in Appendix~\ref{sec:supporting_lemmas}
shows that Assumption~\ref{assu:res_order}, together with
Assumptions~\ref{assu:model}--\ref{assu:ass_f},
implies the order bounds on $w_{ki}$ and its derivatives that are
used in the proof of Theorem~\ref{thm:os_norm}: $\sup_{k,i}|w_{ki}|=O(1)$,
$\sup_{k,i}\|\partial w_{ki}/\partial\beta\|=O(1)$,
$\sup_{k,i}|\partial w_{ki}/\partial\alpha_{i}|=O(1)$, and
$\sup_{k,i\neq j}|\partial w_{ki}/\partial\alpha_{j}|=O(n^{-1})$.
Intuitively, the last bound says that no single foreign node $j$ has
order-1 leverage on node $i$'s profiling weight, which is natural
in dense networks where each node averages information over $n-1$
dyads.
\end{rem}

The next theorem establishes the limit distribution of $\hat{\beta}_{\mathrm{OS}}$. Define the concentrated information limit
\begin{equation}
\mathrm{I}_{0} \coloneqq \mathrm{plim}_{n\to\infty}N^{-1}\mathrm{I}_{n,0}\label{eq:info_mat_i0_def}
\end{equation}
and the asymptotic bias $b_0 = (b_{10},\dots,b_{K0})^\top$ with
\[
b_{k0}=\lim_{n\to\infty}\frac{1}{\sqrt{N}}\mathrm{Tr}[\mathbf{J}_{11,0}^{-1}\text{Cov}(\mathbf{m}_{1,0},\mathbf{s}_{1,0})\mathbf{W}_{k,0}],\quad k=1,\dots,K,
\]
where $[\mathbf{W}_{k}(\bm{\alpha},\beta)]_{ij}=\frac{\partial w_{ki}(\bm{\alpha},\beta)}{\partial\alpha_{j}}$ and the comma-zero subscript denotes evaluation at $(\bm{\alpha}_{0},\beta_{0})$.
\begin{thm}
\label{thm:os_norm} If Assumptions \ref{assu:model}--\ref{assu:res_order}
hold, then
\[
\sqrt{N}(\hat{\beta}_{\mathrm{OS}}-\beta_{0})-\mathrm{I}_{0}^{-1}b_{0}\stackrel{d}{\to}\mathcal{N}(0,\mathrm{I}_{0}^{-1}).
\]
\end{thm}
Theorem~\ref{thm:os_norm} shows $\hat{\beta}_{\mathrm{OS}}$ achieves
the CRLB asymptotically. In the proof of Theorem~\ref{thm:os_norm},
we show that $b_{0}$ is $O(1)$ and depends on the covariance matrix
between $\mathbf{m}_{1}$ and $\mathbf{s}_{1}$. This is because our
plug-in estimator for $\bm{\alpha}$ is obtained from the moment estimating
equation $\mathbf{m}_{1}$, and the one-step estimator (\ref{eq:beta_os})
uses information from $\mathbf{s}_{1}$ to concentrate out $\bm{\alpha}$.
As a result, the covariance between $\mathbf{m}_{1}$ and $\mathbf{s}_{1}$
determines the magnitude of the term $b_{0}$ in the asymptotic bias
of Theorem~\ref{thm:os_norm}.

\subsection{Bagging}\label{subsec:bagging}

While reaching the CRLB, Theorem~\ref{thm:os_norm} reveals that $\hat{\beta}_{\mathrm{OS}}$
incurs an asymptotic bias. As discussed in Module \ref{alg:mod_bg},
one way to debias $\hat{\beta}_{\mathrm{OS}}$ is to use the split-network
jackknife to self-cancel the leading bias. However, it doubles the
asymptotic variance, making the confidence interval wider and less informative. Additionally, it suffers from computational instability because the network is split only once by random. The solution we propose is to split the network randomly many times, compute the split-network jackknife estimator from each split, then average them up. This is equivalent to bagging \citep{breiman1996bagging} on a split-network jackknife estimator. 

To motivate the bagging method, in theory there are a total of $T_{n} \coloneqq \binom{n}{n/2}$
possible ways to divide the network. However, $T_{n}$ can be very
large for a moderate sample size $n.$ For example, $n=100$ produces
$T_{n}=\binom{100}{50}\simeq1.009\times10^{29}$, which is an astronomical
number. We solve this problem by choosing $\tilde{T}_{n}\ll T_{n}$
in the BG module. In the simulations, we set $\tilde{T}_{n}=2n$ and
find that the results are robust to this choice.

The next theorem shows that when $n$ and $\tilde{T}_{n}$ go to infinity,
$\hat{\beta}_{\mathrm{BG}}$ is asymptotically normal, unbiased, and
efficient.
\begin{thm}
\label{thm:bagging_normality}If Assumptions \ref{assu:model}--\ref{assu:res_order}
hold, then 
\[
\sqrt{N}(\hat{\beta}_{\mathrm{BG}}-\beta_{0})\stackrel{d}{\to}\mathcal{N}(0,\mathrm{I}_{0}^{-1})
\]
 as $n\to\infty$ and $\tilde{T}_{n}\to\infty$.
\end{thm}
Theorem~\ref{thm:bagging_normality} is the main theoretical result
of this paper. A few remarks are in order to discuss its implications
and connections with the literature. First, $\hat{\beta}_{\mathrm{BG}}$
involves three modules---JMM, OS, and BG---which play different roles.
Module JMM provides an initial consistent yet biased estimator, which
is fed into Module OS to achieve the CRLB. Module BG corrects for
the bias in the OS estimator via split-network jackknife while maintaining
the efficiency through bagging.

Second, a similar idea to $\hat{\beta}_{\mathrm{OS-SJ}}$ in a panel
setting with fixed effects is presented in \citet{dhaene2015split}.
Although related, $\hat{\beta}_{\mathrm{BG}}$ is preferred over $\hat{\beta}_{\mathrm{OS-SJ}}$
because $\hat{\beta}_{\mathrm{OS-SJ}}$ has an asymptotic variance
of $2\mathrm{I}_{0}^{-1}$ while $\hat{\beta}_{\mathrm{BG}}$'s is
$\mathrm{I}_{0}^{-1}$.

Third, one may be inclined to apply BG to the initial JMM estimator
directly and bypass the one-step approximation. Indeed, BG can correct
for the asymptotic bias of the JMM estimator. However, the JMM-BG
estimator is not efficient because the asymptotic variance of the
initial JMM estimator is preserved through bagging.

Finally, sample splitting across individuals introduces a degree of
extra randomness, which motivates \citet[Footnote 8]{fernandez2016individual}
to suggest averaging of all possible $T_{n}$ partitions and point
out that the average over $\tilde{T}_{n}\ll T_{n}$ splits is sufficient.
The BG estimator in our context not only eliminates randomness from
sample splitting but also simultaneously achieves efficiency and bias
correction. Furthermore, our Theorem~\ref{thm:bagging_normality}
provides formal asymptotic results to justify the use of the BG estimator.

\begin{rem}[Multiple networks]\label{rem:multi-networks}
The estimator extends to the multi-network setting. Consider $V$ independent networks indexed by $v=1,\dots,V$ ($V$ fixed) sharing a common slope $\b_0$ but with network-specific node fixed effects:
\[
\Pr\!\left(Y_{ij,v}=1\mid X_{ij,v},\alpha_{i0,v},\alpha_{j0,v}\right)
= p\!\left(\alpha_{i0,v},\alpha_{j0,v},X_{ij,v}^{\top}\b_0\right),\quad 1\le i\neq j\le n_v,\ 1\le v\le V.
\]
The estimator pools across networks at the $\b$ level while keeping $\bm\alpha_v$ network-specific: the JMM module solves $\widehat{\bm\alpha}_v(\b)$ within each network from its own degree equations, and $\widehat{\b}$ solves the pooled moment $\sum_v m_{2,v}(\widehat{\bm\alpha}_v(\b),\b)=0$; the one-step refinement uses the pooled concentrated information $\sum_v\mathrm{I}_{n_v}(\hat{\bm\alpha}_v,\hat\beta)$ and score $\sum_v s_{n_v}(\hat{\bm\alpha}_v,\hat\beta)$; bagging draws independent equal partitions $\{\mathcal{I}_{1,n_v}^{(t)},\mathcal{I}_{2,n_v}^{(t)}\}$ of each $\mathcal{I}_{n_v}$ for $t=1,\dots,\widetilde T_n$, applies the pooled one-step update on each half-pool $\bigcup_v \mathcal{I}_{h,n_v}^{(t)}$ (using $\sum_v \mathrm{I}_{n_v/2,v}$ and $\sum_v s_{n_v/2,v}$), and averages the resulting split-network jackknife corrections. Assumptions \ref{assu:model}--\ref{assu:ass_f} and~\ref{assu:res_order} hold per network, and Assumption \ref{assu:mm_iden} is imposed on the pooled $\bar S_n$. Theorems~\ref{thm:jmm_alpha_beta_consistency}--\ref{thm:bagging_normality} extend with $N=\sum_v\binom{n_v}{2}$ as $\min_v n_v\to\infty$ (Proposition~\ref{prop:multinetwork} in the Supplement).
\end{rem}

\section{Extensions}\label{sec:extensions}

This section extends our framework in two directions: average partial effects for the binary outcome, and the behavior of our estimators when the link function $p(\cdot)$ is misspecified.

\paragraph{Average partial effects.}
Beyond the homophily parameters themselves, researchers and
policy makers may also be interested in population averages over the
distribution of exogenous regressors and fixed effects. One leading
example concerns the conditional mean of the outcome given covariates
and fixed effects
\begin{equation}
\E\left[\rest{Y_{ij}}X_{ij}=x_{ij},\bm{\alpha}\right]=p_{ij}(\bm\alpha,\beta_0).\label{eq:condl_exp_y}
\end{equation}
Here, the partial effects are defined as the differences or derivatives
of (\ref{eq:condl_exp_y}) with respect to components
of $X_{ij}$, say $X_{ij,k}$, the $k$th coordinate of
$X_{ij}$. We suppress its dependence on $Y$ and $X$ and define
the partial effect of $x_{ij,k}$ for the dyad $(i,j)$ as
\[
\Delta_{ij,k}(\alpha_{i},\alpha_{j},\beta)=\begin{cases}
p(\alpha_{i},\alpha_{j},\beta_{k}+x_{ij,-k}^{\top}\beta_{-k})-p(\alpha_{i},\alpha_{j},x_{ij,-k}^{\top}\beta_{-k}) & (b)\\
\beta_{k}p^{(0,0,1)}(\alpha_i,\alpha_j,x_{ij}^\top\beta) & (c)
\end{cases}
\]
where ``$(b)$'' corresponds to binary $x_{ij,k}$ while ``$(c)$''
refers to continuous $x_{ij,k}$. Define $\Delta_{ij}=(\Delta_{ij,1},\dots,\Delta_{ij,K})^{\top}$.
Then, the unconditional APEs are 
\begin{equation}
\delta_{0}=\mathbb{E}\left[\frac{1}{N}\sum_{i=1}^{n}\sum_{j>i}\Delta_{ij}(\alpha_{i},\alpha_{j},\beta_{0})\right].\label{eq:ape}
\end{equation}
Plugging the JMM estimator $(\widehat{\bm{\alpha}},\widehat{\b})$
into (\ref{eq:ape}) yields an estimator for the APEs 
\[
\hat{\delta}=\frac{1}{N}\sum_{i=1}^{n}\sum_{j>i}\Delta_{ij}(\hat{\alpha}_{i},\hat{\alpha}_{j},\hat{\beta}).
\]

Define an (infeasible) $\bar{\Delta}_{n}=\frac{1}{N}\sum_{i=1}^{n}\sum_{j>i}\Delta_{ij}(\alpha_{i0},\alpha_{j0},\beta_{0})$.
Let the split-jackknife estimator and the bagging estimator of the
APE be
\[
\hat{\delta}_{\mathrm{SJ}} \coloneqq 2\hat{\delta}-\frac{1}{2}(\hat{\delta}_{1}+\hat{\delta}_{2})\quad\text{and}\quad\hat{\delta}_{\mathrm{BG}} \coloneqq \frac{1}{\tilde{T}_{n}}\sum_{t=1}^{\tilde{T}_{n}}\hat{\delta}_{\mathrm{SJ}}^{(t)},
\]
respectively. Here, $(\hat{\delta}_{1},\hat{\delta}_{2})$ are the
plug-in estimators based on two sub-networks after a random split
of the nodes and $\{\hat{\delta}_{\mathrm{SJ}}^{(t)}\}_{t=1}^{\tilde{T}_{n}}$
are split-network jackknife estimators based on $\tilde{T}_{n}$ random
splits. The next theorem shows that the bias of $\hat{\delta}$ is asymptotically
negligible. We use a central limit theorem for U-statistics \citep[Theorem 12.3]{van2000asymptotic}
to prove it. To state the result precisely, we incorporate the asymptotically
vanishing bias terms, as in \citet[Theorem 4.2]{fernandez2016individual},
and establish that $\hat{\delta}_{\mathrm{SJ}}$ and $\hat{\delta}_{\mathrm{BG}}$
are equivalent to  $\hat{\delta}$  asymptotically. Section \ref{sec:simulation}  presents numerical evidence
that supports this claim.

To state the next theorem, we define $\Sigma_{\delta,n} \coloneqq \frac{\Sigma_{\Delta}}{N}+\frac{4\Sigma_{\delta}}{n}$,
where $\Sigma_{\Delta}$ is defined in (\ref{eq:sigma_delta}) and
$\Sigma_{\delta}=\mathrm{Cov}\left(\Delta_{12}(\alpha_{1},\alpha_{2},\beta_{0}),\Delta_{13}(\alpha_{1},\alpha_{3},\beta_{0})\right)$.\footnote{We acknowledge the presence of high-order variance terms arising from U-statistics, but omit these terms for simplicity. Our simulation results confirm that this omission does not compromise robustness.} The two bias terms $B_\beta$ and $B_\alpha$ in the next theorem are defined in (\ref{eq:b_alpha_b_beta}).

\begin{thm}
\label{thm:norm_ape}If Assumptions \ref{assu:model}--\ref{assu:mm_iden}
hold and $\bar{\Delta}_{n}$ is a non-degenerate U-statistic, then
\begin{align*}
\Sigma_{\delta,n}^{-1/2}\left(\hat{\delta}-\delta_{0}-\frac{1}{\sqrt{N}}B_{\beta}-\frac{1}{\sqrt{N}}B_{\alpha}\right) & \stackrel{d}{\to}\mathcal{N}(0,I_{K})\text{ and}\\
\Sigma_{\delta,n}^{-1/2}\left(\hat{\delta}_{\mathrm{BG}}-\delta_{0}\right) & \stackrel{d}{\to}\mathcal{N}(0,I_{K}).
\end{align*}
\end{thm}
In Theorem~\ref{thm:norm_ape}, the rate of convergence of $\hat{\delta}$
(and $\hat{\delta}_{\mathrm{BG}}$) is $\sqrt{n}$ instead of $\sqrt{N}$.
The slower convergence rate in Theorem~\ref{thm:norm_ape} makes
the bias terms introduced by estimating $\bm{\alpha_0}$
asymptotically negligible. Note that $B_{\beta}$ reflects
the bias of the plug-in estimator $\hat{\beta}$ whereas $B_{\alpha}$
arises from the incidental parameter bias of the plug-in estimator
$\widehat{\bm{\alpha}}$. 

For the components of $\Sigma_{\delta,n}$, $\Sigma_{\Delta}$ is the
asymptotic variance of $\sqrt{N}(\hat{\delta}-\bar{\Delta}_{n})$
and $4\Sigma_{\delta}$ is the asymptotic variance of $\sqrt{n}(\bar{\Delta}_{n}-\delta_{0})$.
$\Sigma_{\delta}$ can be estimated by
\[
\hat{\Sigma}_{\delta}=\frac{1}{n(n-1)(n-2)}\sum_{i\neq j\neq k}u_{ij}u_{ik}^{\top},
\]
where $u_{ij}\coloneqq\Delta_{ij}(\hat{\alpha}_{i},\hat{\alpha}_{j},\hat{\beta})-\hat{\delta}$ and $\sum_{i\neq j\neq k}$ denotes the sum over all triads $(i,j,k)$ with distinct indices.
This symmetrized estimator is permutation-invariant in finite samples and is consistent by the law of large numbers for U-statistics.
Although the variance term $\Sigma_{\D}/N$ is dominated asymptotically
by $4\Sigma_{\delta}/n$ in Theorem~\ref{thm:norm_ape}, we find
in simulations that including it improves the coverage probabilities.
This is particularly relevant when unobserved heterogeneity is limited (i.e., $\alpha_i\approx\alpha_j$ for most nodes), in which case the U-statistic leading term $\Sigma_{\delta}$ may be small and the higher-order term $\Sigma_{\Delta}/N$ contributes non-negligibly to the finite-sample variance. We therefore recommend including both terms in practice.
If one instead targets $\bar{\Delta}_{n}$, the asymptotic result becomes $\sqrt{N}(\hat{\delta}-\bar{\Delta}_{n})-B_{\beta}-B_{\alpha}\stackrel{d}{\to}\mathcal{N}(0,\Sigma_{\Delta})$, generalizing Theorem 2 of \citet{chen2021nonlinear} to our setting.

\paragraph{Link function misspecification.}\label{subsec:model_misspecification}
The preceding results assume that the link function $p(\cdot)$ is correctly specified.
In Section~\ref{sec:sm_misspecification} of the Supplemental Material, we show that our estimators remain well-behaved when $p(\cdot)$ is replaced by a misspecified link function $q(\cdot)$. Specifically, the JMM estimator $\widehat{\b}$ converges to a well-defined pseudo-true value $\beta_{n*}$ that solves a pseudo-population moment equation, with a sandwich-form asymptotic variance that can be consistently estimated. The one-step and BG estimators center on a projected pseudo-true value $\beta_{n\sharp}$ that accounts for the nonzero population concentrated score under misspecification. Crucially, the BG estimator retains its bias-correction property: $\sqrt{N}(\widehat{\b}_{\mathrm{BG}}-\beta_{n\sharp})\stackrel{d}{\to}\mathcal{N}(0,\Gamma_*)$, where the sandwich covariance $\Gamma_*$ reflects the discrepancy between the true and misspecified likelihoods. Section~\ref{sec:sm_simulations} of the Supplemental Material provides supporting numerical evidence.

\section{Monte Carlo Simulations}\label{sec:simulation}

We conduct Monte Carlo simulations to evaluate the finite-sample performance of our estimators under both the TU and NTU specifications. Additional exercises---the non-concavity challenge of direct MLE, fixed-effect recovery, APE estimation, link-function misspecification, sparser networks, and extended TU results---are reported in Section~\ref{sec:sm_simulations} of the Supplemental Material.

The data generating process (DGP) is as follows. We set $\beta_{0}=(1,-1)^{\top}$,
and draw the first covariate of $X_{ij}$ as $X_{1,ij}\stackrel{\text{i.i.d.}}{\sim}\text{Bernoulli}(0.3),\ X_{1,ij}=X_{1,ji}$.
In this way, we allow for a discrete variable in $X_{ij}$. For the second
covariate of $X_{ij}$, we draw $X_{i}\stackrel{\text{i.i.d.}}{\sim}U(-0.5,0.5)$
and let $X_{2,ij}=|X_{i}-X_{j}|$. Next, we generate the individual
fixed effects by setting $\a_{i}=0.75\times\xi_{i}+0.25\times X_{i}$,
where $\xi_{i}\stackrel{\text{i.i.d.}}{\sim}U(-0.5,0.5)$ and is independent
of all other variables so that $\a_{i}$ and $X_{ij}$ are correlated
via $X_{i}$.

We consider two specifications for the link function. Under \emph{TU}, a dyadic surplus determines the link:
\[
Y_{ij}^{\mathrm{TU}}=\ind\{\alpha_{i}+\alpha_{j}+X_{ij}^{\top}\beta_{0}-\epsilon_{ij}>0\},
\]
with $\epsilon_{ij}$ drawn i.i.d.\ from either the standard logistic or standard normal distribution, yielding the logit and probit links respectively. Under \emph{NTU}, link formation requires bilateral consent:
\[
Y_{ij}^{\mathrm{NTU}}=\ind\{\alpha_{i}+X_{ij}^{\top}\beta_{0}-\epsilon_{ij}>0\}\cdot \ind\{\alpha_{j}+X_{ij}^{\top}\beta_{0}-\epsilon_{ji}>0\},
\]
with $\epsilon_{ij}$ drawn i.i.d.\ from the standard logistic distribution.

For all simulations in this paper, we run $1{,}000$ replications. For the baseline results, we set $n=100$ and $200$, which are comparable to the sizes of the data used in our empirical applications.\footnote{We also run the same exercise for $n=50$, and the conclusions remain largely the same.} We report the mean and median bias, standard deviation, mean and median absolute bias, and the root mean squared error (RMSE) across replications.

\begin{table}[!htbp]
\centering
\begin{centering}
\caption{TU baseline estimation results for $\protect\b_0$ ($n=100$)}\label{tab:tu_sim_beta_n100}
\par\end{centering}
\begin{threeparttable}

\begin{tabular}{lrrrrrr}
\toprule 
\multirow{2}{*}{\emph{Logit}} & \multicolumn{2}{c}{JMM} & \multicolumn{2}{c}{OS} & \multicolumn{2}{c}{BG}\tabularnewline
\cmidrule{2-7}
 & $\b_{1}$ & $\b_{2}$ & $\b_{1}$ & $\b_{2}$ & $\b_{1}$ & $\b_{2}$\tabularnewline
\midrule
Mean Bias & 2.32 & -1.97 & 2.44 & -1.99 & 0.15 & 0.21\tabularnewline
Median Bias & 2.33 & -2.25 & 2.43 & -2.35 & 0.14 & -0.20\tabularnewline
Standard Deviation & 6.96 & 14.67 & 6.97 & 14.74 & 6.79 & 14.46\tabularnewline
Mean Standard Error & 6.74 & 14.56 & 6.75 & 14.56 & 6.73 & 14.54\tabularnewline
Mean Absolute Bias & 5.82 & 11.87 & 5.87 & 11.93 & 5.39 & 11.54\tabularnewline
Median Absolute Bias & 4.87 & 9.98 & 4.90 & 10.09 & 4.52 & 9.90\tabularnewline
RMSE & 7.33 & 14.80 & 7.38 & 14.87 & 6.79 & 14.46\tabularnewline
90\% Coverage Rate & 86.6 & 90.1 & 86.0 & 89.9 & 89.1 & 89.5\tabularnewline
95\% Coverage Rate & 92.1 & 95.1 & 92.1 & 95.1 & 94.5 & 95.3\tabularnewline
\midrule
\multirow{2}{*}{\emph{Probit}} & \multicolumn{2}{c}{JMM} & \multicolumn{2}{c}{OS} & \multicolumn{2}{c}{BG}\tabularnewline
\cmidrule{2-7}
 & $\b_{1}$ & $\b_{2}$ & $\b_{1}$ & $\b_{2}$ & $\b_{1}$ & $\b_{2}$\tabularnewline
\midrule
Mean Bias & 1.95 & -1.67 & 2.06 & -1.68 & -0.05 & 0.43\tabularnewline
Median Bias & 1.81 & -1.88 & 1.93 & -1.94 & -0.21 & 0.32\tabularnewline
Standard Deviation & 4.44 & 9.46 & 4.44 & 9.50 & 4.32 & 9.35\tabularnewline
Mean Standard Error & 4.37 & 9.35 & 4.36 & 9.34 & 4.35 & 9.32\tabularnewline
Mean Absolute Bias & 3.88 & 7.77 & 3.92 & 7.80 & 3.49 & 7.53\tabularnewline
Median Absolute Bias & 3.41 & 6.75 & 3.40 & 6.83 & 2.90 & 6.51\tabularnewline
RMSE & 4.84 & 9.61 & 4.89 & 9.65 & 4.32 & 9.36\tabularnewline
90\% Coverage Rate & 85.8 & 89.4 & 85.5 & 89.9 & 90.2 & 89.7\tabularnewline
95\% Coverage Rate & 91.6 & 94.6 & 91.3 & 94.2 & 95.5 & 95.4\tabularnewline
\bottomrule
\end{tabular}

\begin{tablenotes}
\setlength{\itemindent}{-0.5cm}
\item \footnotesize \textit{Note:} All values have been multiplied by 100.
\end{tablenotes}
\end{threeparttable}

\end{table}

Table~\ref{tab:tu_sim_beta_n100} reports the TU baseline results at $n=100$ under both link functions. The BG estimator substantially reduces the bias of JMM and OS with essentially the same dispersion, achieving the lowest RMSE, and its coverage probabilities lie close to their nominal levels. The patterns are robust across logit and probit links, consistent with Theorem~\ref{thm:bagging_normality}.

\begin{table}[!htbp]
\centering
\begin{centering}
\caption{NTU baseline estimation results for $\protect\b_{0}$}\label{tab:sim_beta}
\par\end{centering}
\begin{threeparttable}

\begin{tabular}{lrrrrrr}
\toprule 
\multirow{2}{*}{\emph{$n=100$}} & \multicolumn{2}{c}{JMM} & \multicolumn{2}{c}{OS} & \multicolumn{2}{c}{BG}\tabularnewline
\cmidrule{2-7}
 & $\b_{1}$ & $\b_{2}$ & $\b_{1}$ & $\b_{2}$ & $\b_{1}$ & $\b_{2}$\tabularnewline
\midrule 
Mean Bias & 2.95 & -2.91 & 2.77 & -2.71 & -0.37 & 0.33\tabularnewline
Median Bias & 2.90 & -3.01 & 2.77 & -2.79 & -0.40 & 0.24\tabularnewline
Standard Deviation & 5.71 & 13.05 & 5.67 & 13.02 & 5.51 & 12.69\tabularnewline
Mean Standard Error & 5.67 & 12.98 & 5.66 & 12.92 & 5.66 & 12.92\tabularnewline
Mean Absolute Bias & 5.17 & 10.68 & 5.07 & 10.61 & 4.45 & 10.12\tabularnewline
Median Absolute Bias & 4.35 & 8.89 & 4.21 & 8.92 & 3.92 & 8.53\tabularnewline
RMSE & 6.42 & 13.37 & 6.31 & 13.30 & 5.52 & 12.70\tabularnewline
90\% Coverage Rate & 83.7 & 89.7 & 84.0 & 89.5 & 91.2 & 90.9\tabularnewline
95\% Coverage Rate & 91.8 & 94.1 & 92.5 & 94.4 & 95.6 & 95.5\tabularnewline
\midrule
\multirow{2}{*}{\emph{$n=200$}} & \multicolumn{2}{c}{JMM} & \multicolumn{2}{c}{OS} & \multicolumn{2}{c}{BG}\tabularnewline
\cmidrule{2-7}
 & $\b_{1}$ & $\b_{2}$ & $\b_{1}$ & $\b_{2}$ & $\b_{1}$ & $\b_{2}$\tabularnewline
\midrule
Mean Bias & 1.54 & -1.71 & 1.47 & -1.61 & -0.06 & -0.12\tabularnewline
Median Bias & 1.43 & -1.73 & 1.33 & -1.71 & -0.20 & -0.28\tabularnewline
Standard Deviation & 2.85 & 6.41 & 2.85 & 6.41 & 2.80 & 6.33\tabularnewline
Mean Standard Error & 2.78 & 6.35 & 2.78 & 6.32 & 2.78 & 6.32\tabularnewline
Mean Absolute Bias & 2.57 & 5.36 & 2.54 & 5.34 & 2.28 & 5.12\tabularnewline
Median Absolute Bias & 2.29 & 4.66 & 2.23 & 4.72 & 1.96 & 4.47\tabularnewline
RMSE & 3.24 & 6.63 & 3.20 & 6.61 & 2.80 & 6.33\tabularnewline
90\% Coverage Rate & 84.8 & 88.5 & 85.3 & 88.0 & 90.2 & 90.0\tabularnewline
95\% Coverage Rate & 90.8 & 93.9 & 90.9 & 94.3 & 95.3 & 95.0\tabularnewline
\bottomrule
\end{tabular}

\begin{tablenotes}
\setlength{\itemindent}{-0.5cm}
\item \footnotesize \textit{Note:} All values have been multiplied by 100.
\end{tablenotes}
\end{threeparttable}

\end{table}

Table~\ref{tab:sim_beta} reports the NTU results for $n=100$ and $200$ when the network has a density of 25\%. Here are the main observations when $n=100$. First, in terms of the bias, the BG estimator performs significantly better than JMM and OS, which is consistent with Theorem~\ref{thm:bagging_normality}. Second, BG works very well in simultaneously achieving bias-correction and low standard deviation, leading to the lowest RMSE.\footnote{Though not reported in the tables, we find that SJ (without BG) doubles the variance of JMM and OS estimators in the simulations, which is in line with the theory.} Third, the coverage probabilities of the BG confidence intervals are close to their nominal levels, while JMM and OS exhibit systematic undercoverage due to incidental parameter bias, precisely the problem that BG is designed to correct. Finally, the mean standard errors implied by the asymptotic theory are close to the standard deviations computed from the Monte Carlo simulations across all estimators. We also find that the quantiles of the empirical distributions for all estimators are well approximated by the same quantiles of the corresponding asymptotic normal distributions. These results further support our theoretical findings. When $n=200$, the performance of all the estimators improves. The RMSEs, for example, are about half the size of those when $n=100$, which is expected given the $\sqrt{N}$-convergence rate and $\sqrt{N}=O\left(n\right)$. The coverage probabilities also improve.\footnote{We have checked the stability of $\hat{\beta}_{\text{BG}}$ for $\tilde{T}_n$ chosen from a wide range, $\{50,100,200,400,600,800\}$. The bias, RMSE and coverage are almost identical across different choices of $\tilde{T}_n$.}

Finite-sample evidence for the APE estimator of Theorem~\ref{thm:norm_ape} is reported in Section~\ref{sec:sm_simulations} of the Supplemental Material. The plug-in estimator exhibits near-zero bias and coverage close to the nominal level under both TU and NTU, consistent with the $O(n^{-1})$ bias bound. Extended results on fixed-effect recovery, link-function misspecification, and sparser networks are also provided there.

\section{Empirical Applications}\label{sec:empirical_applications}

This section presents two empirical applications that together exhibit the TU and NTU regimes, respectively. First, under TU, we apply our method to the Townsend Thai village networks \citep{kinnan2024propagation}, pooling across 16 villages to estimate link formation for financial, operations, and labor transactions. We find that kinship has a strong and broadly positive effect across all three networks, and that greater absolute net-worth differences between households are also associated with more frequent linking, a pattern consistent with gains-from-trade in transactional relationships under TU. Second, under NTU, we revisit the Nyakatoke risk-sharing network \citep*{deweerdt2004risk}, where our method indicates that wealth differences have no statistically significant effect on link formation, and we provide economic intuition for this result.

The two applications illustrate the framework's ability to accommodate both modeling regimes. For the Thai village networks, households transact in goods, labor, and financial instruments that carry transferable value, so we adopt the TU specification \citep{de2020trade}. For the Nyakatoke risk-sharing network, links record mutual acknowledgements without explicit transfers or side payments, so we adopt the NTU specification. We use logit and probit links for the Thai TU application and the logit link for Nyakatoke. In both cases, we report the JMM, OS, and BG estimates together with the plug-in and bagging APEs.

The BG estimators and standard errors in the two empirical applications are robust to different choices of $\tilde{T}_n$, consistent with our simulation results.
Simulations calibrated to comparable empirical densities (Section~\ref{sec:sm_simulations} of the Supplemental Material) show that the estimators retain their bias-correction and coverage properties.

\subsection{Townsend Thai Village Networks}\label{subsec:emp_thai}

We apply our method to the Townsend Thai monthly panel \citep{kinnan2024propagation}, covering 16 rural and peri-urban Thai villages with an average of 44 households per village (pooled sample of $N=15{,}641$ dyads). We construct three binary network indicators from the monthly transaction records, aggregated to the annual level: a \emph{financial} link (exchange of gifts or informal loans/repayments), an \emph{operations} link (transactions in production inputs, intermediate goods, or output sales), and a \emph{labor} link (hiring or labor exchange). Our dyadic covariates are (i) the (ln) demographic difference between two households, based on household composition and head characteristics; (ii) the (ln) net-worth difference; and (iii) an indicator for kinship. A detailed data description and summary statistics are provided in Section~\ref{sec:sm_thai} of the Supplemental Material.

To handle the multi-village structure, we estimate a pooled TU model with common slope $\b_0$ and village-specific node fixed effects following Remark~\ref{rem:multi-networks}:
\[
\Pr\!\left(Y_{ij,v}=1\mid X_{ij},\alpha_{i0,v},\alpha_{j0,v}\right)
= F\!\left(\alpha_{i0,v}+\alpha_{j0,v}+X_{ij}^{\top}\b_0\right),\quad 1\le v\le V=16,
\]
where $F$ is the logistic or normal CDF. 

\subsubsection{Year Selection and Estimation Results}

Because our asymptotic theory assumes a dense-network regime, we select for each network type the year with the highest pooled density: year 3 for the financial network (1.46\%), year 2 for the operations network (5.36\%), and year 3 for the labor network (7.79\%). Table~\ref{tab:thai_results} reports the pooled BG estimates and APEs under both the logit and probit link functions.

\begin{table}[!htbp]
\centering
\caption{BG estimates and APEs for the Townsend Thai village networks}\label{tab:thai_results}

\begin{threeparttable}
\begin{tabular}{lrrrr}
\toprule
\addlinespace
& \multicolumn{2}{c}{Coefficients} & \multicolumn{2}{c}{APEs} \\
\cmidrule(lr){2-3}\cmidrule(lr){4-5}
Variables & Logit & Probit & Logit & Probit \\
\midrule
\multicolumn{5}{l}{\textbf{Financial network}} \\
(ln) Demographic difference & 0.2629 & 0.1249 & 0.0029 & 0.0024 \\
 & (0.1862) & (0.0967) & (0.0017) & (0.0016) \\
(ln) Net-worth difference & 0.2567 & 0.1293 & 0.0028 & 0.0024 \\
 & (0.1059) & (0.0632) & (0.0011) & (0.0010) \\
Kinship & 3.2259 & 1.9318 & 0.0773 & 0.0896 \\
 & (0.2676) & (0.1512) & (0.0097) & (0.0101) \\
\addlinespace
\multicolumn{5}{l}{\textbf{Operations network}} \\
(ln) Demographic difference & 0.1553 & 0.1006 & 0.0050 & 0.0061 \\
 & (0.1021) & (0.0532) & (0.0029) & (0.0029) \\
(ln) Net-worth difference & 0.1221 & 0.0658 & 0.0040 & 0.0040 \\
 & (0.0563) & (0.0294) & (0.0016) & (0.0016) \\
Kinship & 1.9651 & 1.0820 & 0.0973 & 0.0998 \\
 & (0.1778) & (0.0964) & (0.0119) & (0.0121) \\
\addlinespace
\multicolumn{5}{l}{\textbf{Labor network}} \\
(ln) Demographic difference & 0.0569 & 0.0316 & 0.0029 & 0.0029 \\
 & (0.0771) & (0.0418) & (0.0035) & (0.0035) \\
(ln) Net-worth difference & 0.1691 & 0.1056 & 0.0085 & 0.0096 \\
 & (0.0456) & (0.0251) & (0.0022) & (0.0022) \\
Kinship & 2.2562 & 1.3040 & 0.1711 & 0.1774 \\
 & (0.1532) & (0.0856) & (0.0143) & (0.0147) \\
\bottomrule
\end{tabular}

\begin{tablenotes}
\setlength{\itemindent}{-0.5cm}
\item \footnotesize \textit{Note:} Pooled BG estimates of the common slope $\b_0$ in the TU specification with village-specific node fixed effects, using the logit and probit link functions. For each network type, we select the year with the highest pooled density to stay close to the dense-network regime assumed by the theory. Standard errors in parentheses.
\end{tablenotes}
\end{threeparttable}
\end{table}

Three findings stand out. First, \emph{kinship} is a strong and highly significant predictor of link formation across all three networks, with the largest APE in the labor network (an increase in linking probability of about 17 percentage points). This is consistent with \citet{kinnan2024propagation}, who document that kinship ties are persistent predictors of transactions in these villages over a long time horizon. Second, the coefficient on \emph{(ln) net-worth difference} is positive and statistically significant in all three networks, indicating that households with larger wealth disparities transact more frequently under the TU regime. Economically, this pattern is consistent with gains-from-trade in dyadic transactions: wealthier households can supply credit, labor demand, or production inputs that poorer households consume or provide, yielding bilateral surplus that the TU framework predicts. The contrast with the Nyakatoke finding below, where absolute wealth differences have no significant effect under NTU, highlights the role of the modeling regime: when linking requires mutual consent without transfers, the gains-from-trade channel is muted. Third, the coefficient on \emph{(ln) demographic difference} is positive (directionally consistent with heterophily: demographically distant households transacting more often) but statistically insignificant, suggesting that household composition is not a primary driver of transactional linking after controlling for kinship and wealth. The logit and probit specifications yield qualitatively identical patterns, and the implied APEs are nearly identical in sign and magnitude.

\subsection{Nyakatoke Risk-Sharing Network}

We apply our method to the Nyakatoke risk-sharing network in Tanzania \citep{deweerdt2004risk}, which covers 114 households with $N=6{,}441$ dyadic observations. We estimate the model using three covariates: absolute (ln) wealth difference, (ln) distance, and a categorical kinship/religion tie variable. A detailed data description and summary statistics are provided in Section~\ref{sec:sm_nyakatoke} of the Supplemental Material.

\subsubsection{Results and Discussion}\label{subsec:emp_results_and_discussion}

Table~\ref{tab:real_data_results} presents the estimation results for the homophily coefficients and the APEs. The estimated coefficient on wealth difference is negative under all three methods, but the null of zero cannot be rejected under the BG asymptotic distribution. The NTU bilateral-consent logic provides a natural interpretation: similar-wealth pairs lack the capacity to insure each other against large shocks, while substantially unequal pairs face veto by the wealthier household whose expected surplus from the arrangement is typically negative. Both effects push the average wealth-difference effect toward zero. \citet{gao2023logical} also obtain a negative coefficient but without inference; our framework complements theirs by quantifying statistical uncertainty through the parametric link function.

In addition to the wealth difference, under BG the coefficient for
\textit{distance} is significantly negative at -0.8098, and that of
\textit{tie} is significantly positive at 0.5875. The results are
intuitive. We further report the APEs in the last two columns of Table
\ref{tab:real_data_results}. We find that the APE of wealth difference
is not significant at the 5\% significance level based on either the plug-in or the bagging estimator.
Distances between households and social ties, on the other hand, matter
more significantly in terms of the APE.
\begin{table}[!htbp]
\centering
\caption{Estimation results for the Nyakatoke network}\label{tab:real_data_results}

\begin{threeparttable}

\begin{tabular}{lrrrrr}
\toprule 
\addlinespace
\multirow{2}{*}{Variables} & \multicolumn{3}{c}{Coefficients} & \multicolumn{2}{c}{APEs}\tabularnewline\addlinespace
\cmidrule{2-6}
\addlinespace
 & JMM & OS & BG & Plug-in & Bagging\tabularnewline
\midrule
(ln) wealth difference & -0.0882 & -0.0974 & -0.0783 & -0.0065 & -0.0096\tabularnewline
 & (0.0680) & (0.0641) & (0.0641) & (0.0052) & (0.0052)\tabularnewline
(ln) distance & -0.7824 & -0.8636 & -0.8098 & -0.0576 & -0.0638\tabularnewline
 & (0.0530) & (0.0536) & (0.0536) & (0.0066) & (0.0064)\tabularnewline
Tie & 0.6714 & 0.6287 & 0.5875 & 0.0514 & 0.0500\tabularnewline
 & (0.0548) & (0.0556) & (0.0556) & (0.0060) & (0.0059)\tabularnewline
\bottomrule
\end{tabular}

\begin{tablenotes}
\setlength{\itemindent}{-0.5cm}
\item \footnotesize \textit{Note:} Standard errors are reported in the parentheses.
\end{tablenotes}
\end{threeparttable}
\end{table}

\section{Conclusion}\label{sec:conclusion}

This paper develops a unified estimation and inference framework for
dyadic network formation models with individual fixed effects under
both transferable and nontransferable utilities. The key methodological
innovation is a bagging estimator that combines a one-step
approximation to the MLE with a split-network jackknife,
delivering asymptotically unbiased and efficient inference for the
homophily parameters without requiring concavity of the log-likelihood in high-dimensional fixed effects or specific distributional assumptions on the link function. We also obtain $\ell_{\infty}$-consistent estimators of the individual fixed effects, asymptotically normal estimators of the average partial effects, an extension to multiple networks with a common slope, and a robustness analysis under link-function misspecification.
Two empirical applications, the Townsend Thai village networks (TU)
and the Nyakatoke risk-sharing network (NTU), illustrate that the
method is straightforward to implement under both modeling regimes
and yields interpretable results.

Several extensions merit future research. First, extending the framework to \emph{directed networks} would capture asymmetric relationships such as trade flows or citations by introducing separate in- and out-degree fixed effects. Second, establishing \emph{sparse-network} asymptotics, where the average degree grows slower than $n$, would relax the dense-network regime assumed by our theory and broaden the method's applicability. Third, incorporating \emph{strategic interactions}, where link decisions depend on others' anticipated choices, would bring the framework closer to equilibrium models of network formation.

\onehalfspacing


\appendix

\begin{center}
{\LARGE\bfseries Appendix}
\end{center}

\section{Notation and Definitions}\label{sec:appendix_defs}

We collect additional matrix definitions used in the main text and appendix.

\paragraph{Jacobian of the moment equations.} The Jacobian matrix $\mathbf{J}(\bm{\alpha},\beta) \coloneqq \nabla\mathbf{m}(\bm{\alpha},\beta)$ is partitioned as
\[
\mathbf{J}(\bm{\alpha},\beta)=\begin{pmatrix}\mathbf{J}_{11}(\bm{\alpha},\beta) & \mathbf{J}_{12}(\bm{\alpha},\beta)\\
\mathbf{J}_{21}(\bm{\alpha},\beta) & \mathrm{J}_{22}(\bm{\alpha},\beta)
\end{pmatrix},
\]
where $\mathbf{J}_{11}$ is $n\times n$ with $[\mathbf{J}_{11}]_{ij}=-p_{ij}^{(0,1,0)}$ for $i\neq j$ and $[\mathbf{J}_{11}]_{ii}=-\sum_{j\neq i}p_{ij}^{(1,0,0)}$; the $i$th row of $\mathbf{J}_{12}$ ($n\times K$) is $-\sum_{j\neq i}p_{ij}^{(0,0,1)}x_{ij}^{\top}$; the $i$th column of $\mathbf{J}_{21}$ ($K\times n$) is $-\sum_{j\neq i}p_{ij}^{(1,0,0)}x_{ij}$; and $\mathrm{J}_{22}=-\sum_{i<j}p_{ij}^{(0,0,1)}x_{ij}x_{ij}^{\top}$. The concentrated Jacobian for $\beta$ is
\begin{equation}
\label{eq:def_J_n}
\mathrm{J}_{n}(\bm{\alpha},\beta) \coloneqq \mathrm{J}_{22}-\mathbf{J}_{21}\mathbf{J}_{11}^{-1}\mathbf{J}_{12}.
\end{equation}

\paragraph{Variance of the moment equations.} Define the variance blocks $\mathbf{V}_{11}(\bm{\alpha},\beta)$ ($n\times n$), $\mathbf{V}_{12}(\bm{\alpha},\beta)$ ($n\times K$), and $\mathrm{V}_{22}(\bm{\alpha},\beta)$ ($K\times K$) by $[\mathbf{V}_{11}]_{ij}=p_{ij}(1-p_{ij})$ for $i\neq j$ and $[\mathbf{V}_{11}]_{ii}=\sum_{j\neq i}p_{ij}(1-p_{ij})$; the $i$th row of $\mathbf{V}_{12}$ is $\sum_{j\neq i}p_{ij}(1-p_{ij})x_{ij}^{\top}$; and $\mathrm{V}_{22}=\sum_{i<j}p_{ij}(1-p_{ij})x_{ij}x_{ij}^{\top}$.

\paragraph{Hessian of the log-likelihood.} The Hessian blocks $\mathbf{H}_{11}(\bm{\alpha},\beta)$, $\mathbf{H}_{12}(\bm{\alpha},\beta)$, and $\mathrm{H}_{22}(\bm{\alpha},\beta)$ are the sub-blocks of $\nabla^2\ell_n(\bm{\alpha},\beta)$ partitioned conformably with $(\bm{\alpha},\beta)$.

\paragraph{Bias and variance of $\hat\beta$.} For each $i\in\mathcal{I}_n$, define the symmetric matrix $\mathbf{M}_{i,n}(\bm{\alpha},\beta)$ with entries
\begin{align}
(\mathbf{M}_{i,n})_{ij}&=-(p_{ij}^{(1,1,0)}+p_{ji}^{(2,0,0)})\ (j\neq i),\
(\mathbf{M}_{i,n})_{ii}=-\sum_{j\neq i}p_{ij}^{(2,0,0)},\nonumber\\
(\mathbf{M}_{i,n})_{ll}&=-p_{li}^{(2,0,0)}\ (l\neq i),\text{ and }
(\mathbf{M}_{i,n})_{ab}=0\ \text{otherwise}.\label{eq:M_in}
\end{align}
Let $\varphi_{k}(\bm{\alpha},\beta)^\top \coloneqq e_k^\top\mathbf{J}_{21}\mathbf{J}_{11}^{-1}$ and $\mathbf{G}_{k}(\bm{\alpha},\beta)\coloneqq \sum_{i=1}^n \varphi_{k,i}\mathbf{M}_{i,n}$, and define $\mathbf{R}_{k}(\bm{\alpha},\beta)$ with entries
\begin{equation}
\begin{aligned}
\left(\mathbf{R}_{k}\right)_{ij} & =p_{ij}^{(1,1,0)}x_{ij,k},\ 1\leq i\neq j\leq n,\\
\left(\mathbf{R}_{k}\right)_{ii} & =\sum_{j\neq i}p_{ij}^{(2,0,0)}x_{ij,k},\ i\in\mathcal{I}_n.
\end{aligned}
\label{bias_r}
\end{equation}
The bias components are
\begin{equation}
B_{k0}=B_{k0}^{(1)}+B_{k0}^{(2)}\label{eq:bias_B0}
\end{equation}
with
\begin{equation*}
\begin{aligned}
B_{k0}^{(1)} & \coloneqq \lim_{n\to\infty}\tfrac{1}{2\sqrt{N}}\mathrm{Tr}[\mathbf{J}_{11,0}^{-1}\mathbf{V}_{11,0}(\mathbf{J}_{11,0}^{-1})^\top\mathbf{G}_{k,0}],\\
B_{k0}^{(2)} & \coloneqq \lim_{n\to\infty}\tfrac{1}{2\sqrt{N}}\mathrm{Tr}[\mathbf{J}_{11,0}^{-1}\mathbf{V}_{11,0}(\mathbf{J}_{11,0}^{-1})^\top\mathbf{R}_{k,0}],
\end{aligned}
\end{equation*}
and $B_{0}=(B_{10},\dots,B_{K0})^{\top}$. The sandwich-form asymptotic variance of $\hat\beta$ is
\begin{equation}
\Omega_{0} \coloneqq \lim_{n\to\infty}N^{-1}\mathrm{J}_{0}^{-1}\left[
\begin{aligned}
&\mathrm{V}_{22,0}+\mathbf{J}_{21,0}\mathbf{J}_{11,0}^{-1}\mathbf{V}_{11,0}(\mathbf{J}_{21,0}\mathbf{J}_{11,0}^{-1})^{\top}\\
&-\mathbf{J}_{21,0}\mathbf{J}_{11,0}^{-1}\mathbf{V}_{12,0}-(\mathbf{J}_{21,0}\mathbf{J}_{11,0}^{-1}\mathbf{V}_{12,0})^{\top}
\end{aligned}\right](\mathrm{J}_{0}^{-1})^{\top}.\label{limit_var_mm}
\end{equation}

\paragraph{Bias and variance of $\hat\delta$.} Define
\begin{align*}
\Delta_{\beta}(\bm{\alpha},\beta) &\coloneqq \frac{1}{N}\sum_{i=1}^{n}\sum_{j>i}\frac{\partial\Delta_{ij}}{\partial\beta}(\alpha_{i},\alpha_{j},\beta),\quad
\Delta_{\bm{\alpha}}(\bm{\alpha},\beta) \coloneqq \frac{1}{N}\Big(\sum_{j\neq1}\frac{\partial\Delta_{1j}}{\partial\alpha_{1}}, \dots, \sum_{j\neq n}\frac{\partial\Delta_{nj}}{\partial\alpha_{n}}\Big)^\top,
\end{align*}
and the second-derivative matrices $\mathbf{R}_{k}^{\mu}$ and $\mathbf{G}_{k}^{\mu}$ analogously to $\mathbf{R}_{k}$ and $\mathbf{G}_{k}$ but with $\Delta_{ij,k}$ replacing $p_{ij}x_{ij,k}$. The bias terms are
\begin{equation}
\begin{aligned}
B_{\alpha,k} &\coloneqq \lim_{n\to\infty}\frac{1}{2\sqrt{N}}\mathrm{Tr}\left[\mathbf{J}_{11,0}^{-1}\mathbf{V}_{11,0}\left(\mathbf{J}_{11,0}^{-1}\right)^{\top}(\mathbf{R}_{k,0}^{\mu}+\mathbf{G}_{k,0}^\mu)\right],\\
B_{\beta} &\coloneqq \lim_{n\to\infty}(\Delta_{\beta,0}^{\top}-\Delta_{\bm{\alpha},0}^{\top}\mathbf{J}_{11,0}^{-1}\mathbf{J}_{12,0})\mathrm{J}_{0}^{-1}B_{0},
\end{aligned}\label{eq:b_alpha_b_beta}
\end{equation}
and $\Sigma_{\Delta}$ is the asymptotic variance of $\sqrt{N}(\hat\delta-\bar\Delta_n)$, given by
\begin{equation}
\Sigma_{\Delta} = \lim_{n\to\infty}\frac{1}{N}\left\{
\begin{aligned}
&\mathbf{A}_{0}\mathrm{V}_{22,0}\mathbf{A}_{0}^{\top}+\mathbf{C}_{0}\mathbf{J}_{11,0}^{-1}\mathbf{V}_{11,0}(\mathbf{C}_{0}\mathbf{J}_{11,0}^{-1})^{\top}\\
&-\mathbf{C}_{0}\mathbf{J}_{11,0}^{-1}\mathbf{V}_{12,0}\mathbf{A}_{0}^{\top}-(\mathbf{C}_{0}\mathbf{J}_{11,0}^{-1}\mathbf{V}_{12,0}\mathbf{A}_{0}^{\top})^{\top}
\end{aligned}\right\},\label{eq:sigma_delta}
\end{equation}
where $\mathbf{A}_{0}\coloneqq(\Delta_{\beta,0}^{\top}-\Delta_{\bm{\alpha},0}^{\top}\mathbf{J}_{11,0}^{-1}\mathbf{J}_{12,0})\mathrm{J}_{0}^{-1}$ and $\mathbf{C}_{0}\coloneqq\mathbf{A}_{0}\mathbf{J}_{21,0}-N\Delta_{\bm{\alpha},0}^{\top}$.

\section{Supporting Lemmas}\label{sec:supporting_lemmas}

We state auxiliary results used in the proofs below. Their proofs are collected in the Supplemental Material.

\begin{lem}[Inverse approximation {\citep{yan2019approximating}}]
\label{lem:appro_inverse} Suppose an $n\times n$ matrix $\mathbf{A}=(a_{ij})_{n\times n}$
is invertible with all entries positive and $a_{ii}\geq\sum_{j\neq i}a_{ji}.$
Let $\mathbf{B}=[\mathrm{diag}(a_{11},\dots,a_{nn})]^{-1}$,
$\Delta_{i}=a_{ii}-\sum_{j\neq i}a_{ji}$, $M=\max\{\max_{i\neq j}a_{ij},\max_{i}\Delta_{i}\}$,
and $m=\min_{i\neq j}a_{ij}$. If $M\asymp1$ and $m\asymp1$, then
$\Vert\mathbf{A}^{-1}-\mathbf{B}\Vert_{\max}=O(n^{-2})$.
\end{lem}

\begin{lem}[Diagonal approximation for $\mathbf{I}_{11}^{-1}$]
\label{lem:I11_diag_approx}
If Assumptions~\ref{assu:model}--\ref{assu:res_order} hold, then uniformly over $(\bm{\alpha},\beta)\in\mathbb{A}\times\mathbb{B}$,
$\mathbf{D}_{ii}(\bm{\alpha},\beta)\asymp n$, $\lVert\mathbf{I}_{11}^{-1}(\bm{\alpha},\beta)\rVert_{\infty}=O(n^{-1})$, and $\lVert\mathbf{I}_{11}^{-1}(\bm{\alpha},\beta)-\mathbf{D}^{-1}(\bm{\alpha},\beta)\rVert_{\mathrm{max}}=O(n^{-2})$.
\end{lem}

\begin{lem}[Profiling weight bounds]\label{prop:assu5_primitive}
If Assumptions~\ref{assu:model}--\ref{assu:res_order} hold, then uniformly over $(\bm{\alpha},\beta)\in\mathbb{A}\times\mathbb{B}$ and all $k$,
$\sup_{i}|w_{ki}|=O(1)$, $\sup_{i}\|\partial w_{ki}/\partial\beta\|_{2}=O(1)$, $\sup_{i}|\partial w_{ki}/\partial\alpha_{i}|=O(1)$, and $\sup_{i\neq j}|\partial w_{ki}/\partial\alpha_{j}|=O(n^{-1})$.
\end{lem}

\begin{lem}[Well-definedness of limiting objects]\label{lem:limits_exist}
If Assumptions~\ref{assu:model}--\ref{assu:mm_iden} hold, then the limits $\mathrm{J}_0$, $\Omega_0$, and $B_0$ exist as finite quantities, $\mathrm{J}_0$ is invertible, and $\Omega_0$ is positive definite.
If Assumptions~\ref{assu:model}--\ref{assu:res_order} hold, then the analogous limits $\mathrm{I}_0$ and $b_0$ satisfy the same properties, with $\mathrm{I}_0$ positive definite.
\end{lem}

We say an $n\times n$ matrix $\mathbf{A}$ belongs to the class $\mathcal{G}_{n}(\delta)$ if $\Vert\mathbf{A}\Vert_{1}\leq1$ and, for each $1\leq i\neq j\leq n$, $\mathbf{A}_{ii}\ge\delta$ and $\mathbf{A}_{ij}\le-\delta/(n-1)$.

\begin{lem}[Matrix contraction]
\label{matrix_inequality} If $\mathbf{A},\mathbf{B}\in\mathcal{G}_{n}(\delta)$, then
\begin{equation*}
\Vert\mathbf{A}\mathbf{B}\Vert_{1}\leq1-2(n-2)(n-1)^{-1}\delta^{2}.
\end{equation*}
\end{lem}

\begin{lem}[Deviation bound]
\label{lem:dev_bound} If Assumptions~\ref{assu:model}, \ref{assu:compact_support_and_sampling}, and \ref{assu:ass_f} hold, then for any bounded array $(\lambda_{ij})$,
$\Pr\left\{ \max_{i}(n-1)^{-1}\left|\sum_{j\neq i}\lambda_{ij}(y_{ij}-p_{ij,0})\right|>C_{1}\sqrt{6\log n/(n-1)}\right\} \leq2n^{-2}$.
\end{lem}

Define $\bm{\eta}_{0}=\sum_{k=1}^{n}(\hat{\alpha}_{k}(\beta_{0})-\alpha_{k0})\frac{\partial\mathbf{J}_{11}(\bm{\alpha}_{0},\beta_{0})}{\partial\alpha_{k}}[\widehat{\bm{\alpha}}(\beta_{0})-\bm{\alpha}_{0}]$, and let $\bm{\rho}_{0}$ denote the third-order Taylor remainder of $\mathbf{m}_{1}(\widehat{\bm{\alpha}}(\beta_{0}),\beta_{0})$ around $\bm{\alpha}_{0}$, defined explicitly in (\ref{eq:rho0_def}) of the Supplemental Material.

\begin{lem}[Higher-order expansion of $\widehat{\bm\alpha}$]
\label{yi_bound} If Assumptions~\ref{assu:model}, \ref{assu:compact_support_and_sampling}, and \ref{assu:ass_f} hold, then:
\begin{enumerate}[label=(\alph*), font=\upshape]
  \item \label{yi-bound-a} $\Vert\widehat{\bm{\alpha}}(\beta_{0})-\bm{\alpha}_{0}\Vert_{\infty}=O_{p}\left(\sqrt{\log n/n}\right)$.
  \item \label{yi-bound-b} $\mathrm{plim}_{n\to\infty}\frac{1}{2\sqrt{N}}\mathbf{J}_{21,0}\mathbf{J}_{11,0}^{-1}\bm{\eta}_{0}=B_{0}^{(1)}$.
  \item \label{yi-bound-c} $\left\Vert \widehat{\bm{\alpha}}(\beta_{0})-\bm{\alpha}_{0}+\mathbf{J}_{11,0}^{-1}\mathbf{m}_{1,0}+\frac{1}{2}\mathbf{J}_{11,0}^{-1}\bm{\eta}_{0}\right\Vert _{\infty}=O_{p}\left((\log n)^{3/2}/n^{3/2}\right)$.
\end{enumerate}
\end{lem}

Define $S_{n}(\beta)\coloneqq N^{-1}m_{2}(\widehat{\bm\alpha}(\beta),\beta)$. Recall $\bar{S}_{n}(\beta)$ from \eqref{eq:def_S_n_bar}.

\begin{lem}[Uniform convergence]
\label{lem:bound_sn_s} If Assumptions~\ref{assu:model}--\ref{assu:mm_iden} hold, then
\[
\sup_{\beta\in\mathbb{B}}\left\Vert S_{n}(\beta)-\bar{S}_{n}(\beta)\right\Vert _{2}\stackrel{p}{\to}0.
\]
\end{lem}

\begin{lem}[One-step score bounds]
\label{lem:one_step_bound} If Assumptions~\ref{assu:model}--\ref{assu:res_order} hold, then:
\begin{enumerate}[label=(\alph*), font=\upshape]
\item $N^{-1/2}\nabla_{\bm{\alpha}^{\top}}s_{n}(\widehat{\bm{\alpha}}(\beta_{0}),\beta_{0})[\widehat{\bm{\alpha}}(\beta_{0})-\bm{\alpha}_{0}]\stackrel{p}{\to}b_{0}$.
\item $N^{-1}\nabla_{\beta^{\top}}s_{n}(\widehat{\bm{\alpha}}(\bar{\beta}),\bar{\beta})+N^{-1}\nabla_{\bm{\alpha}^{\top}}s_{n}(\widehat{\bm{\alpha}}(\bar{\beta}),\bar{\beta})\partial\widehat{\bm{\alpha}}(\bar{\beta})/\partial\beta^{\top}+\mathrm{I}_{0}\stackrel{p}{\to}0$,
\end{enumerate}
for any $\bar\beta$ between $\beta_0$ and its $\sqrt{N}$-consistent estimator.
\end{lem}

\section{Proofs}\label{sec:proofs_main}

This appendix contains the proofs of Lemma~\ref{lem:alpha_conv_rate} and Theorems~\ref{thm:jmm_alpha_beta_consistency}--\ref{thm:norm_ape}. Proofs of the supporting lemmas stated in Appendix~\ref{sec:supporting_lemmas} are deferred to Section~\ref{sec:sm_lemma_proofs} of the Supplemental Material.

\subsection{Proof of Lemma~\ref{lem:alpha_conv_rate}}

\begin{proof}[Proof of Lemma~\ref{lem:alpha_conv_rate}]
First, suppose a solution to $\mathbf{m}_{1}(\bm{\alpha},\beta)=0$
exists. Let $\mathbf{G}^{\circ}(\bm{\alpha},\widehat{\bm{\alpha}})$ be the
matrix whose $(i,j)$th element is 
\[
[\mathbf{G}^{\circ}(\bm{\alpha},\widehat{\bm{\alpha}})]_{ij}=\int_{0}^{1}\frac{\partial r_{i}}{\partial\alpha_{j}}(t\bm{\alpha}+(1-t)\widehat{\bm{\alpha}})dt.
\]
Then, by an integral form of the mean value theorem, we have
\[
\mathbf{r}(\bm{\alpha})-\mathbf{r}(\widehat{\bm{\alpha}})={\bf G}^{\circ}(\bm{\alpha},\widehat{\bm{\alpha}})(\bm{\alpha}-\widehat{\bm{\alpha}}).
\]
Notice that for $i\neq j,$ $\partial r_{j}/\partial\alpha_{i}=-(n-1)^{-1}p_{ji}^{(0,1,0)}(\bm{\alpha},\beta)<0$;
while for each $i$, $\partial r_{i}/\partial\alpha_{i}=1-(n-1)^{-1}\sum_{j\neq i}p_{ij}^{(1,0,0)}(\bm{\alpha},\beta)>0$.
Moreover, for each $i,$ 
\[
\sum_{j=1}^{n}\left|\frac{\partial r_{j}}{\partial\alpha_{i}}\right|=\frac{\partial r_{i}}{\partial\alpha_{i}}-\sum_{j\neq i}\frac{\partial r_{j}}{\partial\alpha_{i}}\equiv1.
\]
For each $i$ and any $\bm{\alpha},$ this proves $\sum_{j=1}^{n}\lvert[\mathbf{G}^{\circ}(\bm{\alpha},\widehat{\bm{\alpha}})]_{ji}\rvert=1,$
i.e., $\rVert\mathbf{G}^{\circ}(\bm{\alpha},\widehat{\bm{\alpha}})\rVert_{1}=1.$
By Assumptions \ref{assu:compact_support_and_sampling}
and \ref{assu:ass_f}, the derivatives are uniformly bounded and
$\partial r_{j}/\partial\alpha_{i}<0$ for $i\neq j$, while
$\partial r_{i}/\partial\alpha_{i}>0$. By Assumption~\ref{assu:ass_f},
there exists a constant $\delta\in(0,1)$ such that
$[\mathbf{G}^{\circ}(\bm{\alpha},\widehat{\bm{\alpha}})]_{ii}\ge\delta$
and $[\mathbf{G}^{\circ}(\bm{\alpha},\widehat{\bm{\alpha}})]_{ij}\le-\delta/(n-1)$
for all $i\neq j$. Therefore, we have
$\mathbf{G}^{\circ}(\bm{\alpha},\widehat{\bm{\alpha}})\in\mathcal{G}_{n}(\delta)$. 

By the updating algorithm (\ref{eq:iter}), after every two updates,
we have 
\begin{align*}
\Vert\bm{\alpha}^{k+2}(\beta)-\widehat{\bm{\alpha}}(\beta)\Vert_{1}= & \ \Vert\mathbf{r}(\mathbf{r}(\bm{\alpha}^{k}(\beta)))-\mathbf{r}(\mathbf{r}(\widehat{\bm{\alpha}}(\beta)))\Vert_{1}\\
= & \ \Vert\mathbf{G}^{\circ}(\mathbf{r}(\bm{\alpha}^{k}(\beta)),\widehat{\bm{\alpha}}(\beta))(\mathbf{r}(\bm{\alpha}^{k}(\beta))-\widehat{\bm{\alpha}}(\beta))\Vert_{1}\\
= & \ \Vert\mathbf{G}^{\circ}(\mathbf{r}(\bm{\alpha}^{k}(\beta)),\widehat{\bm{\alpha}}(\beta))\mathbf{G}^{\circ}(\bm{\alpha}^{k}(\beta),\widehat{\bm{\alpha}}(\beta))(\bm{\alpha}^{k}(\beta)-\widehat{\bm{\alpha}}(\beta))\Vert_{1}\\
\leq & \ \Vert\mathbf{G}^{\circ}(\mathbf{r}(\bm{\alpha}^{k}(\beta)),\widehat{\bm{\alpha}}(\beta))\mathbf{G}^{\circ}(\bm{\alpha}^{k}(\beta),\widehat{\bm{\alpha}}(\beta))\Vert_{1}\lVert\bm{\alpha}^{k}(\beta)-\widehat{\bm{\alpha}}(\beta)\Vert_{1}\\
\leq & \ \left(1-\frac{2(n-2)}{n-1}\delta^{2}\right)\Vert\bm{\alpha}^{k}(\beta)-\widehat{\bm{\alpha}}(\beta)\Vert_{1},
\end{align*}
where the first equality holds by the fact that $\widehat{\bm{\alpha}}(\beta)=\mathbf{r}(\widehat{\bm{\alpha}}(\beta))$
which implies $\widehat{\bm{\alpha}}(\beta)$ is the fixed point of
the updating function, and the last inequality holds by Lemma~\ref{matrix_inequality}.
We write $\bar{\delta} \coloneqq 1-\frac{2(n-2)}{n-1}\delta^{2}$, and the second
inequality of Lemma~\ref{lem:alpha_conv_rate}
follows.

By this result, $\mathbf{r}(\bm{\alpha})$ is a contraction mapping
for $(\bm{\alpha},\beta)\in\mathbb{A}\times\mathbb{B}$. So, if there
exists a solution $\widehat{\bm{\alpha}}(\beta)\in\mathbb{A}$, the
solution is unique. Now we show the existence of the solution, where
the main technique is adapted from \citet{yan2016asymptotics} and
\citet{yan2019statistical}. Define a sequence of Newton iterations
$\bm{\alpha}^{(k+1)}=\bm{\alpha}^{(k)}-\mathbf{J}_{11}^{-1}(\bm{\alpha}^{(k)},\beta)\mathbf{m}_{1}(\bm{\alpha}^{(k)},\beta)$,
and choose the initial value as $\bm{\alpha}^{(0)}=\bm{\alpha}_{0}$.
Following Proposition A.1 of \citet{yan2016asymptotics}, in a convex
subset $\mathbb{D}\subset\mathbb{A}$ that contains $\bm{\alpha}_{0}$,
to obtain the existence of the solution it is sufficient to establish
three facts: (1) $\mathbf{J}_{11}(\bm{\alpha},\beta)$
is Lipschitz continuous with Lipschitz constant of order $O(n)$,
(2) $\left\Vert \mathbf{J}_{11}^{-1}(\bm{\alpha}_{0},\beta)\right\Vert _{\infty}=O(n^{-1})$,
and (3) $\left\Vert \mathbf{J}_{11}^{-1}(\bm{\alpha}_{0},\beta)\mathbf{m}_{1}(\bm{\alpha}_{0},\beta)\right\Vert _{\infty}=O(\left\Vert \beta-\beta_{0}\right\Vert _{2}).$

For the first fact, we calculate the derivative of $\mathbf{J}_{11}(\bm{\alpha},\beta)$
with respect to $\bm{\alpha}$:
\[
\begin{aligned}
\frac{\partial\mathbf{J}_{11,ij}}{\partial\alpha_{k}} = & -\mathbf{1}\{i=j=k\}\textstyle\sum_{l\neq i}p_{il}^{(2,0,0)}-\mathbf{1}\{i=j\neq k\}p_{ik}^{(1,1,0)}\\
& -\mathbf{1}\{j\neq i=k\}p_{ij}^{(1,1,0)}-\mathbf{1}\{i\neq j=k\}p_{ij}^{(0,2,0)},
\end{aligned}
\]
which implies that $\max_{i}\sum_{j,k}\left|\int_{0}^{1}\frac{\partial\mathbf{J}_{11,ij}(t\bm{\alpha}_{1}+(1-t)\bm{\alpha}_{2})}{\partial\alpha_{k}}dt\right|=O(n).$
Hence $\mathbf{J}_{11}(\bm{\alpha},\beta)$ is Lipschitz continuous
with Lipschitz constant $O(n)$. The second fact is a direct application
of the inverse approximation Lemma~\ref{lem:appro_inverse}. Finally,
the third result follows from
\begin{align*}
 & \ \left\Vert [\mathbf{J}_{11}(\bm{\alpha}_{0},\beta)]^{-1}\mathbf{m}_{1}(\bm{\alpha}_{0},\beta)\right\Vert _{\infty}\\
\leq & \ \left\Vert [\mathbf{J}_{11}(\bm{\alpha}_{0},\beta)]^{-1}\mathbf{m}_{1,0}\right\Vert _{\infty}+\left\Vert [\mathbf{J}_{11}(\bm{\alpha}_{0},\beta)]^{-1}\left[\mathbf{m}_{1}(\bm{\alpha}_{0},\beta)-\mathbf{m}_{1,0}\right]\right\Vert _{\infty}\\
\leq & \ o_{p}(1)+O\left(\left\Vert \beta-\beta_{0}\right\Vert _{2}\right),
\end{align*}
where the first inequality holds by the triangle inequality and
the second inequality is true by Lemma~\ref{lem:dev_bound} and the
Lipschitz continuity of $p(\cdot)$ under Assumption~\ref{assu:ass_f}.
In particular, for any $\beta$ in a sufficiently small neighborhood
of $\beta_{0}$, the right-hand side is $o_{p}(1)$.
Then, by an application of Proposition A.1 of \citet{yan2016asymptotics},
we have $\lim_{k\to\infty}\bm{\alpha}^{(k)}$ exists and the limit
equals $\widehat{\bm{\alpha}}(\beta)$ if $\left\Vert \beta-\beta_{0}\right\Vert _{2}<c$
for some constant $c>0$.\qedhere
\end{proof}

\subsection{Proof of Theorem~\ref{thm:jmm_alpha_beta_consistency}}

\begin{proof}[Proof of Theorem~\ref{thm:jmm_alpha_beta_consistency}]
\textit{Consistency.}\quad
Recall $S_{n}(\beta)$ and $\bar{S}_{n}(\beta)$ from Lemma~\ref{lem:bound_sn_s}.
By Assumption~\ref{assu:mm_iden}, $\hat{\beta}$ and $\beta_{0}$
are unique solutions to $S_{n}(\beta)=0$ and $\bar{S}_{n}(\beta)=0$,
respectively. By Lemma~\ref{lem:bound_sn_s},
\begin{equation}
\left\Vert \bar{S}_{n}(\hat{\beta})\right\Vert _{2}=\left\Vert \bar{S}_{n}(\hat{\beta})-S_{n}(\hat{\beta})\right\Vert _{2}\leq\sup_{\beta\in\mathbb{B}}\left\Vert S_{n}(\beta)-\bar{S}_{n}(\beta)\right\Vert _{2}\stackrel{p}{\to}0.\label{eq:sn_beta_hat_op1}
\end{equation}
Fix $\delta>0$. By Assumption~\ref{assu:mm_iden}, there exists an
$\epsilon>0$ such that $\left\Vert \beta-\beta_{0}\right\Vert _{2}\geq\delta$
implies $\left\Vert \bar{S}_{n}(\beta)\right\Vert _{2}\geq\epsilon$,
hence
\[
\Pr\left(\left\Vert \hat{\beta}-\beta_{0}\right\Vert _{2}\geq\delta\right)\leq\Pr\left(\left\Vert \bar{S}_{n}(\hat{\beta})\right\Vert _{2}\geq\epsilon\right)\leq\Pr\left(\sup_{\beta\in\mathbb{B}}\left\Vert S_{n}(\beta)-\bar{S}_{n}(\beta)\right\Vert _{2}\geq\epsilon\right)\to0
\]
by (\ref{eq:sn_beta_hat_op1}).

We turn to the proof of the convergence of $\widehat{\bm{\alpha}}$
to $\bm{\alpha}_{0}$ in the $\ell_{\infty}$ norm. By the integral
type mean-value theorem, we have 
\begin{align*}
\widehat{\bm{\alpha}}-\bm{\alpha}_{0}= & -\left[\mathbf{J}_{11}^{\circ}(\widehat{\bm{\alpha}},\bm{\alpha}_{0})\right]^{-1}\mathbf{m}_{1}(\bm{\alpha}_{0},\hat{\beta})\\
= & -\left[\mathbf{J}_{11}^{\circ}(\widehat{\bm{\alpha}},\bm{\alpha}_{0})\right]^{-1}\mathbf{m}_{1,0}-\left[\mathbf{J}_{11}^{\circ}(\widehat{\bm{\alpha}},\bm{\alpha}_{0})\right]^{-1}\left(\mathbf{m}_{1}(\bm{\alpha}_{0},\hat{\beta})-\mathbf{m}_{1,0}\right)
\end{align*}
Following the proof of Lemma~\ref{yi_bound}, we have $\Vert[\mathbf{J}_{11}^{\circ}(\widehat{\bm{\alpha}},\bm{\alpha}_{0})]^{-1}\Vert_{\infty}=O(n^{-1})$
and $\Vert\mathbf{m}_{1,0}\Vert_{\infty}=O_{p}(\sqrt{n\log n})$,
hence $\left\Vert \left[\mathbf{J}_{11}^{\circ}(\widehat{\bm{\alpha}},\bm{\alpha}_{0})\right]^{-1}\mathbf{m}_{1,0}\right\Vert _{\infty}\stackrel{p}{\to}0$.
Thus, we only need to show that $O(n^{-1})\cdot\Vert\mathbf{m}_{1}(\bm{\alpha}_{0},\hat{\beta})-\mathbf{m}_{1,0}\Vert_{\infty}\stackrel{p}{\to}0$.
Notice that
\begin{align*}
\Vert\mathbf{m}_{1}(\bm{\alpha}_{0},\hat{\beta})-\mathbf{m}_{1,0}\Vert_{\infty}
& = \max_{1\leq i\leq n}\left|\textstyle\sum_{j\neq i}\left[p_{ij}(\bm{\alpha}_{0},\hat{\beta})-p_{ij,0}\right]\right|\\
& \leq \max_{1\leq i\leq n}\left\Vert \textstyle\sum_{j\neq i}p_{ij}^{(0,0,1)}(\bm{\alpha}_{0},\bar{\beta})x_{ij}\right\Vert _{2}\cdot\Vert\hat{\beta}-\beta_{0}\Vert_{2} = o_{p}(n),
\end{align*}
where we use a Taylor expansion of $p_{ij}(\bm{\alpha}_{0},\beta)$
around $\beta_{0}$ ($\bar{\beta}$ is the mean value which may vary
with $i$) and the fact that $p_{ij}^{(0,0,1)}$ is bounded by
Assumption~\ref{assu:ass_f}.

\textit{Asymptotic normality.}\quad
By Lemma~\ref{lem:limits_exist},
$N^{-1}\mathrm{J}_{n,0}\stackrel{p}{\to}\mathrm{J}_{0}$,
a finite positive-definite matrix. Combined with
$\hat{\beta}\stackrel{p}{\to}\beta_{0}$,
$\lVert\widehat{\bm{\alpha}}-\bm{\alpha}_{0}\rVert_{\infty}=o_{p}(1)$
from the consistency part, and the continuity of
$N^{-1}\mathrm{J}_{n}$ in $(\bm{\alpha},\beta)$, we have
$N^{-1}\mathrm{J}_{n}(\widehat{\bm{\alpha}}(\bar\beta),\bar{\beta})\stackrel{p}{\to}\mathrm{J}_{0}$
for any $\bar{\beta}$ between $\hat{\beta}$ and $\beta_{0}$.

By a first-order Taylor expansion of $m_{n}(\hat{\beta})=m_{2}(\widehat{\bm{\alpha}}(\hat{\beta}),\hat{\beta})$
around $\beta_{0}$, we have 
\[
m_{n}(\hat{\beta})-m_{n}(\beta_{0})=\mathrm{J}_{n}(\widehat{\bm{\alpha}}(\bar\beta),\bar{\beta})(\hat{\beta}-\beta_{0}),
\]
where $\bar{\beta}$ is the mean-value between $\hat{\beta}$ and
$\beta_{0}$. By $m_{n}(\hat{\beta})=0$, we obtain 
\begin{align}
\sqrt{N}(\hat{\beta}-\beta_{0})= & -[N^{-1}\mathrm{J}_{n}(\widehat{\bm{\alpha}}(\bar\beta),\bar{\beta})]^{-1}\frac{1}{\sqrt{N}}m_{2}(\widehat{\bm{\alpha}}(\beta_{0}),\beta_{0})\nonumber \\
= & -\mathrm{J}_{0}^{-1}\left\{ \frac{1}{\sqrt{N}}\sum_{i=1}^{n}\sum_{j>i}\left[y_{ij}-p_{ij}(\widehat{\bm{\alpha}}(\beta_{0}),\beta_{0})\right]x_{ij}\right\} +o_{p}(1)\label{eq:beta_hat_linear_expansion}
\end{align}
by the definition of $\mathrm{J}_{0}$. The bracketed term depends on $\widehat{\bm{\alpha}}(\beta_{0})$ and requires expansion before a CLT applies. Setting $\bm{\zeta}=\widehat{\bm{\alpha}}(\beta_{0})-\bm{\alpha}_{0}$, a third-order Taylor expansion of $m_{2}(\widehat{\bm{\alpha}}(\beta_{0}),\beta_{0})$ in $\bm\alpha$ around $\bm\alpha_{0}$ gives
\begin{align}
 & \ \frac{1}{\sqrt{N}}\sum_{i=1}^{n}\sum_{j>i}\left[y_{ij}-p_{ij}(\widehat{\bm{\alpha}}(\beta_{0}),\beta_{0})\right]x_{ij}\nonumber\\
= & \ \frac{1}{\sqrt{N}}m_{2,0}+\frac{1}{\sqrt{N}}\mathbf{J}_{21,0}\bm{\zeta}
 -\frac{1}{2\sqrt{N}}\sum_{k=1}^{n}\zeta_{k}\sum_{i=1}^{n}\sum_{j>i}\frac{\partial^{2}p_{ij}(\bm{\alpha}_{0},\beta_{0})}{\partial\alpha_{k}\partial\bm{\alpha}^{\top}}\bm{\zeta}\, x_{ij}\nonumber\\
 & \ -\frac{1}{6\sqrt{N}}\sum_{k=1}^{n}\sum_{l=1}^{n}\zeta_{k}\zeta_{l}\sum_{i=1}^{n}\sum_{j>i}\frac{\partial^{3}p_{ij}(\bar{\bm{\alpha}},\beta_{0})}{\partial\alpha_{k}\partial\alpha_{l}\partial\bm{\alpha}^{\top}}\bm{\zeta}\, x_{ij}\nonumber\\
\eqqcolon & \ \left(I\right)+\left(II\right)+\left(III\right)+\left(IV\right).\label{thirdorder_ta}
\end{align}

Since $p_{ij}$ depends only on $\alpha_{i}$ and $\alpha_{j}$,
and $\sup_{i}|\hat{\alpha}_{i}(\beta_{0})-\alpha_{i0}|=O_{p}(\sqrt{(\log n)/n})$ by Lemma~\ref{yi_bound},
the third-order remainder satisfies
$(IV)=O_{p}(N^{-1/2}\cdot N\cdot(\log n/n)^{3/2})=o_{p}(1)$
under Assumptions~\ref{assu:compact_support_and_sampling} and~\ref{assu:ass_f}.

For $(III)$, substituting the asymptotic linear approximation of
$\widehat{\bm{\alpha}}(\beta_{0})-\bm{\alpha}_{0}$,
its $k$th entry equals
$-\frac{1}{2\sqrt{N}}\mathrm{Tr}[\mathbf{J}_{11,0}^{-1}\mathbf{V}_{11,0}(\mathbf{J}_{11,0}^{-1})^{\top}\mathbf{R}_{k,0}]+o_{p}(1)$,
where $\mathbf{R}_{k}$ for $k=1,\dots,K$ is defined in (\ref{bias_r}).
Recalling the definition of $B_{k0}^{(2)}$ in (\ref{eq:bias_B0}), it follows
that $(III)=-B_{0}^{(2)}+o_{p}(1)$, where
$B_{0}^{(2)}\coloneqq(B_{10}^{(2)},\dots,B_{K0}^{(2)})^{\top}$.

We directly combine the rest of terms, $(I)\text{ and }(II)$, to obtain
\begin{align}
 & \sqrt{N}(\hat{\beta}-\beta_{0})-\mathrm{J}_{0}^{-1}B_{0}\nonumber\\
= & -\mathrm{J}_{0}^{-1}\left\{ \frac{1}{\sqrt{N}}m_{2,0}+\frac{1}{\sqrt{N}}\mathbf{J}_{21,0}[\widehat{\bm{\alpha}}(\beta_{0})-\bm{\alpha}_{0}]\right\}+\mathrm{J}_{0}^{-1}B_{0}^{(2)}-\mathrm{J}_{0}^{-1}B_{0}+o_{p}(1)\nonumber\\
= & -\mathrm{J}_{0}^{-1}\left\{ \frac{1}{\sqrt{N}}m_{2,0}-\frac{1}{\sqrt{N}}\mathbf{J}_{21,0}\mathbf{J}_{11,0}^{-1}\mathbf{m}_{1,0}\right\}+o_{p}(1),
\label{eq:asym_expansion}
\end{align}
where we use the facts that $B_{0}^{(1)}\coloneqq(B_{10}^{(1)},\dots,B_{K0}^{(1)})^{\top}$;
the third-order expansion $\widehat{\bm{\alpha}}(\beta_{0})-\bm{\alpha}_{0}=-\mathbf{J}_{11,0}^{-1}\mathbf{m}_{1,0}-\frac{1}{2}\mathbf{J}_{11,0}^{-1}\bm{\eta}_{0}-\frac{1}{6}\mathbf{J}_{11,0}^{-1}\bm{\rho}_{0}$ with $\Vert\mathbf{J}_{11,0}^{-1}\bm{\rho}_{0}\Vert_{\infty}=O((\log n)^{3/2}/n^{3/2})$ by Lemma~\ref{yi_bound}\ref{yi-bound-c}; and $\frac{1}{2\sqrt{N}}\mathbf{J}_{21,0}\mathbf{J}_{11,0}^{-1}\bm{\eta}_{0}-B_{0}^{(1)}=o_{p}(1)$ by Lemma~\ref{yi_bound}\ref{yi-bound-b}.
To apply the CLT, we verify the Lindeberg condition. Define
\begin{align}
\frac{1}{\sqrt{N}}\sum_{i=1}^{n}\sum_{j>i}\xi_{ij} \coloneqq & -\mathrm{J}_{0}^{-1}\frac{1}{\sqrt{N}}\left\{ \sum_{i=1}^{n}\sum_{j>i}(y_{ij}-p_{ij,0})x_{ij}-\mathbf{J}_{21,0}\mathbf{J}_{11,0}^{-1}\mathbf{m}_{1,0}\right\} \nonumber\\
= & -\mathrm{J}_{0}^{-1}\frac{1}{\sqrt{N}}\sum_{i=1}^{n}\sum_{j>i}(y_{ij}-p_{ij,0})\tilde{x}_{ij},
\label{eq:weight_sum_y_p}
\end{align}
where $\tilde{x}_{ij}$ collects the two multipliers of $(y_{ij}-p_{ij,0})$ from the definitions in Appendix~\ref{sec:appendix_defs} and is uniformly bounded by Assumptions~\ref{assu:compact_support_and_sampling} and~\ref{assu:ass_f}. Since the centered Bernoulli $y_{ij}-p_{ij,0}$ are independent across dyads and bounded in $[-1,1]$, the Lindeberg condition holds, yielding 
$
\sqrt{N}(\hat{\beta}-\beta_{0})-\mathrm{J}_{0}^{-1}B_{0}\stackrel{d}{\to}\mathcal{N}(0,\Omega_{0}),
$
where $\Omega_{0}$ is defined in (\ref{limit_var_mm}).
\end{proof}

\subsection{Proof of Theorem~\ref{thm:os_norm}}

\begin{proof}[Proof of Theorem~\ref{thm:os_norm}]
By definition, $\hat{\beta}_{\mathrm{OS}}=\hat{\beta}+\mathrm{I}_{n}(\widehat{\bm{\alpha}},\hat{\beta})^{-1}s_{n}(\widehat{\bm{\alpha}},\hat{\beta})$ with the JMM estimator $(\widehat{\bm{\alpha}},\hat{\beta})$. A first-order Taylor expansion of $s_{n}(\widehat{\bm\alpha},\hat\beta)$ around $\beta_{0}$, followed by a first-order expansion of $s_{n}(\widehat{\bm\alpha}(\beta_{0}),\beta_{0})$ around $\bm{\alpha}_{0}$, together with $N^{-1}\mathrm{I}_{n}(\widehat{\bm{\alpha}},\hat{\beta})\stackrel{p}{\to}\mathrm{I}_{0}$, yields
\begin{align}
 & \ \sqrt{N}(\hat{\beta}_{\mathrm{OS}}-\beta_{0})\nonumber\\
= & \ \frac{1}{\sqrt{N}}\mathrm{I}_{0}^{-1}s_{n}(\widehat{\bm{\alpha}}(\beta_{0}),\beta_{0})\nonumber\\
 & \ +\mathrm{I}_{0}^{-1}\left[\frac{1}{N}\nabla_{\beta^{\top}}s_{n}(\widehat{\bm{\alpha}}(\bar{\beta}),\bar{\beta})+\frac{1}{N}\nabla_{\bm{\alpha}^{\top}}s_{n}(\widehat{\bm{\alpha}}(\bar{\beta}),\bar{\beta})\frac{\partial\widehat{\bm{\alpha}}(\bar{\beta})}{\partial\beta^{\top}}+\mathrm{I}_{0}\right]\sqrt{N}(\hat{\beta}-\beta_{0})+o_{p}(1)\nonumber\\
= & \ \frac{1}{\sqrt{N}}\mathrm{I}_{0}^{-1}s_{n,0}+\frac{1}{\sqrt{N}}\mathrm{I}_{0}^{-1}\nabla_{\bm{\alpha}^{\top}}s_{n}(\bar{\bm{\alpha}},\beta_{0})(\widehat{\bm{\alpha}}(\beta_{0})-\bm{\alpha}_{0})\nonumber\\
 & \ +\mathrm{I}_{0}^{-1}\left[\frac{1}{N}\nabla_{\beta^{\top}}s_{n}(\widehat{\bm{\alpha}}(\bar{\beta}),\bar{\beta})+\frac{1}{N}\nabla_{\bm{\alpha}^{\top}}s_{n}(\widehat{\bm{\alpha}}(\bar{\beta}),\bar{\beta})\frac{\partial\widehat{\bm{\alpha}}(\bar{\beta})}{\partial\beta^{\top}}+\mathrm{I}_{0}\right]\sqrt{N}(\hat{\beta}-\beta_{0})+o_{p}(1).
\label{asym_repre_os}
\end{align}
By Lemma~\ref{lem:one_step_bound},
we have $\frac{1}{\sqrt{N}}\nabla_{\bm{\alpha}^{\top}}s_{n}(\bar{\bm{\alpha}},\beta_{0})(\widehat{\bm{\alpha}}(\beta_{0})-\bm{\alpha}_{0})\stackrel{p}{\to}b_{0}$ and
\[
\frac{1}{N}\nabla_{\beta^{\top}}s_{n}(\widehat{\bm{\alpha}}(\bar{\beta}),\bar{\beta})+\frac{1}{N}\nabla_{\bm{\alpha}^{\top}}s_{n}(\widehat{\bm{\alpha}}(\bar{\beta}),\bar{\beta})\frac{\partial\widehat{\bm{\alpha}}(\bar{\beta})}{\partial\beta^{\top}}+\mathrm{I}_{0}=o_{p}(1).
\]
Hence, using the result that $\sqrt{N}(\hat{\beta}-\beta_{0})=O_{p}(1)$
by Theorem~\ref{thm:jmm_alpha_beta_consistency}, we simplify
(\ref{asym_repre_os}) as 
\begin{align}
\sqrt{N}(\hat{\beta}_{\mathrm{OS}}-\beta_{0})= & \ \frac{1}{\sqrt{N}}\mathrm{I}_{0}^{-1}s_{n,0}+\mathrm{I}_{0}^{-1}b_{0}+o_{p}(1)\nonumber \\
= & \ \frac{1}{\sqrt{N}}\mathrm{I}_{0}^{-1}\sum_{i=1}^{n}\sum_{j>i}s_{ij,0}+\mathrm{I}_{0}^{-1}b_{0}+o_{p}(1),\label{eq:asym_repre_os1}
\end{align}
where $s_{ij,0}$ is dyad $(i,j)$'s contribution
to the asymptotic representation. Then, by the Lindeberg--Feller CLT,
as in the proof of Theorem~\ref{thm:jmm_alpha_beta_consistency},
we have the stated asymptotic normality.
\end{proof}

\subsection{Proof of Theorem~\ref{thm:bagging_normality}}

\begin{proof}
We define $\hat{\beta}_{T_{n}}=\frac{1}{T_{n}}\sum_{s=1}^{T_{n}}\hat{\beta}_{\mathrm{OS-SJ}}^{(s)}$
as the average of all possible OS-SJ estimators. Let $\mathcal{F}_{n}$
be the $\sigma$-algebra generated by all observed information. It
is clear that $\hat{\beta}_{T_{n}}$ is $\mathcal{F}_{n}$-measurable.
Let $\mathbb{E}^{*}$ represent the expectation over randomness from
the random splits conditional on $\mathcal{F}_{n}$. Our proof contains two immediate results: (i) $\sqrt{N}(\hat{\beta}_{T_{n}}-\beta_{0})\stackrel{d}{\to}\mathcal{N}(0,\mathrm{I}_{0}^{-1})$
as $n\to\infty$; (ii) $\sqrt{N}(\hat{\beta}_{\mathrm{BG}}-\hat{\beta}_{T_{n}})\stackrel{p}{\to}0$
as $n\to\infty$ and $\tilde{T}_{n}\to\infty$. Theorem~\ref{thm:bagging_normality}
follows by a combination of these two results.

\noindent\textbf{Step (i). }By (\ref{eq:asym_repre_os1}), we have
\[
\sqrt{N}(\hat{\beta}_{\mathrm{OS}}-\beta_{0})=\mathrm{I}_{0}^{-1}b_{0}+\frac{1}{\sqrt{N}}\mathrm{I}_{0}^{-1}\sum_{(i,j)\in\mathcal{I}_{n}\times\mathcal{I}_{n};j>i}s_{ij,0}+\mathcal{R}(\mathbf{y},\mathbf{x},\bm{\alpha}_{0}),
\]
where $\mathcal{R}(\mathbf{y},\mathbf{x},\bm{\alpha}_{0})$ is a residual
term of order $o_{p}(1)$, as shown in the proof of Lemma~\ref{lem:one_step_bound}.
Thus, the one-step estimators based on sub-networks are (using $N/4$ to approximate the sub-network dyad count $\binom{n/2}{2}$, with the $O(n^{-1})$ discrepancy absorbed into $\mathcal{R}$)
\[
\sqrt{N/4}(\hat{\beta}_{\mathrm{OS,1}}^{(t)}-\beta_{0})=\mathrm{I}_{0}^{-1}b_{0}+\frac{1}{\sqrt{N/4}}\mathrm{I}_{0}^{-1}\sum_{(i,j)\in\mathcal{I}_{1,n}^{(t)}\times\mathcal{I}_{1,n}^{(t)};j>i}s_{ij,0}+\mathcal{R}(\mathbf{y}_{1}^{(t)},\mathbf{x}_{1}^{(t)},\bm{\alpha}_{0,1}^{(t)}),
\]
\[
\sqrt{N/4}(\hat{\beta}_{\mathrm{OS,2}}^{(t)}-\beta_{0})=\mathrm{I}_{0}^{-1}b_{0}+\frac{1}{\sqrt{N/4}}\mathrm{I}_{0}^{-1}\sum_{(i,j)\in\mathcal{I}_{2,n}^{(t)}\times\mathcal{I}_{2,n}^{(t)};j>i}s_{ij,0}+\mathcal{R}(\mathbf{y}_{2}^{(t)},\mathbf{x}_{2}^{(t)},\bm{\alpha}_{0,2}^{(t)}),
\]
where $\bm{\alpha}_{0,1}^{(t)}$ is the sub-vector of $\bm{\alpha}_{0}$
indexed by $\mathcal{I}_{1,n}^{(t)}$ and similarly for $\bm{\alpha}_{0,2}^{(t)}$.
Hence, we have
\begin{equation}
\sqrt{N}(\hat{\beta}_{\mathrm{OS-SJ}}^{(t)}-\beta_{0})=\frac{2}{\sqrt{N}}\mathrm{I}_{0}^{-1}\sum_{(i,j)\in\mathcal{I}_{1,n}^{(t)}\times\mathcal{I}_{2,n}^{(t)}}s_{ij,0}+\mathcal{R}^{(t)}(\mathbf{y},\mathbf{x},\bm{\alpha}_{0}),\label{eq:os_sj_t}
\end{equation}
with $\mathcal{R}^{(t)}(\mathbf{y},\mathbf{x},\bm{\alpha}_{0}) \coloneqq 2\mathcal{R}(\mathbf{y},\mathbf{x},\bm{\alpha}_{0})-\left[\mathcal{R}(\mathbf{y}_{1}^{(t)},\mathbf{x}_{1}^{(t)},\bm{\alpha}_{0,1}^{(t)})+\mathcal{R}(\mathbf{y}_{2}^{(t)},\mathbf{x}_{2}^{(t)},\bm{\alpha}_{0,2}^{(t)})\right]$,
which is also $o_{p}(1)$ for each fixed $t$. Because the $T_{n}$ equal
splits are averaged uniformly,
\[
\frac{1}{T_{n}}\sum_{t=1}^{T_{n}}\mathcal{R}^{(t)}(\mathbf{y},\mathbf{x},\bm{\alpha}_{0})
=\mathbb{E}^{*}\!\left[\mathcal{R}^{(1)}(\mathbf{y},\mathbf{x},\bm{\alpha}_{0})\mid\mathcal{F}_{n}\right].
\]
By Jensen's inequality,
\begin{align*}
\left\Vert \frac{1}{T_{n}}\sum_{t=1}^{T_{n}}\mathcal{R}^{(t)}\right\Vert _{2}
\leq & \ 2\left\Vert \mathcal{R}(\mathbf{y},\mathbf{x},\bm{\alpha}_{0})\right\Vert _{2}\\
& +\mathbb{E}^{*}\!\left[\left\Vert \mathcal{R}(\mathbf{y}_{1}^{(1)},\mathbf{x}_{1}^{(1)},\bm{\alpha}_{0,1}^{(1)})\right\Vert _{2}+\left\Vert \mathcal{R}(\mathbf{y}_{2}^{(1)},\mathbf{x}_{2}^{(1)},\bm{\alpha}_{0,2}^{(1)})\right\Vert _{2}\mid\mathcal{F}_{n}\right].
\end{align*}
The full-sample remainder is $o_{p}(1)$ by (\ref{eq:asym_repre_os1}),
and the same remainder calculations as in the proof of Lemma~\ref{lem:one_step_bound}
apply to a generic half-sample because the primitive bounds there depend
only on the uniform constants in the assumptions and on the subsample
size, which is $n/2\to\infty$. Therefore,
\[
\frac{1}{T_{n}}\sum_{t=1}^{T_{n}}\mathcal{R}^{(t)}(\mathbf{y},\mathbf{x},\bm{\alpha}_{0})=o_{p}(1).
\]
Then, taking the average of (\ref{eq:os_sj_t}) over all $1\leq t\leq T_{n}$
yields 
\begin{align}
\sqrt{N}(\hat{\beta}_{T_{n}}-\beta_{0})= & \ \frac{2}{\sqrt{N}}\mathrm{I}_{0}^{-1}\frac{1}{T_{n}}\sum_{t=1}^{T_{n}}\sum_{(i,j)\in\mathcal{I}_{1,n}^{(t)}\times\mathcal{I}_{2,n}^{(t)}}s_{ij,0}+\frac{1}{T_{n}}\sum_{t=1}^{T_{n}}\mathcal{R}^{(t)}(\mathbf{y},\mathbf{x},\bm{\alpha}_{0})\nonumber \\
= & \ \frac{2}{\sqrt{N}}\mathrm{I}_{0}^{-1}\sum_{i=1}^{n}\sum_{j\neq i}\frac{\binom{n-2}{n/2-1}}{\binom{n}{n/2}}s_{ij,0}+o_{p}(1)\nonumber \\
= & \ \frac{1}{\sqrt{N}}\times\frac{n}{n-1}\mathrm{I}_{0}^{-1}\sum_{i=1}^{n}\sum_{j>i}s_{ij,0}+o_{p}(1)\stackrel{d}{\to}\mathcal{N}(0,\mathrm{I}_{0}^{-1}),\label{eq:oracle_normality}
\end{align}
where the second equality holds because for each $(i,j),i\neq j$,
there are $\binom{n-2}{n/2-1}$ different splits containing them. This
proves the first result.

\noindent\textbf{Step (ii). }Conditional on $\mathcal{F}_{n}$, random
draws $\hat{\beta}_{\mathrm{OS-SJ}}^{(t)}$ for $t=1,\dots,\tilde{T}_{n}$
are independent and uniformly distributed over $\{\hat{\beta}_{\mathrm{OS-SJ}}^{(s)}\}_{s=1}^{T_{n}}$.
Thus, $\mathbb{E}^{*}[\hat{\beta}_{\mathrm{OS-SJ}}^{(t)}]=\hat{\beta}_{T_{n}}$
and $\mathbb{E}^{*}[\|\hat{\beta}_{\mathrm{OS-SJ}}^{(t)}-\hat{\beta}_{T_{n}}\|_{2}^{2}]=\frac{1}{T_{n}}\sum_{s=1}^{T_{n}}\|\hat{\beta}_{\mathrm{OS-SJ}}^{(s)}-\hat{\beta}_{T_{n}}\|_{2}^{2} \eqqcolon \sigma_{n}^{2}$.
Note that $\sigma_{n}^{2}$ is $\mathcal{F}_{n}$-measurable and $\sigma_{n}^{2}=O_{p}(N^{-1})$
by (\ref{eq:os_sj_t}) and (\ref{eq:oracle_normality}). By
Markov's inequality, for any $\epsilon>0$ 
\[
\Pr(\|\sqrt{N}(\hat{\beta}_{\mathrm{BG}}-\hat{\beta}_{T_{n}})\|_{2}\geq\epsilon|\mathcal{F}_{n})\leq\frac{N\mathbb{E}^{*}[\|\hat{\beta}_{\mathrm{BG}}-\hat{\beta}_{T_{n}}\|_{2}^{2}]}{\epsilon^{2}}=\frac{N\sigma_{n}^{2}}{\tilde{T}_{n}\epsilon^{2}}.
\]
Hence, $\Pr(\|\sqrt{N}(\hat{\beta}_{\mathrm{BG}}-\hat{\beta}_{T_{n}})\|_{2}\geq\epsilon)\leq\mathbb{E}[\tilde{T}_{n}^{-1}\epsilon^{-2}N\sigma_{n}^{2}]=O(\tilde{T}_{n}^{-1})=o(1)$
as $n\to\infty$ and $\tilde{T}_{n}\to\infty$ for any $\epsilon>0$.
This proves the second result.

Combining Step (i) and Step (ii), we have
\[
\sqrt{N}(\hat{\beta}_{\mathrm{BG}}-\beta_{0})=\sqrt{N}(\hat{\beta}_{\mathrm{BG}}-\hat{\beta}_{T_{n}})+\sqrt{N}(\hat{\beta}_{T_{n}}-\beta_{0})=\sqrt{N}(\hat{\beta}_{T_{n}}-\beta_{0})+o_{p}(1)\stackrel{d}{\to}\mathcal{N}(0,\mathrm{I}_{0}^{-1})
\]
 as $\tilde{T}_{n}\to\infty$ and $n\to\infty$.
\end{proof}

\subsection{Proof of Theorem~\ref{thm:norm_ape}}

\begin{proof}[Proof of Theorem~\ref{thm:norm_ape}]
We decompose  $\hat{\delta}-\delta_{0}=\left(\hat{\delta}-\bar{\Delta}_{n}\right)+\left(\bar{\Delta}_{n}-\delta_{0}\right).$
The second term is a U-statistic:
\[
\bar{\Delta}_{n}-\delta_{0}=N^{-1}\sum_{i=1}^{n}\sum_{j>i}\left[\Delta_{ij}(\alpha_{i0},\alpha_{j0},\beta_{0})-\mathbb{E}\Delta_{ij}(\alpha_{i0},\alpha_{j0},\beta_{0})\right]
\]
 with kernel $\Delta_{ij}(\alpha_{i0},\alpha_{j0},\beta_{0})-\mathbb{E}\Delta_{ij}(\alpha_{i0},\alpha_{j0},\beta_{0})$.
By Theorem 12.3 of \citet{van2000asymptotic}, we have
\begin{equation}
\sqrt{n}(\bar{\Delta}_{n}-\delta_{0})\stackrel{d}{\to}\mathcal{N}(0,4\Sigma_{\delta}),\label{eq:ape_norm_ustat}
\end{equation}
where  $\Sigma_{\delta}=\mathrm{Cov}(\Delta_{ij}(\alpha_{i0},\alpha_{j0},\beta_{0}),\Delta_{ik}(\alpha_{i0},\alpha_{k0},\beta_{0}))$.
Next, for the first term, notice that $\widehat{\bm{\alpha}}\equiv\widehat{\bm{\alpha}}(\hat{\beta})$
and we can decompose it as
\begin{align*}
\text{\ensuremath{\sqrt{N}}}\left(\hat{\delta}-\bar{\Delta}_{n}\right)= & \ \frac{1}{\sqrt{N}}\sum_{i=1}^{n}\sum_{j>i}\left[\Delta_{ij}(\hat{\alpha}_{i}(\hat{\beta}),\hat{\alpha}_{j}(\hat{\beta}),\hat{\beta})-\Delta_{ij}(\alpha_{i0},\alpha_{j0},\beta_{0})\right]\\
= & \ \frac{1}{\sqrt{N}}\sum_{i=1}^{n}\sum_{j>i}\left[\Delta_{ij}(\hat{\alpha}_{i}(\hat{\beta}),\hat{\alpha}_{j}(\hat{\beta}),\hat{\beta})-\Delta_{ij}(\hat{\alpha}_{i}(\beta_{0}),\hat{\alpha}_{j}(\beta_{0}),\beta_{0})\right]\\
 & \ +\frac{1}{\sqrt{N}}\sum_{i=1}^{n}\sum_{j>i}\left[\Delta_{ij}(\hat{\alpha}_{i}(\beta_{0}),\hat{\alpha}_{j}(\beta_{0}),\beta_{0})-\Delta_{ij}(\alpha_{i0},\alpha_{j0},\beta_{0})\right]\\
=: & \ U_{1}+U_{2},
\end{align*}
where $U_{1}$ captures the variation from $\hat{\beta}$ and $U_{2}$
captures the variation from $\widehat{\bm{\alpha}}(\beta_{0}).$

For $U_{1}$, a first-order Taylor expansion around $\beta_{0}$ yields
\begin{align}
U_{1}=\  & \frac{1}{\sqrt{N}}\biggl\{ \sum_{i=1}^{n}\sum_{j>i}\frac{\partial\Delta_{ij}}{\partial\beta^{\top}}(\hat{\alpha}_{i}(\bar{\beta}),\hat{\alpha}_{j}(\bar{\beta}),\bar{\beta})\nonumber\\
 &\qquad +\sum_{i=1}^{n}\sum_{j\neq i}\frac{\partial\Delta_{ij}}{\partial\alpha_{i}}(\hat{\alpha}_{i}(\bar{\beta}),\hat{\alpha}_{j}(\bar{\beta}),\bar{\beta})\frac{\partial\hat{\alpha}_{i}}{\partial\beta^{\top}}(\bar{\beta})\biggr\} (\hat{\beta}-\beta_{0})\nonumber \\
=\  & \bigl\{ \Delta_{\beta}(\widehat{\bm{\alpha}}(\bar{\beta}),\bar{\beta})^{\top}-\Delta_{\bm{\alpha}}(\widehat{\bm{\alpha}}(\bar{\beta}),\bar{\beta})^{\top}\mathbf{J}_{11}(\widehat{\bm{\alpha}}(\bar{\beta}),\bar{\beta})^{-1}\mathbf{J}_{12}(\widehat{\bm{\alpha}}(\bar{\beta}),\bar{\beta})\bigr\}\nonumber\\
 &\qquad \times\sqrt{N}(\hat{\beta}-\beta_{0})\nonumber \\
=\  & (\Delta_{\beta,0}^{\top}-\Delta_{\bm{\alpha},0}^{\top}\mathbf{J}_{11,0}^{-1}\mathbf{J}_{12,0})\sqrt{N}(\hat{\beta}-\beta_{0})+o_{p}(1),\label{eq:u1}
\end{align}
where $\bar{\beta}$ lies in the segment between $\hat{\beta}$ and
$\beta_{0}$ and the last equality uses the fact that $\bar{\beta}\stackrel{p}{\to}\beta_{0}$
and $\left\Vert \widehat{\bm{\alpha}}(\bar{\beta})-\bm{\alpha}_{0}\right\Vert _{\infty}\stackrel{p}{\to}0.$

For $U_{2}$, a third-order Taylor expansion yields the following, with
$\bm{\zeta}=\widehat{\bm{\alpha}}(\beta_{0})-\bm{\alpha}_{0}$ as before:
\begin{align}
U_{2}= & \ \sqrt{N}\Delta_{\bm{\alpha}}^{\top}\bm{\zeta}+\frac{1}{2\sqrt{N}}\sum_{k=1}^{n}\zeta_{k}\sum_{i=1}^{n}\sum_{j>i}\frac{\partial^{2}\Delta_{ij}(\bm{\alpha}_{0},\beta_{0})}{\partial\alpha_{k}\partial\bm{\alpha}^{\top}}\bm{\zeta}\nonumber \\
 & \ +\frac{1}{6\sqrt{N}}\sum_{k=1}^{n}\sum_{l=1}^{n}\zeta_{k}\zeta_{l}\sum_{i=1}^{n}\sum_{j>i}\frac{\partial^{3}\Delta_{ij}(\bar{\bm{\alpha}},\beta_{0})}{\partial\alpha_{k}\partial\alpha_{l}\partial\bm{\alpha}^{\top}}\bm{\zeta}.\label{eq:u2}
\end{align}
Recall that $\mathbf{R}_{k}^{\mu}$ and $\mathbf{G}_{k}^{\mu}$ are defined analogously to $\mathbf{R}_{k}$ and $\mathbf{G}_{k}$ in Appendix~\ref{sec:appendix_defs}.
Similarly to the proof of Theorem~\ref{thm:jmm_alpha_beta_consistency},
we can show that the second term of (\ref{eq:u2}) converges in probability
to a bias term $\lim_{n\to\infty}\frac{1}{2\sqrt{N}}\mathrm{Tr}\left[\mathbf{J}_{11,0}^{-1}\mathbf{V}_{11,0}\left(\mathbf{J}_{11,0}^{-1}\right)^{\top}\mathbf{R}_{k,0}^{\mu}\right]$ in (\ref{eq:b_alpha_b_beta}).
The last term of (\ref{eq:u2}) is $o_{p}(1)$ (equivalent to the
limit of part (IV) of (\ref{thirdorder_ta})).
By Lemma~\ref{yi_bound}\ref{yi-bound-b},
$$\sqrt{N}e_k^\top \Delta_{\bm\alpha,0}^\top\mathbf{J}_{11,0}^{-1}\bm{\eta}_{0}\stackrel{p}{\to} \lim_{n\to\infty}\frac{1}{2\sqrt{N}}\mathrm{Tr}[\mathbf{J}_{11,0}^{-1}\mathbf{V}_{11,0}\mathbf{J}_{11,0}^{-1}\mathbf{G}_{k,0}^\mu].$$
Additionally, from the
proof of Theorem~\ref{thm:jmm_alpha_beta_consistency},
we have
\begin{equation}
\sqrt{N}(\hat{\beta}-\beta_{0})=-\mathrm{J}_{0}^{-1}\big\{ \frac{1}{\sqrt{N}}m_{2,0}-\frac{1}{\sqrt{N}}\mathbf{J}_{21,0}\mathbf{J}_{11,0}^{-1}\mathbf{m}_{1,0}-B_0\big\} +o_{p}(1)\label{eq:b61}.
\end{equation}
By Lemma~\ref{yi_bound}\ref{yi-bound-c},
\begin{equation}
\left\Vert \bm{\zeta}+\mathbf{J}_{11,0}^{-1}\mathbf{m}_{1,0}+\frac{1}{2}\mathbf{J}_{11,0}^{-1}\bm{\eta}_{0}\right\Vert _{\infty}=O_{p}((\log n)^{3/2}/n^{3/2}).\label{eq:b62}
\end{equation}
Substituting (\ref{eq:b61}) and (\ref{eq:b62}) into (\ref{eq:u1})
and (\ref{eq:u2}) respectively, we have
\begin{align*}
 & \ \sqrt{N}(\hat{\delta}-\bar{\Delta}_{n})-B_{\beta}-B_{\alpha}\\
= & \ -(\Delta_{\beta,0}^{\top}-\Delta_{\bm{\alpha},0}^{\top}\mathbf{J}_{11,0}^{-1}\mathbf{J}_{12,0})\mathrm{J}_{0}^{-1}\left\{ \frac{1}{\sqrt{N}}m_{2,0}-\frac{1}{\sqrt{N}}\mathbf{J}_{21,0}\mathbf{J}_{11,0}^{-1}\mathbf{m}_{1,0}\right\} \\
 & \ -\sqrt{N}\Delta_{\bm{\alpha},0}^{\top}\mathbf{J}_{11,0}^{-1}\mathbf{m}_{1,0}+o_{p}(1)\\
= & \ -\frac{1}{\sqrt{N}}(\Delta_{\beta,0}^{\top}-\Delta_{\bm{\alpha},0}^{\top}\mathbf{J}_{11,0}^{-1}\mathbf{J}_{12,0})\mathrm{J}_{0}^{-1}m_{2,0}\\
 & \ +\frac{1}{\sqrt{N}}\left[(\Delta_{\beta,0}^{\top}-\Delta_{\bm{\alpha},0}^{\top}\mathbf{J}_{11,0}^{-1}\mathbf{J}_{12,0})\mathrm{J}_{0}^{-1}\mathbf{J}_{21,0}-N\Delta_{\bm{\alpha},0}^{\top}\right]\mathbf{J}_{11,0}^{-1}\mathbf{m}_{1,0}+o_{p}(1),
\end{align*}
where $B_{\beta}$ is defined in (\ref{eq:b_alpha_b_beta}).

Finally, by the Lindeberg--Feller CLT, we have
\begin{equation}
\sqrt{N}(\hat{\delta}-\bar{\Delta}_{n})-B_{\beta}-B_{\alpha}\stackrel{d}{\to}\mathcal{N}(0,\Sigma_{\Delta}).\label{eq:ape_uncon_norm}
\end{equation}
Combining (\ref{eq:ape_norm_ustat}), (\ref{eq:ape_uncon_norm}),
and the fact that $\hat{\delta}-\bar{\Delta}_{n}$ is uncorrelated
with $\bar{\Delta}_{n}-\delta_{0}$ asymptotically, we have
\[
\left(\frac{\Sigma_{\Delta}}{N}+\frac{4\Sigma_{\delta}}{n}\right)^{-1/2}\left(\hat{\delta}-\delta_{0}-\frac{1}{\sqrt{N}}B_{\beta}-\frac{1}{\sqrt{N}}B_{\alpha}\right)\stackrel{d}{\to}\mathcal{N}(0,I_{K}).
\]

Since the asymptotic normality of the plug-in estimator $\hat{\delta}$
has already been established, the asymptotic normality of $\hat{\delta}_{\mathrm{BG}}$
follows by an argument analogous to the proof of Theorem~\ref{thm:bagging_normality},
and is therefore omitted for brevity.
\end{proof}

\putbib
\end{bibunit}

\newpage{}

\onehalfspacing \normalsize

\renewcommand{\theequation}{A.\arabic{equation}}
\renewcommand{\thethm}{A.\arabic{thm}}
\renewcommand{\thelem}{A.\arabic{lem}}
\renewcommand{\theprop}{A.\arabic{prop}}
\renewcommand{\therem}{A.\arabic{rem}}
\setcounter{equation}{0}
\setcounter{thm}{0}
\setcounter{lem}{0}
\setcounter{prop}{0}
\setcounter{rem}{0}
\setcounter{page}{1}

\begin{bibunit}

\appendix
\setcounter{section}{0}

\begin{center}
{\Large\bfseries Supplement to ``Bagging the Network''}

\bigskip

{\large Ming Li \qquad Zhentao Shi \qquad Yapeng Zheng\textsuperscript{\textcolor{black}{\normalsize *}}}

\bigskip

\end{center}
\begingroup
\renewcommand{\thefootnote}{}%
\footnotetext{\textsuperscript{\textcolor{black}{*}}Li: Department of Economics and Risk Management Institute, National University of Singapore, \texttt{mli@nus.edu.sg}. Shi (corresponding author): Department of Economics, The Chinese University of Hong Kong, \texttt{zhentao.shi@cuhk.edu.hk}. Zheng: Department of Economics, The Chinese University of Hong Kong, \texttt{yapengzheng@link.cuhk.edu.hk}.}
\endgroup
\setcounter{footnote}{0}

\bigskip

\noindent This supplement contains material that complements the main text. Section~\ref{sec:sm_lemma_proofs} provides proofs of the supporting lemmas stated in the appendix and formalizes the multi-network extension of Remark~\ref{rem:multi-networks} (Section~\ref{subsec:sm_multi_networks}). Section~\ref{sec:sm_misspecification} develops the full analysis of link function misspecification, including the identification assumption, formal theorems, and proofs. Section~\ref{sec:sm_simulations} reports extended Monte Carlo simulation results: the non-concavity challenge of direct MLE, extended NTU results on fixed-effect recovery, APE estimation, link-function misspecification, and sparser networks, and a brief summary of analogous TU results. Section~\ref{sec:sm_data} provides detailed data descriptions and summary statistics for both empirical applications.

\section{Proofs of Supporting Lemmas}\label{sec:sm_lemma_proofs}

We write ``$a_n\asymp b_n$'' to denote $a_n=O(b_n)$ and $b_n=O(a_n)$, and use $C_1,C_2,\dots$ for positive finite constants. All probabilities are conditional on $\bm\alpha$ and $\mathbf{x}$ unless stated otherwise.

\subsection{Inverse Approximation and Diagonal Bounds}
We adapt Theorem 1 of \citet{yan2019approximating} to the NTU framework
to analytically approximate the inverse of the Jacobian matrix $\mathbf{J}_{11}(\bm{\alpha},\beta)$
and bound the approximation errors. Similar techniques have been used
to prove asymptotic normality in network estimation problems; see,
for example, \citet{yan2013central}, \citet{graham2017econometric},
and \citet{yan2019statistical}. We prove that $\mathbf{J}_{11}(\bm{\alpha},\beta)$
is non-singular for $n$ large enough and $\mathbf{J}_{11}^{-1}(\bm{\alpha},\beta)$
is well approximated by a diagonal matrix.
\begin{proof}[Proof of Lemma~\ref{lem:appro_inverse}]
Following \citet{yan2019approximating}, let $I_{n}$ be the $n\times n$ identity matrix and define
$\mathbf{F}=(f_{ij})_{n\times n}=\mathbf{A}^{-1}-\mathbf{B}$, $\mathbf{U}=I_{n}-\mathbf{A}\mathbf{B}$,
$\mathbf{W}=\mathbf{B}\mathbf{U}$, which satisfy
$\mathbf{F}=\mathbf{F}\mathbf{U}+\mathbf{W}$.
Direct calculation gives $u_{ij}=(\delta_{ij}-1)a_{ij}/a_{jj}$ and
$w_{ij}=(\delta_{ij}-1)a_{ij}/(a_{ii}a_{jj})$, so that
$\max(|w_{ij}|,|w_{ij}-w_{ik}|)\leq M/[m^{2}(n-1)^{2}]$
for all $i,j,k$.

The recursion $f_{ij}=\sum_{k}f_{ik}(\delta_{kj}-1)a_{kj}/a_{jj}+w_{ij}$
and the identity $1\equiv\sum_{k\neq\theta}a_{k\theta}/a_{\theta\theta}+\Delta_{\theta}/a_{\theta\theta}$
yield, after the algebraic manipulation detailed in equations (13)--(17) of \citet{yan2019approximating} (obtained by maximizing over all cardinalities $\lambda$ of the index set),
\[
f_{i\theta}-f_{i\xi}\leq\frac{M/[m^{2}(n-1)^{2}]}{C(n,m,M)},
\]
where $f_{i\theta}\coloneqq\max_{k}f_{ik}$, $f_{i\xi}\coloneqq\min_{k}f_{ik}$, and
$C(n,m,M)\asymp 1$ when $m/M\asymp 1$.
Moreover, for each fixed $i$,
\[
\sum_{k=1}^{n}f_{ik}a_{ki}=\sum_{k=1}^{n}\left([\mathbf{A}^{-1}]_{ik}-\frac{\delta_{ik}}{a_{ii}}\right)a_{ki}=1-1=0,
\]
so $f_{i\xi}<0<f_{i\theta}$ and therefore $\max_{k}|f_{ik}|\leq f_{i\theta}-f_{i\xi}$.
We thus obtain $\lVert\mathbf{F}\rVert_{\max}=O(n^{-2}).\qedhere$
\end{proof}
Based on Lemma~\ref{lem:appro_inverse}, we prove that $\mathbf{J}_{11}(\bm{\alpha},\beta)$
is non-singular for $(\bm{\alpha},\beta)\in\mathbb{A}\times\mathbb{B}$
and large $n$.
\begin{lem}
\label{lem:nonsingular_J}If Assumptions \ref{assu:model}, \ref{assu:compact_support_and_sampling}, and \ref{assu:ass_f}
hold, for $n$ large enough, the Jacobian matrix
$\mathbf{J}_{11}(\bm{\alpha},\beta)$ is invertible for all $(\bm{\alpha},\beta)\in\mathbb{A}\times\mathbb{B}$.
\end{lem}
\begin{proof}
We partition $\mathbf{J}_{11}(\bm{\alpha},\beta)$ into a block matrix
as 
\[
\mathbf{J}_{11}(\bm{\alpha},\beta)=\begin{pmatrix}\left[\mathbf{J}_{11}(\bm{\alpha},\beta)\right]_{(1:n-1)\times(1:n-1)} & \left[\mathbf{J}_{11}(\bm{\alpha},\beta)\right]_{(1:n-1)\times n}\\
\left[\mathbf{J}_{11}(\bm{\alpha},\beta)\right]_{n\times(1:n-1)} & \left[\mathbf{J}_{11}(\bm{\alpha},\beta)\right]_{nn}
\end{pmatrix},
\]
where the subscript denotes the specific rows/columns that each sub-matrix
includes. Recall that $\left[\mathbf{J}_{11}(\bm{\alpha},\beta)\right]_{ii}=\sum_{j\neq i}\left[\mathbf{J}_{11}(\bm{\alpha},\beta)\right]_{ji}$.
The first sub-matrix $\left[\mathbf{J}_{11}(\bm{\alpha},\beta)\right]_{(1:n-1)\times(1:n-1)}$
is strictly diagonally dominant with all negative entries, hence it
is non-singular. Lemma~\ref{lem:appro_inverse} demonstrates that
its inverse can be approximated by $\mathrm{diag}\left(\left[\mathbf{J}_{11}(\bm{\alpha},\beta)\right]_{11}^{-1},\dots,\left[\mathbf{J}_{11}(\bm{\alpha},\beta)\right]_{n-1n-1}^{-1}\right)$
with maximum entry-wise error of $O(n^{-2})$. Under Assumptions \ref{assu:compact_support_and_sampling}
and \ref{assu:ass_f}, $\left[\mathbf{J}_{11}(\bm{\alpha},\beta)\right]_{ii}\asymp-n,\ \left[\mathbf{J}_{11}(\bm{\alpha},\beta)\right]_{ij}\asymp-1,\ j\neq i$.
Write $\mathbf{J}\coloneqq\mathbf{J}_{11}(\bm{\alpha},\beta)$ and partition it as
$\mathbf{J}=\bigl(\begin{smallmatrix}\mathbf{J}_{1:n-1,1:n-1}&\mathbf{J}_{1:n-1,n}\\\mathbf{J}_{n,1:n-1}&\mathbf{J}_{nn}\end{smallmatrix}\bigr)$.
Then
\begin{align*}
 & \ \mathbf{J}_{n,1:n-1}\mathbf{J}_{1:n-1,1:n-1}^{-1}\mathbf{J}_{1:n-1,n}\\
= & \ \mathbf{J}_{n,1:n-1}\mathrm{diag}\bigl(\mathbf{J}_{11}^{-1},\dots,\mathbf{J}_{n-1\,n-1}^{-1}\bigr)\mathbf{J}_{1:n-1,n}\\
 & \ +O(n^{-2})\times\mathbf{J}_{n,1:n-1}\mathbf{1}\mathbf{1}^{\top}\mathbf{J}_{1:n-1,n}\\
= & \ \sum_{i\neq n}\frac{\mathbf{J}_{ni}\mathbf{J}_{in}}{\sum_{j\neq i}\mathbf{J}_{ji}}
+O(n^{-2})\Bigl(\sum_{i\neq n}\mathbf{J}_{ni}\Bigr)\Bigl(\sum_{i\neq n}\mathbf{J}_{in}\Bigr)=O(1).
\end{align*}
Thus, the Schur complement satisfies
$\mathbf{J}_{nn}-\mathbf{J}_{n,1:n-1}\mathbf{J}_{1:n-1,1:n-1}^{-1}\mathbf{J}_{1:n-1,n}\asymp -n-O(1)\neq0$
for $n$ large enough. By the block-determinant formula,
\[
\mathrm{det}(\mathbf{J})
=\mathrm{det}(\mathbf{J}_{1:n-1,1:n-1})
\times\bigl(\mathbf{J}_{nn}-\mathbf{J}_{n,1:n-1}\mathbf{J}_{1:n-1,1:n-1}^{-1}\mathbf{J}_{1:n-1,n}\bigr)\neq0
\]
for $n$ large enough. Hence, $\mathbf{J}_{11}(\bm{\alpha},\beta)$
is invertible for large $n$.
\end{proof}
For the inverse of $\mathbf{J}_{11}(\bm{\alpha},\beta)$, it is straightforward
to verify that $-\mathbf{J}_{11}(\bm{\alpha},\beta)$ satisfies conditions
in Lemma~\ref{lem:appro_inverse}. Let $\mathbf{T}(\bm{\alpha},\beta)=\left[\mathrm{diag}\left(\mathbf{J}_{11}(\bm{\alpha},\beta)\right)\right]^{-1}$.
Applying Lemma~\ref{lem:appro_inverse} to $-\mathbf{J}_{11}(\bm{\alpha},\beta)$,
we have $\lVert\left[-\mathbf{J}_{11}(\bm{\alpha},\beta)\right]^{-1}+\mathbf{T}(\bm{\alpha},\beta)\lVert_{\mathrm{max}}=O(n^{-2})$
under Assumptions \ref{assu:model}--\ref{assu:ass_f}. All of these results could be applied to $\mathbf{J}_{11}^{\circ}(\hat{\bm{\alpha}},\bm{\alpha})$,
where $\mathbf{T}^{\circ}(\hat{\bm{\alpha}},\bm{\alpha})$ denotes
the diagonal approximation for $\left[\mathbf{J}_{11}^{\circ}(\hat{\bm{\alpha}},\bm{\alpha})\right]^{-1}$. The diagonal approximation technique extends from $\mathbf{J}_{11}$
to the information matrix $\mathbf{I}_{11}$, and combined with dyadic
locality, yields the profiling-weight bounds used in the proof of
Theorem~\ref{thm:os_norm}. 

Let $v_{k}(\bm{\alpha},\beta)\coloneqq\mathbf{I}_{12}(\bm{\alpha},\beta)e_{k}$
denote the $k$th column of $\mathbf{I}_{12}(\bm{\alpha},\beta)$, so
that $w_{k}(\bm{\alpha},\beta)=\mathbf{I}_{11}(\bm{\alpha},\beta)^{-1}v_{k}(\bm{\alpha},\beta)$.

\begin{proof}[Proof of Lemma~\ref{lem:I11_diag_approx}]
By Assumption~\ref{assu:ass_f}, the information kernels satisfy
$[\mathbf{I}_{11}]_{ij}\asymp1$ for $i\neq j$
(cf.\ Appendix~\ref{sec:appendix_defs}),
so $\mathbf{D}_{ii}=[\mathbf{I}_{11}]_{ii}\asymp n$
uniformly. Write $\mathbf{I}_{11}=\mathbf{D}^{1/2}\mathbf{Q}\mathbf{D}^{1/2}$
where $\mathbf{Q}=\mathbf{D}^{-1/2}\mathbf{I}_{11}\mathbf{D}^{-1/2}$.
Then $\mathbf{I}_{11}^{-1}=\mathbf{D}^{-1/2}\mathbf{Q}^{-1}\mathbf{D}^{-1/2}$,
whence
\[
\lVert\mathbf{I}_{11}^{-1}\rVert_{\infty}
\leq\lVert\mathbf{D}^{-1/2}\rVert_{\infty}^{2}\lVert\mathbf{Q}^{-1}\rVert_{\infty}
=O(n^{-1})
\]
by Assumption~\ref{assu:res_order}. For the entrywise bound, note that
$\mathbf{Q}_{ii}=1$ and, for $i\neq j$,
$|\mathbf{Q}_{ij}|=|[\mathbf{I}_{11}]_{ij}|/\sqrt{\mathbf{D}_{ii}\mathbf{D}_{jj}}=O(n^{-1})$.
Using the identity $\mathbf{Q}^{-1}-I_{n}=-(\mathbf{Q}-I_{n})\mathbf{Q}^{-1}$,
\[
\lVert\mathbf{Q}^{-1}-I_{n}\rVert_{\mathrm{max}}
\leq\lVert\mathbf{Q}-I_{n}\rVert_{\mathrm{max}}\cdot\lVert\mathbf{Q}^{-1}\rVert_{1}
=O(n^{-1}).
\]
Consequently,
$\mathbf{I}_{11}^{-1}-\mathbf{D}^{-1}=\mathbf{D}^{-1/2}(\mathbf{Q}^{-1}-I_{n})\mathbf{D}^{-1/2}$,
so
\[
\lVert\mathbf{I}_{11}^{-1}-\mathbf{D}^{-1}\rVert_{\mathrm{max}}
\leq\lVert\mathbf{Q}^{-1}-I_{n}\rVert_{\mathrm{max}}\cdot\max_{i}\mathbf{D}_{ii}^{-1}
=O(n^{-1})\cdot O(n^{-1})=O(n^{-2}).\qedhere
\]
\end{proof}
\begin{lem}[Dyadic locality]\label{lem:dyadic_locality}
If Assumptions~\ref{assu:model}, \ref{assu:compact_support_and_sampling}, and \ref{assu:ass_f} hold, then uniformly over $(\bm{\alpha},\beta)\in\mathbb{A}\times\mathbb{B}$
and all $k\in\{1,\dots,K\}$,
\[
\lVert v_{k}(\bm{\alpha},\beta)\rVert_{\infty}=O(n)\text{ and }
\left\Vert \frac{\partial v_{k}(\bm{\alpha},\beta)}{\partial\beta}\right\Vert _{\infty}=O(n).
\]
Moreover, for every $j\in\{1,\dots,n\}$,
\[
\left|\left[\frac{\partial v_{k}}{\partial\alpha_{j}}\right]_{j}\right|=O(n)\text{ and }
\sup_{i\neq j}\left|\left[\frac{\partial v_{k}}{\partial\alpha_{j}}\right]_{i}\right|=O(1).
\]
Further, for any deterministic vector $u\in\mathbb{R}^{n}$ with
$\lVert u\rVert_{\infty}\leq C$,
\[
\left\Vert \frac{\partial\mathbf{I}_{11}}{\partial\beta}u\right\Vert _{\infty}=O(n),
\]
and, for each $j$,
\[
\left|\left[\frac{\partial\mathbf{I}_{11}}{\partial\alpha_{j}}u\right]_{j}\right|=O(n)
\text{ and }\sup_{i\neq j}\left|\left[\frac{\partial\mathbf{I}_{11}}{\partial\alpha_{j}}u\right]_{i}\right|=O(1).
\]
\end{lem}
\begin{proof}
Since $[v_{k}]_{i}=[\mathbf{I}_{12}]_{ik}=\sum_{l\neq i}\frac{p_{il}^{(1,0,0)}p_{il}^{(0,0,1)}x_{il,k}}{p_{il}(1-p_{il})}$,
Assumption~\ref{assu:ass_f} gives $\lVert v_{k}\rVert_{\infty}=O(n)$.
Differentiating with respect to $\beta$ and using the bound on
third-order derivatives of $p$ yields
$\lVert\partial_{\beta}v_{k}\rVert_{\infty}=O(n)$.

Fix $j$. If $i=j$, every term in $[v_{k}]_{j}=\sum_{l\neq j}(\cdot)$
depends on $\alpha_{j}$, so $|[\partial_{\alpha_{j}}v_{k}]_{j}|\leq\sum_{l\neq j}|\partial_{\alpha_{j}}(\cdot)|=O(n)$.
If $i\neq j$, only the single summand $l=j$ in $[v_{k}]_{i}$ depends
on $\alpha_{j}$, giving $|[\partial_{\alpha_{j}}v_{k}]_{i}|=O(1)$.

For $\partial_{\beta}\mathbf{I}_{11}$, the $i$th component of
$(\partial_{\beta}\mathbf{I}_{11})u$ sums $n$ bounded terms (by
Assumption~\ref{assu:ass_f}) times $\lVert u\rVert_{\infty}$,
yielding $O(n)$. For $\partial_{\alpha_{j}}\mathbf{I}_{11}$, dyadic
locality implies that when $i=j$ both $\partial_{\alpha_{j}}[\mathbf{I}_{11}]_{jj}$
and the off-diagonal derivatives $\partial_{\alpha_{j}}[\mathbf{I}_{11}]_{j\ell}$
for $\ell\neq j$ contribute, giving a total of $O(n)$. When $i\neq j$,
at most two entries ($[\mathbf{I}_{11}]_{ii}$ through the $l=j$ summand,
and $[\mathbf{I}_{11}]_{ij}$) have nonzero $\alpha_{j}$-derivatives,
each bounded by Assumption~\ref{assu:ass_f}, so the contribution is $O(1).\qedhere$
\end{proof}
\begin{proof}[Proof of Lemma~\ref{prop:assu5_primitive}]
Since $w_{k}=\mathbf{I}_{11}^{-1}v_{k}$, Lemmas~\ref{lem:I11_diag_approx}
and~\ref{lem:dyadic_locality} give
$\lVert w_{k}\rVert_{\infty}\leq\lVert\mathbf{I}_{11}^{-1}\rVert_{\infty}\lVert v_{k}\rVert_{\infty}=O(n^{-1})\cdot O(n)=O(1)$.

Differentiating $\mathbf{I}_{11}w_{k}=v_{k}$ with respect to a generic
parameter $\theta\in\{\beta,\alpha_{1},\dots,\alpha_{n}\}$,
\[
\partial_{\theta}w_{k}=\mathbf{I}_{11}^{-1}\underbrace{\left(\partial_{\theta}v_{k}-(\partial_{\theta}\mathbf{I}_{11})w_{k}\right)}_{r_{k,\theta}}.
\]

\textit{Case $\theta=\beta$.} By Lemma~\ref{lem:dyadic_locality}
and $\lVert w_{k}\rVert_{\infty}=O(1)$,
$\lVert r_{k,\beta}\rVert_{\infty}=O(n)$.
Hence $\lVert\partial_{\beta}w_{k}\rVert_{\infty}\leq\lVert\mathbf{I}_{11}^{-1}\rVert_{\infty}\lVert r_{k,\beta}\rVert_{\infty}=O(n^{-1})\cdot O(n)=O(1)$.

\textit{Case $\theta=\alpha_{j}$.} By Lemma~\ref{lem:dyadic_locality}
and $\lVert w_{k}\rVert_{\infty}=O(1)$,
\[
|[r_{k,j}]_{j}|=O(n),\qquad\sup_{i\neq j}|[r_{k,j}]_{i}|=O(1),\qquad\lVert r_{k,j}\rVert_{1}=O(n).
\]
Decomposing via the diagonal approximation in Lemma~\ref{lem:I11_diag_approx},
\[
[\partial_{\alpha_{j}}w_{k}]_{i}
=\mathbf{D}_{ii}^{-1}[r_{k,j}]_{i}
+[(\mathbf{I}_{11}^{-1}-\mathbf{D}^{-1})r_{k,j}]_{i}.
\]
For $i=j$: $\mathbf{D}_{jj}^{-1}|[r_{k,j}]_{j}|+\lVert\mathbf{I}_{11}^{-1}-\mathbf{D}^{-1}\rVert_{\mathrm{max}}\lVert r_{k,j}\rVert_{1}=O(n^{-1})\cdot O(n)+O(n^{-2})\cdot O(n)=O(1)$. For $i\neq j$: $\mathbf{D}_{ii}^{-1}|[r_{k,j}]_{i}|+\lVert\mathbf{I}_{11}^{-1}-\mathbf{D}^{-1}\rVert_{\mathrm{max}}\lVert r_{k,j}\rVert_{1}=O(n^{-1})\cdot O(1)+O(n^{-2})\cdot O(n)=O(n^{-1}).\qedhere$
\end{proof}
\begin{rem}\label{rem:assu5_additive}
Under additive models $p(\alpha_{i},\alpha_{j},x_{ij}^{\top}\beta)=F(\alpha_{i}+\alpha_{j}+x_{ij}^{\top}\beta)$
with $F'>0$, we have $p^{(1,0,0)}=p^{(0,1,0)}=F'$, which gives
$[\mathbf{I}_{11}]_{ij}=[F'(\cdot)]^{2}/[p_{ij}(1-p_{ij})]>0$ for
$i\neq j$ and $[\mathbf{I}_{11}]_{ii}=\sum_{j\neq i}[\mathbf{I}_{11}]_{ji}$. Lemma \ref{lem:appro_inverse} implies that $\|\mathbf{I}_{11}^{-1}-\mathbf{D}^{-1}\|_{\max}=O(n^{-2})$ and in consequence,
$$
\|\mathbf{Q}^{-1}\|_1\leq \|\mathbf{D}^{1/2}\mathbf{D}^{-1}\mathbf{D}^{1/2}\|_1+\|\mathbf{D}^{1/2}(\mathbf{I}_{11}^{-1}-\mathbf{D}^{-1})\mathbf{D}^{1/2}\|_1=O(1).
$$Hence, Assumption~\ref{assu:res_order} is automatically
satisfied for additive logit, probit, and other smooth link functions.
\end{rem}

\subsection{Well-definedness of limiting objects}\label{subsec:sm_limits}

\begin{proof}[Proof of Lemma~\ref{lem:limits_exist}]
Throughout this proof, $p_{ij,0}^{(r_1,r_2,r_3)}= p^{(r_1,r_2,r_3)}(\alpha_{i0},\alpha_{j0},X_{ij}^{\top}\beta_{0})$ per the convention, and $X_{ij}=h(X_i,X_j)$ per Assumption~\ref{assu:iid_sampling}. By Assumptions~\ref{assu:compact_support_and_sampling}--\ref{assu:ass_f}, these quantities are uniformly bounded and $p_{ij,0}^{(1,0,0)}\geq c_{2}>0$.

\smallskip

\noindent\textit{Part (a), existence of $\mathrm{J}_0$ and its invertibility.}
Using the decomposition from the proof of
Theorem~\ref{thm:jmm_alpha_beta_consistency},
\[
N^{-1}\mathrm{J}_{n,0}
=N^{-1}\mathrm{J}_{22,0}
-N^{-1}\mathbf{J}_{21,0}\mathbf{T}\mathbf{J}_{12,0}
-N^{-1}\mathbf{J}_{21,0}(\mathbf{J}_{11,0}^{-1}-\mathbf{T})\mathbf{J}_{12,0},
\]
where $\mathbf{T}\coloneqq[\mathrm{diag}(\mathbf{J}_{11,0})]^{-1}$
satisfies $\lVert\mathbf{J}_{11,0}^{-1}-\mathbf{T}\rVert_{\max}=O(n^{-2})$
by Lemma~\ref{lem:appro_inverse}. We handle the three terms in turn.

\emph{Term 1}: $(N^{-1}\mathrm{J}_{22,0})_{kl}
=-N^{-1}\sum_{i<j}p_{ij,0}^{(0,0,1)}X_{ij,k}X_{ij,l}$
is a U-statistic of order $2$ with bounded symmetric kernel. By the
law of large numbers for U-statistics
\citep[e.g.,][Theorem 12.3]{van2000asymptotic},
\[
(N^{-1}\mathrm{J}_{22,0})_{kl}
\stackrel{p}{\to}-\mathbb{E}\bigl[p^{(0,0,1)}(A_{1},A_{2},h(X_{1},X_{2})^{\top}\beta_{0})\,
h(X_{1},X_{2})_{k}\,h(X_{1},X_{2})_{l}\bigr],
\]
with $(A_1,X_1),(A_2,X_2)\sim\nu$ independent, a finite matrix.

\emph{Term 2}: As shown in the proof of
Theorem~\ref{thm:jmm_alpha_beta_consistency},
$-N^{-1}\mathbf{J}_{21,0}\mathbf{T}\mathbf{J}_{12,0}$ admits the
representation
\[
(-N^{-1}\mathbf{J}_{21,0}\mathbf{T}\mathbf{J}_{12,0})_{kl}
=\frac{2}{n}\sum_{i=1}^{n}\frac{\bar{A}_{i,k,n}\,\bar{B}_{i,l,n}}{\bar{C}_{i,n}},
\]
where $\bar{A}_{i,k,n}\coloneqq(n-1)^{-1}\sum_{j\neq i}p_{ij,0}^{(1,0,0)}X_{ij,k}$,
$\bar{B}_{i,l,n}\coloneqq(n-1)^{-1}\sum_{j\neq i}p_{ij,0}^{(0,0,1)}X_{ij,l}$,
and $\bar{C}_{i,n}\coloneqq(n-1)^{-1}\sum_{j\neq i}p_{ij,0}^{(1,0,0)}$.
Conditional on $(\alpha_{i0},X_{i})$, the quantities
$\{(\alpha_{j0},X_{j})\}_{j\neq i}$ are i.i.d.\ from $\nu$, so by the
conditional LLN
\[
\bar{A}_{i,k,n}\stackrel{p}{\to}
a_{k}(\alpha_{i0},X_{i})\coloneqq
\mathbb{E}\bigl[p^{(1,0,0)}(\alpha_{i0},A',h(X_{i},X')^{\top}\beta_{0})
h(X_{i},X')_{k}\mid\alpha_{i0},X_{i}\bigr],
\]
and similarly $\bar{B}_{i,l,n}\stackrel{p}{\to} b_{l}(\alpha_{i0},X_{i})$
and $\bar{C}_{i,n}\stackrel{p}{\to} c(\alpha_{i0},X_{i})$, with
$c(\alpha_{i0},X_{i})\geq c_{2}>0$ uniformly by
Assumption~\ref{assu:ass_f}. Because all integrands are uniformly
bounded and continuous on the compact
$\mathbb{A}\times\mathbb{X}_0\times\mathbb{A}\times\mathbb{X}_0$, a
uniform law of large numbers (Hoeffding's inequality applied
to the bounded summands combined with a covering argument over the
compact index set) gives the uniform convergence
\[
\max_{1\leq i\leq n}\left|\frac{\bar{A}_{i,k,n}\bar{B}_{i,l,n}}{\bar{C}_{i,n}}
-\frac{a_{k}(\alpha_{i0},X_{i})b_{l}(\alpha_{i0},X_{i})}{c(\alpha_{i0},X_{i})}\right|\stackrel{p}{\to}0.
\]
Combined with the LLN over i.i.d.\ $\{(\alpha_{i0},X_{i})\}$,
\[
(-N^{-1}\mathbf{J}_{21,0}\mathbf{T}\mathbf{J}_{12,0})_{kl}
\stackrel{p}{\to}
2\,\mathbb{E}\!\left[\frac{a_{k}(A,X)b_{l}(A,X)}{c(A,X)}\right],
\]
a finite number since the integrand is bounded above (by
$\sup|a_{k}|\sup|b_{l}|/c_{2}<\infty$).

\emph{Term 3 (residual)}: Write $-\mathbf{J}_{11,0}=\mathbf{D}+\mathbf{O}$
with $\mathbf{D}=\mathrm{diag}(D_{i})$, $D_{i}=\sum_{j\neq i}
p_{ij,0}^{(1,0,0)}\asymp n$, and $\mathbf{O}_{ij}=p_{ij,0}^{(0,1,0)}$
for $i\neq j$, $\mathbf{O}_{ii}=0$. The Neumann expansion yields
$\mathbf{J}_{11,0}^{-1}-\mathbf{T}=\mathbf{D}^{-1}\mathbf{O}\mathbf{D}^{-1}+\mathbf{E}$,
where $\lVert\mathbf{E}\rVert_{\infty}=O(n^{-2})$ by
Lemma~\ref{matrix_inequality}. The $\mathbf{E}$-contribution to
$-N^{-1}\mathbf{J}_{21,0}(\mathbf{J}_{11,0}^{-1}-\mathbf{T})\mathbf{J}_{12,0}$
is $O(n^{-1})=o(1)$ by the same scale accounting as in the proof of
Theorem~\ref{thm:jmm_alpha_beta_consistency}. For the leading term,
using $(\mathbf{J}_{21,0})_{ki}/D_{i}=-\bar{A}_{i,k,n}/\bar{C}_{i,n}$
and $(\mathbf{J}_{12,0})_{jl}/D_{j}=-\bar{B}_{j,l,n}/\bar{C}_{j,n}$,
\begin{align*}
-N^{-1}(\mathbf{J}_{21,0}\mathbf{D}^{-1}\mathbf{O}\mathbf{D}^{-1}\mathbf{J}_{12,0})_{kl}
& = -\frac{1}{N}\sum_{i\neq j}\frac{\bar{A}_{i,k,n}\bar{B}_{j,l,n}}{\bar{C}_{i,n}\bar{C}_{j,n}}\,p_{ij,0}^{(0,1,0)}\\
& \stackrel{p}{\to}-\mathbb{E}\!\left[\frac{a_{k}(A_{1},X_{1})b_{l}(A_{2},X_{2})}{c(A_{1},X_{1})c(A_{2},X_{2})}\,p^{(0,1,0)}_{12}\right],
\end{align*}
a finite number, by the uniform convergence of
$\bar{A},\bar{B},\bar{C}$ and the LLN over i.i.d.\ $\{(\alpha_{i0},X_{i})\}$.
Combining Terms 1--3 establishes the existence of $\mathrm{J}_{0}$
as a finite matrix.
Note from \eqref{eq:def_S_n_bar} and the chain
rule that
$
\nabla_{\beta^{\top}}\bar{S}_{n}(\beta_{0})=N^{-1}\mathrm{J}_{n,0},
$
which is of full rank by Assumption~\ref{assu:mm_iden}. 
Therefore $\mathrm{J}_{0}=\mathrm{plim}_{n\to\infty}N^{-1}\mathrm{J}_{n,0}$ is invertible.

\smallskip

\noindent\textit{Part (a), existence of $\mathrm{I}_0$.}
The proof is identical to that for $\mathrm{J}_{0}$, with
$\mathbf{I}_{11,0},\mathbf{I}_{12,0},\mathrm{I}_{22,0}$ replacing
$-\mathbf{J}_{11,0},-\mathbf{J}_{12,0},-\mathrm{J}_{22,0}$ and
Lemma~\ref{lem:I11_diag_approx} replacing
Lemma~\ref{lem:appro_inverse}. The entrywise scaling and the bounds
in Assumption~\ref{assu:ass_f} carry over unchanged.

\smallskip

\noindent\textit{Part (b), existence of $B_0$ and $b_0$.} The trace expression for $B_{k0}$ in \eqref{eq:bias_B0} is analyzed by substituting $\mathbf{J}_{11,0}^{-1}=\mathbf{T}+(\mathbf{J}_{11,0}^{-1}-\mathbf{T})$ twice and using the $O(n^{-2})$ entrywise residual bound; this reduces the trace to that of $\mathbf{T}\mathbf{V}_{11,0}\mathbf{T}^{\top}(\mathbf{G}_{k,0}+\mathbf{R}_{k,0})$ plus $o(1)$, whose diagonal entries are bounded continuous functions of the node-level averages $\bar{A},\bar{B},\bar{C}$. The scaled trace therefore equals $n^{-1}\sum_{i=1}^{n}\phi_{k}(\alpha_{i0},X_{i})+o(1)$ for a bounded continuous $\phi_{k}$, converging to $\mathbb{E}[\phi_{k}(A,X)]$ by the LLN. The argument for $b_{k0}$ is identical, using the bounds on $\mathbf{W}_{k,0}$ from Lemma~\ref{prop:assu5_primitive}.

\smallskip

\noindent\textit{Part (c), positive definiteness of $\Omega_0$ and $\mathrm{I}_0$.} We prove $\mathrm{I}_{0}\succ 0$; the argument for $\Omega_{0}$ is analogous. Suppose $v^{\top}\mathrm{I}_{0}v=0$ for some nonzero $v\in\mathbb{R}^{K}$. By the Frisch--Waugh--Lovell residualization from the discussion following Assumption~\ref{assu:res_order},
\[
v^{\top}\mathrm{I}_{0}v=\lim_{n\to\infty}\frac{1}{N}\sum_{i<j}p_{ij,0}(1-p_{ij,0})(v^{\top}x_{ij}-w_{i,v}-w_{j,v})^{2}=0,
\]
where $w_{i,v}$ are the profiling weights for $v^{\top}\beta$. Since $p_{ij,0}(1-p_{ij,0})\geq c_{1}(1-c_{1})>0$ by Assumption~\ref{assu:ass_f}, this forces $v^{\top}x_{ij}=w_{i,v}+w_{j,v}$ in $L^{2}(\mu_{2})$, i.e., $v^{\top}x_{ij}$ is additively separable in the limit, contradicting Assumption~\ref{assu:mm_iden} since a non-degenerate covariate distribution guarantees variation in $\bar{S}_{n}(\beta_{0}+tv)$ orthogonal to additive node components. Hence $v=0$.
\end{proof}

\subsection{Convergence Bounds}

\begin{proof}[Proof of Lemma~\ref{matrix_inequality}]
This is equivalent to proving $\Vert\mathbf{B}^{\top}\mathbf{A}^{\top}\Vert_{\infty}\leq1-\frac{2(n-2)}{n-1}\delta^{2},$
which is a direct application of Lemma 2.1 of \citet{chatterjee2011random}.
\end{proof}

\subsection{Deviation Bounds}

Lemma~\ref{lem:dev_bound} provides
a bound on the deviation of the weighted sum of centered Bernoulli
random variables, $\sum_{j\neq i}\lambda_{ij}(y_{ij}-p_{ij,0})$, and
is used extensively in the proofs.

Let $\left\{ \l_{ij}\right\} _{i,j=1}^{n}$ denote a sequence of bounded
constants that satisfy $\max_{i,j}|\lambda_{ij}|<C_{1}$. In what follows, we apply the mean value theorem for vector-valued
functions in its integral form, as in \citet{chatterjee2011random}.
For example,
\begin{align*}
\mathbf{m}_{1}(\widehat{\bm{\alpha}},\beta)-\mathbf{m}_{1}(\bm{\alpha},\beta)
&=\left[\int_{0}^{1}\mathbf{J}_{11}(\bm{\alpha}+t(\widehat{\bm{\alpha}}-\bm{\alpha}),\beta)\,dt\right](\widehat{\bm{\alpha}}-\bm{\alpha})\\
&=:\mathbf{J}_{11}^{\circ}(\widehat{\bm{\alpha}},\bm{\alpha};\beta)(\widehat{\bm{\alpha}}-\bm{\alpha}).
\end{align*}
We write $\mathbf{J}_{11}^{\circ}(\widehat{\bm{\alpha}},\bm{\alpha};\beta)$
as $\mathbf{J}_{11}^{\circ}(\widehat{\bm{\alpha}},\bm{\alpha})$ whenever
there is no confusion, and other integral form Jacobian matrices are
defined similarly.
\begin{proof}[Proof of Lemma~\ref{lem:dev_bound}]
First, notice that $|\lambda_{ij}(y_{ij}-p_{ij,0})|<2C_{1}$ because
$y_{ij}-p_{ij,0}\in(-1,1)$; in addition, $y_{ij}$'s are independent
Bernoulli random variables with expectations $p_{ij,0}$. By Hoeffding's
inequality (see Theorem 2.8 of \citealp{boucheron2013concentration})
for the sum of bounded and independent random variables, we have 
\[
\Pr\left(\frac{1}{n-1}\left|\sum_{j\neq i}\lambda_{ij}(y_{ij}-p_{ij,0})\right|>t\right)\leq2\exp\left(-\frac{(n-1)t^{2}}{2C_{1}^{2}}\right).
\]
Letting $t=C_{1}\sqrt{6(n-1)^{-1}\log n}$, we obtain
\[
\Pr\left(\frac{1}{n-1}\left|\sum_{j\neq i}\lambda_{ij}(y_{ij}-p_{ij,0})\right|>C_{1}\sqrt{\frac{6\log n}{n-1}}\right)\leq2n^{-\frac{3(n-1)}{n-1}}=2n^{-3}.
\]
By Boole's inequality,
\[
\Pr\left(\max_{1\leq i\leq n}\frac{1}{n-1}\left|\sum_{j\neq i}\lambda_{ij}(y_{ij}-p_{ij,0})\right|>C_{1}\sqrt{\frac{6\log n}{n-1}}\right)\leq n\cdot2n^{-3}=2n^{-2}.\qedhere
\]
\end{proof}
Using Lemma~\ref{lem:dev_bound}, we can bound the estimation error
of $\widehat{\bm{\alpha}}(\beta_{0})-\bm{\alpha}_{0}$, which guarantees
that our moment estimator is consistent for $\bm{\alpha}_{0}$ when
$\beta_{0}$ is known. This result can be strengthened to prove the
second part of Theorem~\ref{thm:jmm_alpha_beta_consistency},
which we do in Appendix~\ref{sec:proofs_main}.
\begin{lem}
\label{dev_bound2} If Assumptions \ref{assu:model}, \ref{assu:compact_support_and_sampling}, and \ref{assu:ass_f}
hold, then we have
\begin{equation*}
\Pr\left\{ \left|\frac{1}{N}\sum_{i=1}^{n}\sum_{j>i}\lambda_{ij}(y_{ij}-p_{ij,0})\right|>C_{1}\sqrt{\frac{2\log N}{N}}\right\} \leq\frac{4}{n(n-1)}.
\end{equation*}
\end{lem}
\begin{proof}
Similar to the proof of Lemma~\ref{lem:dev_bound}, by Hoeffding's
inequality, we have 
\[
\Pr\left(\frac{1}{N}\left|\sum_{i=1}^{n}\sum_{j>i}\lambda_{ij}(y_{ij}-p_{ij,0})\right|>t\right)\leq2\exp\left(-\frac{Nt^{2}}{2C_{1}^{2}}\right).
\]
Letting $t=C_{1}\sqrt{\frac{2\log N}{N}}$, we obtain 
\[
\Pr\left(\frac{1}{N}\left|\sum_{i=1}^{n}\sum_{j>i}\lambda_{ij}(y_{ij}-p_{ij,0})\right|>C_{1}\sqrt{\frac{2\log N}{N}}\right)\leq2N^{-1}=\frac{4}{n(n-1)}.\qedhere
\]
\end{proof}
\begin{proof}[Proof of Lemma~\ref{yi_bound}]
\noindent\textbf{Part \ref{yi-bound-a}.} The rest of the proof is conditional on the following event, which happens
with probability at least $1-2n^{-2}$ by Lemma~\ref{lem:dev_bound}:
\[
\mathcal{E}_{n} \coloneqq \left\{ \max_{1\leq i\leq n}\frac{1}{n-1}\left|\sum_{j\neq i}(y_{ij}-p_{ij,0})\right|\leq\sqrt{\frac{6\log n}{n-1}}=O\left(\sqrt{\frac{\log n}{n}}\right)\right\} .
\]
Since $\mathbf{m}_{1}(\widehat{\bm\alpha}(\beta_{0}),\beta_{0})=0$ by definition, a first-order Taylor expansion of $\mathbf{m}_{1}$ around $\bm\alpha_{0}$ gives $\widehat{\bm\alpha}(\beta_{0}) - \bm\alpha_{0} = -\tilde{\mathbf{J}}^{-1}\mathbf{m}_{1,0}$, where $\tilde{\mathbf{J}}\coloneqq\mathbf{J}_{11}^{\circ}(\widehat{\bm\alpha}(\beta_{0}),\bm\alpha_{0})$ and its diagonal approximation (per Lemma~\ref{lem:appro_inverse}) is $\tilde{\mathbf{T}}\coloneqq\mathbf{T}^{\circ}(\widehat{\bm\alpha}(\beta_{0}),\bm\alpha_{0})$. By Lemma~\ref{lem:appro_inverse} and the triangle inequality,
\begin{align*}
\Vert\widehat{\bm{\alpha}}(\beta_{0})-\bm{\alpha}_{0}\Vert_{\infty}
&= \Vert\tilde{\mathbf{J}}^{-1}\mathbf{m}_{1,0}\Vert_{\infty}\\
&\leq \Vert\tilde{\mathbf{T}}\Vert_{\infty}\Vert\mathbf{m}_{1,0}\Vert_{\infty}
+\Vert\tilde{\mathbf{J}}^{-1}-\tilde{\mathbf{T}}\Vert_{\infty}\Vert\mathbf{m}_{1,0}\Vert_{\infty}.
\end{align*}
For the first part, $\tilde{\mathbf{T}}$
is diagonal with entries $O(n^{-1})$
uniformly; by Lemma~\ref{lem:dev_bound},
\[
\Vert\tilde{\mathbf{T}}\Vert_{\infty}\Vert\mathbf{m}_{1,0}\Vert_{\infty}
=O(n^{-1})\cdot\max_{1\leq i\leq n}\left|\textstyle\sum_{j\neq i}(y_{ij}-p_{ij,0})\right|
=O\!\left(\sqrt{\log n/n}\right).
\]
For the second part, by Lemma~\ref{lem:appro_inverse},
$\Vert\tilde{\mathbf{J}}^{-1}-\tilde{\mathbf{T}}\Vert_{\infty}
\leq n\Vert\tilde{\mathbf{J}}^{-1}-\tilde{\mathbf{T}}\Vert_{\max}=O(n^{-1})$,
so
\[
\Vert\tilde{\mathbf{J}}^{-1}-\tilde{\mathbf{T}}\Vert_{\infty}\Vert\mathbf{m}_{1,0}\Vert_{\infty}
=O\!\left(\sqrt{\log n/n}\right).
\]
 Combining these two results, we have 
$
\Vert\widehat{\bm{\alpha}}(\beta_{0})-\bm{\alpha}_{0}\Vert_{\infty}=O\left(\sqrt{\log{n}/n}\right).
$

  \noindent\textbf{Parts \ref{yi-bound-b} and \ref{yi-bound-c}.} By the
third-order Taylor expansion, which is also used in the proof of
Lemma 6 of \citet{graham2017econometric}, we have
\begin{align}
&\mathbf{m}_{1}(\widehat{\bm{\alpha}}(\beta_{0}),\beta_{0})-\mathbf{m}_{1,0}\nonumber\\
=\ &  \mathbf{J}_{11,0}[\widehat{\bm{\alpha}}(\beta_{0})-\bm{\alpha}_{0}]
  +\frac{1}{2}\left[\sum_{k=1}^{n}(\hat{\alpha}_{k}(\beta_{0})-\alpha_{k0})\frac{\partial\mathbf{J}_{11}(\bm{\alpha}_{0},\beta_{0})}{\partial\alpha_{k}}\right][\widehat{\bm{\alpha}}(\beta_{0})-\bm{\alpha}_{0}]
\nonumber \\
 & +\frac{1}{6}\left[\sum_{k=1}^{n}\sum_{l=1}^{n}(\hat{\alpha}_{k}(\beta_{0})-\alpha_{k0})(\hat{\alpha}_{l}(\beta_{0})-\alpha_{l0})\frac{\partial^{2}\mathbf{J}_{11}(\bar{\bm{\alpha}}^{kl},\beta_{0})}{\partial\alpha_{k}\partial\alpha_{l}}\right][\widehat{\bm{\alpha}}(\beta_{0})-\bm{\alpha}_{0}],
\label{eq:alpha_ta}
\end{align}
where for each $(k,l)$, $\bar{\bm{\alpha}}^{kl}=\bm{\alpha}_{0}+s_{kl}(\widehat{\bm{\alpha}}(\beta_{0})-\bm{\alpha}_{0})$
for some $s_{kl}\in(0,1)$. 

Recall the definition of $\bm{\eta}_0$ and define the $n\times1$ vector
\begin{equation}
\bm{\rho}_{0}\coloneqq\left[\sum_{k=1}^{n}\sum_{l=1}^{n}(\hat{\alpha}_{k}(\beta_{0})-\alpha_{k0})(\hat{\alpha}_{l}(\beta_{0})-\alpha_{l0})\frac{\partial^{2}\mathbf{J}_{11}(\bar{\bm{\alpha}}^{kl},\beta_{0})}{\partial\alpha_{k}\partial\alpha_{l}}\right][\widehat{\bm{\alpha}}(\beta_{0})-\bm{\alpha}_{0}].\label{eq:rho0_def}
\end{equation}
Let $\delta_i\coloneqq\hat{\alpha}_i(\beta_0)-\alpha_{i0}$, by part \ref{yi-bound-a}, we have $\sup_{i}|\delta_i|=O_p(\sqrt{\log n/n})$.
Because only the $k$th row and the $k$th
column of $\mathbf{J}_{11}(\bm{\alpha},\beta)$ contain functions
of $\alpha_{k}$, by a direct calculation we write the entries of
$\Lambda_{k} \coloneqq \frac{\partial\mathbf{J}_{11}(\bm{\alpha}_{0},\beta_{0})}{\partial\alpha_{k}}$
as 
\[
(\Lambda_{k})_{kl}=-p_{kl,0}^{(1,1,0)}\ (l\neq k),\quad (\Lambda_{k})_{lk}=-p_{kl,0}^{(2,0,0)}\ (l\neq k),\quad (\Lambda_{k})_{kk}=-\textstyle\sum_{p\neq k}p_{kp,0}^{(2,0,0)},
\]
with all other entries equal to zero.
Hence, let $\Lambda=\sum_{k=1}^{n}(\hat{\alpha}_{k}(\beta_{0})-\alpha_{k0})\frac{\partial\mathbf{J}_{11}(\bm{\alpha}_{0},\beta_{0})}{\partial\alpha_{k}}$
whose entries are 
\[
\begin{aligned}\Lambda_{ij} & =-(\hat{\alpha}_{i}(\beta_{0})-\alpha_{i0})p_{ij,0}^{(1,1,0)}-(\hat{\alpha}_{j}(\beta_{0})-\alpha_{j0})p_{ji,0}^{(2,0,0)},\ i\neq j,\\
\Lambda_{ii} & =-(\hat{\alpha}_{i}(\beta_{0})-\alpha_{i0})\sum_{j\neq i}p_{ij,0}^{(2,0,0)}-\sum_{k\neq i}(\hat{\alpha}_{k}(\beta_{0})-\alpha_{k0})p_{ik,0}^{(1,1,0)}.
\end{aligned}
\]
Then, the $i$th element of $\bm\eta_{0}$ can be calculated as
\[
\begin{aligned}\eta_{i,0}= & \ \Lambda_{ii}\cdot(\hat{\alpha}_{i}(\beta_{0})-\alpha_{i0})+\sum_{j\neq i}\Lambda_{ij}\cdot(\hat{\alpha}_{j}(\beta_{0})-\alpha_{j0})\\
= & \ -\sum_{j\neq i}p_{ij,0}^{(2,0,0)}(\hat{\alpha}_{i}(\beta_{0})-\alpha_{i0})^{2}\\
 & \ \ -2\sum_{j\neq i}p_{ij,0}^{(1,1,0)}(\hat{\alpha}_{i}(\beta_{0})-\alpha_{i0})(\hat{\alpha}_{j}(\beta_{0})-\alpha_{j0})\\
 & \ \ -\sum_{j\neq i}p_{ji,0}^{(2,0,0)}(\hat{\alpha}_{j}(\beta_{0})-\alpha_{j0})^{2}.
\end{aligned}
\]
It is clear that $\|\bm\eta_{0}\|_\infty=O_p(\log n/n)$ by $\|\hat{\bm\alpha}(\beta_0)-\bm\alpha_0\|_\infty=O_p(\sqrt{\log n/n})$.
Recall the symmetric matrix $\mathbf{M}_{i,n}$ defined in (\ref{eq:M_in}).
Then, we have $\eta_{i,0}=\bm\delta^\top\mathbf{M}_{i,n,0}\bm\delta.$
Recall from Appendix~\ref{sec:appendix_defs} that $\varphi_{k,0}^\top = e_k^\top\mathbf{J}_{21,0}\mathbf{J}_{11,0}^{-1}$ and $\mathbf{G}_{k,0}=\sum_{i=1}^n \varphi_{k,i,0}\mathbf{M}_{i,n,0}$.
Hence, $e_k^\top\mathbf{J}_{21,0}\mathbf{J}_{11,0}^{-1}\bm\eta_{0}
=\varphi_{k,0}^\top\bm\eta_{0}
=\sum_{i=1}^n \varphi_{k,i,0}\,\eta_{i,0}
=\bm\delta^\top\mathbf{G}_{k,0}\bm\delta,$ where $e_{k}\in\mathbb{R}^{K}$ is the $k$th canonical basis vector.
Substituting $\bm\delta=-\mathbf{J}_{11,0}^{-1}\mathbf{m}_{1,0}+\mathbf{r}_n$ with $\|\mathbf{r}_n\|_\infty=O_p(\log n/n)$, we have
\begin{align*}
\frac{1}{\sqrt{N}}e_k^\top\mathbf{J}_{21,0}\mathbf{J}_{11,0}^{-1}\bm\eta_{0}
&=\frac{1}{\sqrt{N}}\mathbf{m}_{1,0}^{\top}(\mathbf{J}_{11,0}^{-1})^\top\mathbf{G}_{k,0}\mathbf{J}_{11,0}^{-1}\mathbf{m}_{1,0} \\
&\quad -\frac{2}{\sqrt{N}}\mathbf{m}_{1,0}^{\top}(\mathbf{J}_{11,0}^{-1})^\top\mathbf{G}_{k,0}\mathbf{r}_n
+\frac{1}{\sqrt{N}}\mathbf{r}_n^\top\mathbf{G}_{k,0}\mathbf{r}_n.
\end{align*}
By the definition of $\mathbf{G}_{k}$ in Appendix~\ref{sec:appendix_defs} and bounded second-order derivatives, $\mathbf{G}_{k,0}$ has
uniformly bounded row sums, so the last two terms are $o_p(1)$. Therefore,
\begin{align*}
\frac{1}{2\sqrt{N}}e_k^\top\mathbf{J}_{21,0}\mathbf{J}_{11,0}^{-1}\bm\eta_{0}
&=\frac{1}{2\sqrt{N}}\mathbf{m}_{1,0}^{\top}(\mathbf{J}_{11,0}^{-1})^\top\mathbf{G}_{k,0}\mathbf{J}_{11,0}^{-1}\mathbf{m}_{1,0}+o_p(1)\\
&\stackrel{p}{\to} \lim_{n\to\infty}\frac{1}{2\sqrt{N}}\mathrm{Tr}[\mathbf{J}_{11,0}^{-1}\mathbf{V}_{11,0}(\mathbf{J}_{11,0}^{-1})^\top\mathbf{G}_{k,0}]=B_{k0}^{(1)}.
\end{align*}

Next, uniformly for $i\in\mathcal{I}_n$, we have:
\begin{align*}
|\rho_{i,0}|
&\leq C\sum_{j\neq i}\bigl(|\delta_i|^3+|\delta_i|^2|\delta_j|+|\delta_i||\delta_j|^2+|\delta_j|^3\bigr)\\
&\leq 8C(n-1)(\sup_i|\delta_i|)^{3}
=O\!\left((\log n)^{3/2}/\sqrt{n}\right),
\end{align*}
by bounded third-order derivatives.
Hence,
\[
\Vert\mathbf{J}_{11,0}^{-1}\bm{\rho}_{0}\Vert_{\infty}
\leq\Vert\mathbf{J}_{11,0}^{-1}\Vert_{\infty}\Vert\bm{\rho}_{0}\Vert_{\infty}
=O\left(\frac{1}{n}\right)\cdot O_p\left(\frac{(\log n)^{3/2}}{\sqrt{n}}\right)
=O_p\left(\frac{(\log n)^{3/2}}{n^{3/2}}\right).
\]
Finally, by (\ref{eq:alpha_ta}), we have
\[
\left\Vert \widehat{\bm{\alpha}}(\beta_{0})-\bm{\alpha}_{0}+\mathbf{J}_{11,0}^{-1}\mathbf{m}_{1,0}+\frac{1}{2}\mathbf{J}_{11,0}^{-1}\bm{\eta}_{0}\right\Vert _{\infty}
=\left\Vert \frac{1}{6}\mathbf{J}_{11,0}^{-1}\bm{\rho}_{0}\right\Vert_{\infty}
=O_p\left(\frac{(\log n)^{3/2}}{n^{3/2}}\right).\qedhere
\]
\end{proof}

\subsection{Uniform Convergence of the Concentrated Moment Equation}

Recall the concentrated moment equation and its population counterpart:
\[
S_{n}(\beta) \coloneqq N^{-1}m_{2}(\widehat{\bm{\alpha}}(\beta),\beta)\text{ and }\bar{S}_{n}(\beta) \coloneqq N^{-1}\mathbb{E}[m_{2}(\bm{\alpha}(\beta),\beta)|\mathbf{x},\bm{\alpha}_{0}],
\]
where $\bm{\alpha}(\beta)$ is the unique solution to
$\mathbb{E}[\mathbf{m}_{1}(\bm{\alpha},\beta)|\mathbf{x},\bm{\alpha}_{0}]=0.$

\begin{proof}[Proof of Lemma~\ref{lem:bound_sn_s}]
By the mean-value theorem, we have
$
\widehat{\bm{\alpha}}(\beta)-\bm{\alpha}(\beta)=-[\mathbf{J}_{11}^{\circ}(\widehat{\bm{\alpha}}(\beta),\bm{\alpha}(\beta))]^{-1}\mathbf{m}_{1,0}.
$
Decomposing $S_{n}(\beta)-\bar{S}_{n}(\beta)$ via the diagonal approximation
$\mathbf{T}^{\circ}=[\mathrm{diag}(\mathbf{J}_{11}^{\circ})]^{-1}$
from Lemma~\ref{lem:appro_inverse} gives
\begin{align*}
S_{n}(\beta)-\bar{S}_{n}(\beta)
=\  &\underbrace{N^{-1}\textstyle\sum_{i}\sum_{j>i}(y_{ij}-p_{ij,0})x_{ij}}_{R_{1}}-\underbrace{N^{-1}\mathbf{J}_{21}^{\circ}\mathbf{T}^{\circ}\mathbf{m}_{1,0}}_{R_{2}}\\
 &+\underbrace{N^{-1}\mathbf{J}_{21}^{\circ}[\mathbf{T}^{\circ}-(\mathbf{J}_{11}^{\circ})^{-1}]\mathbf{m}_{1,0}}_{R_{3}}.
\end{align*}
By Lemmas~\ref{dev_bound2} and~\ref{lem:dev_bound},
$\lVert R_{1}\rVert_{\infty}=O_{p}(\sqrt{(\log N)/N})$.
Since $\mathbf{T}^{\circ}$ is diagonal with entries $O(n^{-1})$ and
$\lVert\mathbf{J}_{21}^{\circ}\rVert_{\max}=O(n)$,
$\lVert R_{2}\rVert_{\infty}=O_{p}(\sqrt{(\log n)/n})$.
For $R_{3}$, using $\lVert\mathbf{T}^{\circ}-(\mathbf{J}_{11}^{\circ})^{-1}\rVert_{\max}=O(n^{-2})$,
$\lVert R_{3}\rVert_{\infty}=O_{p}(\sqrt{(\log n)/n})$.
All bounds hold uniformly in $\beta\in\mathbb{B}$.
\end{proof}

\subsection{One-Step Score Bounds}

Let $s_{n}(\bm{\alpha},\beta)=s_{2}(\bm{\alpha},\beta)-\sum_{i=1}^{n}s_{1i}(\bm{\alpha},\beta)w_{i}(\bm{\alpha},\beta)$,
where $w_{i}(\bm{\alpha},\beta)$ is the $i$th column of $\mathbf{I}_{12}(\bm{\alpha},\beta)^{\top}\mathbf{I}_{11}(\bm{\alpha},\beta)^{-1}$.
Taking derivatives, we have
\begin{align*}
\nabla_{\bm{\alpha}^{\top}}s_{n}&=\mathbf{H}_{12}^{\top}-\mathbf{I}_{12}^{\top}\mathbf{I}_{11}^{-1}\mathbf{H}_{11}-\textstyle\sum_{i=1}^{n}s_{1i}\,\partial_{\bm{\alpha}^{\top}} w_{i}\text{ and }\\
\nabla_{\beta^{\top}}s_{n}&=\mathrm{H}_{22}-\mathbf{I}_{12}^{\top}\mathbf{I}_{11}^{-1}\mathbf{H}_{12}-\textstyle\sum_{i=1}^{n}s_{1i}\,\partial_{\beta^{\top}} w_{i},
\end{align*}
where all matrices are evaluated at $(\bm{\alpha},\beta)$.

\begin{proof}[Proof of Lemma~\ref{lem:one_step_bound}]
Each row of $\mathbf{H}_{12,0}+\mathbf{I}_{12,0}$ is a weighted sum of
$(y_{ij}-p_{ij,0})$ terms, so by Lemma~\ref{yi_bound} and the continuous mapping theorem (CMT),
\[
N^{-1/2}\left\Vert \mathbf{H}_{12}(\widehat{\bm{\alpha}}(\beta_{0}),\beta_{0})
+\mathbf{I}_{12}(\widehat{\bm{\alpha}}(\beta_{0}),\beta_{0})\right\Vert _{\max}=o_{p}(1).
\]
Similarly,
\[
N^{-1/2}\left\Vert \mathbf{I}_{12}(\widehat{\bm{\alpha}}(\beta_{0}),\beta_{0})^{\top}
\mathbf{I}_{11}(\widehat{\bm{\alpha}}(\beta_{0}),\beta_{0})^{-1}
\left[\mathbf{H}_{11}(\widehat{\bm{\alpha}}(\beta_{0}),\beta_{0})
+\mathbf{I}_{11}(\widehat{\bm{\alpha}}(\beta_{0}),\beta_{0})\right]\right\Vert_{\max}=o_{p}(1).
\]
Combining these two bounds,
\begin{equation}
N^{-1/2}\lVert\mathbf{H}_{12}(\widehat{\bm{\alpha}}(\beta_{0}),\beta_{0})^{\top}
-\mathbf{I}_{12}^{\top}\mathbf{I}_{11}^{-1}\mathbf{H}_{11}\rVert_{\max}=o_{p}(1),
\label{eq:bound_h_i}
\end{equation}
evaluated at $(\widehat{\bm{\alpha}}(\beta_{0}),\beta_{0})$.
From the exact mean-value representation in the proof of Lemma~\ref{yi_bound},
$
\widehat{\bm{\alpha}}(\beta_{0})-\bm{\alpha}_{0}
=-\left[\mathbf{J}_{11}^{\circ}(\widehat{\bm{\alpha}}(\beta_{0}),\bm{\alpha}_{0})\right]^{-1}\mathbf{m}_{1,0}.
$
Lemma~\ref{lem:appro_inverse} implies that the diagonal entries of
$\left[\mathbf{J}_{11}^{\circ}(\widehat{\bm{\alpha}}(\beta_{0}),\bm{\alpha}_{0})\right]^{-1}$
are $O_p(n^{-1})$ and the off-diagonal entries are $O(n^{-2})$. Hence,
for any deterministic unit vector $\mathbf{c}$, $\mathbf{c}^{\top}(\widehat{\bm{\alpha}}(\beta_{0})-\bm{\alpha}_{0})$
is a weighted sum of conditionally independent centered dyad shocks
with total squared weight $O(n^{-1})$, so $\sqrt{n}\mathbf{c}^{\top}(\widehat{\bm{\alpha}}(\beta_{0})-\bm{\alpha}_{0})=O_{p}(1)$.
Combining this with (\ref{eq:bound_h_i}) yields
\begin{equation}
\frac{1}{\sqrt{N}}[\mathbf{H}_{12}^{\top}-\mathbf{I}_{12}^{\top}\mathbf{I}_{11}^{-1}\mathbf{H}_{11}](\widehat{\bm{\alpha}}(\beta_{0})-\bm{\alpha}_{0})=o_{p}(1).\label{eq:h_i_alpha}
\end{equation}
Next, similarly to the process of finding the bias term in the proof
of Theorem~\ref{thm:jmm_alpha_beta_consistency}, we have
\begin{align*}
 & \ -\frac{1}{\sqrt{N}}\sum_{i=1}^{n}s_{1i}(\widehat{\bm{\alpha}}(\beta_{0}),\beta_{0})\frac{\partial w_{ki}(\widehat{\bm{\alpha}}(\beta_{0}),\beta_{0})}{\partial\bm{\alpha}^{\top}}(\widehat{\bm{\alpha}}(\beta_{0})-\bm{\alpha}_{0})\\
= & \ \frac{1}{\sqrt{N}}\mathbf{s}_{1,0}^{\top}\mathbf{W}_{k,0}\mathbf{J}_{11,0}^{-1}\mathbf{m}_{1,0}+o_{p}(1),
\end{align*}
where $[\mathbf{W}_{k}(\bm{\alpha},\beta)]_{ij}=\frac{\partial w_{ki}(\bm{\alpha},\beta)}{\partial\alpha_{j}}$
and the equality holds by $\lVert\widehat{\bm{\alpha}}(\beta_{0})-\bm{\alpha}_{0}\Vert_{\infty}=o_{p}(1)$
and the CMT.

The asymptotic bias $b_{0} \coloneqq (b_{10},\dots,b_{K0})^{\top}$ for the
one-step estimator is defined as
\begin{equation}
b_{k0}=\lim_{n\to\infty}\frac{1}{\sqrt{N}}\mathrm{Tr}[\mathbf{J}_{11,0}^{-1}\text{Cov}(\mathbf{m}_{1,0},\mathbf{s}_{1,0})\mathbf{W}_{k,0}],\quad k=1,\dots,K,
\end{equation}
where the entries of the $n\times n$ covariance matrix $\mathrm{Cov}(\mathbf{m}_{1,0},\mathbf{s}_{1,0})$
are
\begin{align}
[\mathrm{Cov(\mathbf{m}_{1,0},\mathbf{s}_{1,0})}]_{ij}= & \frac{p_{ji,0}^{(1,0,0)}\mathrm{Var}(y_{ij})}{p_{ji,0}(1-p_{ij,0})}=p_{ji,0}^{(1,0,0)},\quad1\leq i\neq j\leq n,\nonumber\\{}
[\mathrm{Cov(\mathbf{m}_{1,0},\mathbf{s}_{1,0})}]_{ii}= & \sum_{k\neq i}\frac{p_{ik,0}^{(1,0,0)}\mathrm{Var}(y_{ik})}{p_{ik,0}(1-p_{ik,0})}=\sum_{k\neq i}p_{ik,0}^{(1,0,0)},\quad1\leq i\leq n.\label{eq:cov_ms}
\end{align}
 
 Next, we have
\begin{align}
 & \ \frac{1}{\sqrt{N}}\mathbf{s}_{1,0}^{\top}\mathbf{W}_{k,0}\mathbf{J}_{11,0}^{-1}\mathbf{m}_{1,0}=\frac{1}{\sqrt{N}}\mathrm{Tr}(\mathbf{J}_{11,0}^{-1}\mathbf{m}_{1,0}\mathbf{s}_{1,0}^{\top}\mathbf{W}_{k,0})\nonumber \\
= & \ \frac{1}{\sqrt{N}}\mathrm{Tr}[\mathbf{J}_{11,0}^{-1}\mathrm{Cov}(\mathbf{m}_{1,0},\mathbf{s}_{1,0})\mathbf{W}_{k,0}] \nonumber\\
& \ + \frac{1}{\sqrt{N}}\left\{\mathrm{Tr}(\mathbf{J}_{11,0}^{-1}\mathbf{m}_{1,0}\mathbf{s}_{1,0}^{\top}\mathbf{W}_{k,0})-\mathrm{Tr}[\mathbf{J}_{11,0}^{-1}\mathrm{Cov}(\mathbf{m}_{1,0},\mathbf{s}_{1,0})\mathbf{W}_{k,0}]\right\} \nonumber\\
= & \ R_{1}+R_{2}.\label{eq:bias_b}
\end{align}
Notice that $R_{1}\to b_{k0}$ by definition. By the law of large
numbers for U-statistics, $R_{2}\stackrel{p}{\to}0$ under the bounds
in Lemma~\ref{prop:assu5_primitive}. By (\ref{eq:h_i_alpha}) and (\ref{eq:bias_b}),
we have
$N^{-1/2}\nabla_{\bm{\alpha}^{\top}}s_{n}(\widehat{\bm{\alpha}}(\beta_{0}),\beta_{0})(\widehat{\bm{\alpha}}(\beta_{0})-\bm{\alpha}_{0})\stackrel{p}{\to}b_{0}$,
which proves Lemma~\ref{lem:one_step_bound}(a).

For the second result of Lemma~\ref{lem:one_step_bound}, the law of large numbers gives
$N^{-1}[\mathrm{H}_{22}-\mathbf{I}_{12}^{\top}\mathbf{I}_{11}^{-1}\mathbf{H}_{12}]\stackrel{p}{\to}-\mathrm{I}_{0}$
at $(\widehat{\bm{\alpha}}(\bar{\beta}),\bar{\beta})$.
Since $\lVert\partial_{\beta}w_{i}\rVert_{\infty}=O(1)$ by Lemma~\ref{prop:assu5_primitive}
and $N^{-1}\sum_{i}|s_{1i}|=o_{p}(1)$ by the law of large numbers,
\begin{equation}
N^{-1}\nabla_{\beta^{\top}}s_{n}(\widehat{\bm{\alpha}}(\bar{\beta}),\bar{\beta})+\mathrm{I}_{0}=o_{p}(1).\label{eq:information_equality}
\end{equation}
By (\ref{eq:bound_h_i}) and an argument analogous to the proof of the first result,
$N^{-1}\nabla_{\bm{\alpha}^{\top}}s_{n}\cdot\partial\widehat{\bm{\alpha}}/\partial\beta^{\top}=o_{p}(1)$,
which combined with (\ref{eq:information_equality}) proves Lemma~\ref{lem:one_step_bound}(b).
\end{proof}

\subsection{Extension to Multiple Networks}\label{subsec:sm_multi_networks}

We formalize the multi-network extension of Remark~\ref{rem:multi-networks}. Consider $V$ independent networks ($V$ fixed, $\min_v n_v\to\infty$) sharing a common slope $\b_0$ but with network-specific node fixed effects $\bm\alpha_{0,v}$. The link function $p(\cdot)$ may follow either the TU or NTU specification.

\paragraph{Modified assumptions.} Assumptions~\ref{assu:model}--\ref{assu:ass_f} and~\ref{assu:res_order} hold within each network with constants uniform in $v$. Assumption~\ref{assu:mm_iden} is imposed on the pooled concentrated moment
\[
\bar S_n(\beta)\coloneqq N^{-1}\sum_{v=1}^V\sum_{i<j}\mathbb E\!\left[(y_{ij,v}-p_{ij,v}(\bm\alpha_v(\beta),\beta))\,x_{ij,v}\,\big|\,\mathbf x,\bm\alpha_{0,v}\right],
\]
where $\bm\alpha_v(\beta)$ profiles the network-$v$ degree equations and $N\coloneqq\sum_v\binom{n_v}{2}$.

\begin{prop}[Multi-network extension]\label{prop:multinetwork}
Under the modified assumptions above, the conclusions of Theorems~\ref{thm:jmm_alpha_beta_consistency}--\ref{thm:bagging_normality} extend with $N=\sum_v\binom{n_v}{2}$, pooled bias $B_0=\lim\sum_v\sqrt{N_v/N}\,B_{0,v}$, and pooled information $\mathrm I_0=\lim\sum_v (N_v/N)\,\mathrm I_{0,v}$. In particular, $\sqrt N(\widehat\b_{\mathrm{BG}}-\b_0)\stackrel{d}{\to}\mathcal N(0,\mathrm I_0^{-1})$.
\end{prop}

\begin{proof}[Proof of Proposition~\ref{prop:multinetwork}]
Independence across networks makes $\mathbf{m}_1$, $\mathbf{s}_1$, $\mathbf{J}_{11}$, $\mathbf{I}_{11}$, $\mathbf{H}_{11}$ block-diagonal in $v$, so the per-network bounds in Lemmas~\ref{lem:alpha_conv_rate}--\ref{lem:one_step_bound} hold uniformly in $v$, and the proofs of Theorems~\ref{thm:jmm_alpha_beta_consistency}--\ref{thm:bagging_normality} carry over with the following adjustments.

\smallskip\noindent
\emph{JMM and OS.} Block-diagonality gives
\[
N^{-1}\mathrm{I}_n^{\mathrm{pool}}=\sum_v(N_v/N)\,N_v^{-1}\mathrm{I}_{n_v,v}\xrightarrow{p}\sum_v(\lim N_v/N)\,\mathrm{I}_{0,v}=\mathrm I_0,
\]
and analogously $\mathrm J_0$. The leading sums in the pooled JMM and OS expansions are sums of $V$ mutually independent within-network sums of bounded centered dyad-summands, so a Lyapunov CLT yields the stated Gaussian limits. The pooled biases $B_0=\lim\sum_v\sqrt{N_v/N}\,B_{0,v}$ and $b_0=\lim\sum_v\sqrt{N_v/N}\,b_{0,v}$ inherit the per-network $1/\sqrt{N_v}$ normalization through the rescaling $\sqrt{N_v}/\sqrt N$.

\smallskip\noindent
\emph{Bagging.} Halving each $\mathcal{I}_{n_v}$ reduces the pooled dyad count to $N/4$, so each half-pool $\widehat\beta_{\mathrm{OS},h}^{(t)}$ carries twice the leading bias of $\widehat\beta_{\mathrm{OS}}$; the split-network jackknife $\widehat\beta_{\mathrm{OS-SJ}}^{(t)}=2\widehat\beta_{\mathrm{OS}}-\tfrac12(\widehat\beta_{\mathrm{OS},1}^{(t)}+\widehat\beta_{\mathrm{OS},2}^{(t)})$ therefore self-cancels the leading bias for every $t$. The variance argument in the proof of Theorem~\ref{thm:bagging_normality} carries over because the per-network independent splits assign each within-network dyad to be cross-half (and thus dropped from both half-pools) with probability $\tfrac12$ independently across $t$, identical to the single-network setting; averaging over $\widetilde T_n\to\infty$ deflates the variance back to $\mathrm I_0^{-1}/N$. Hence $\sqrt{N}(\widehat\beta_{\mathrm{BG}}-\beta_0)\xrightarrow{d}\mathcal{N}(0,\mathrm{I}_0^{-1})$.\qedhere
\end{proof}

\setcounter{equation}{0}
\setcounter{thm}{0}
\setcounter{lem}{0}
\setcounter{prop}{0}
\setcounter{rem}{0}
\renewcommand{\theequation}{B.\arabic{equation}}
\renewcommand{\thethm}{B.\arabic{thm}}
\renewcommand{\thelem}{B.\arabic{lem}}
\renewcommand{\theprop}{B.\arabic{prop}}
\renewcommand{\therem}{B.\arabic{rem}}

\section{Link Function Misspecification}\label{sec:sm_misspecification}

We study dyadic network formation under possible link function misspecification in this section. The analysis applies to TU and NTU models. Suppose researchers misspecify the link function to be $q(\cdot)$ which differs from
$p\left(\cd\right)$ at points with strictly positive probability
measure. For a fixed $n$, we impose the following identification
assumption to facilitate the analysis. Let $q_{ij}(\bm{\alpha},\beta) \coloneqq q(\alpha_i,\alpha_j,x_{ij}^{\top}\beta)$
be the misspecified probability of linking between $i$ and $j$.
\begin{assumption}[Identification under Link Function Misspecification]
\label{assu:identification_model_misspecification}For sufficiently large $n$,
the normalized nonlinear function
\[
\tilde{S}_{n}(\beta) \coloneqq N^{-1}\sum_{i=1}^{n}\sum_{j>i}\left[p_{ij,0}-q_{ij}(\bm{\alpha}(\beta),\beta)\right]x_{ij}
\]
has a unique root $\beta_{n*}$, and satisfies
\[
\inf_{\beta\in\mathbb{B}:\left\Vert \beta-\beta_{n*}\right\Vert _{2}\geq\delta}\left\Vert \tilde{S}_{n}(\beta)\right\Vert _{2}>0
\]
 for all $\delta>0$, where $\bm{\alpha}(\beta)$
is the unique solution to the following system of equations
\[
\begin{pmatrix}\sum_{j\neq1}p_{1j,0}-\sum_{j\neq1}q_{1j}(\bm{\alpha},\beta),\cdots,\sum_{j\neq n}p_{nj,0}-\sum_{j\neq n}q_{nj}(\bm{\alpha},\beta)\end{pmatrix}^{\top}=0.
\]
\end{assumption}
Assumption~\ref{assu:identification_model_misspecification} is the
counterpart of Assumption~\ref{assu:mm_iden} under a misspecified
link function. Similarly to Lemma~\ref{lem:alpha_conv_rate},
the degree-matching system in Assumption~\ref{assu:identification_model_misspecification} has a unique solution with high
probability under mild conditions on $(\bm{\alpha}_{0},\beta_{0})$
and $\beta$. Thus, Assumption~\ref{assu:identification_model_misspecification}
identifies the homophily parameter under link function misspecification.
Notice that $\beta_{n*}$ depends on the true link function $p\left(\cd\right)$,
misspecified link function $q(\cdot)$, and the true parameter values.
As a result, $\beta_{n*}$ may vary with $n$. The following theorem
shows that the JMM estimator based on the misspecified link function
$q\left(\cd\right)$ is centered at $\beta_{n*}$ up to a bias, which
the split-network jackknife procedure removes asymptotically. Let
$\bm{\alpha}_{*} \coloneqq \bm{\alpha}(\beta_{n*})$ with $\bm{\alpha}\left(\cdot\right)$
satisfying Assumption~\ref{assu:identification_model_misspecification}.

In the next Theorem, $\mathrm{J}_{*}$, $B_{*}$, and $\Omega_{*}$
are defined analogously to $\mathrm{J}_{0}$, $B_{0}$, and $\Omega_{0}$
in Theorem~\ref{thm:jmm_alpha_beta_consistency},
with the pseudo values $(\bm{\alpha}_{*},\beta_{n*})$ and the
misspecified link function $q(\cdot)$. For the sandwich-form variance $\Omega_{*}$, the sandwich ``meat'' remains governed by the true DGP through
$\mathrm{Var}(y_{ij}\mid \mathbf{x},\bm{\alpha}_{0})=p_{ij,0}(1-p_{ij,0})$.  Furthermore,
 $\Omega_{*}$ can be consistently estimated by the plug-in sandwich formula.
\begin{thm}[JMM Estimation under Link Function Misspecification]
\label{thm:jmm_model_miss}If Assumptions \ref{assu:model}--\ref{assu:ass_f}
and \ref{assu:identification_model_misspecification} hold, then 
\[
\sqrt{N}(\hat{\beta}-\beta_{n*})-\mathrm{J}_{*}^{-1}B_{*}\stackrel{d}{\to}\mathcal{N}\left(0,\Omega_{*}\right).
\]
\end{thm}
Theorem~\ref{thm:jmm_model_miss} shows that the JMM estimator remains consistent for the pseudo-true value $\beta_{n*}$ (the unique solution to the pseudo-moment equations in Assumption~\ref{assu:identification_model_misspecification}) even when the link function is misspecified.

Under the link function misspecification, the one-step estimator becomes
\begin{equation}
\hat{\beta}_{\mathrm{OS}} \coloneqq \hat{\beta}-\mathrm{H}(\widehat{\bm{\alpha}},\hat{\beta})^{-1}s_{n}(\widehat{\bm{\alpha}},\hat{\beta}),\label{eq:os_miss}
\end{equation}
with the JMM estimator $(\widehat{\bm{\alpha}},\hat{\beta})$ substituted
in. Note that
\begin{equation}
\mathrm{H}(\bm{\alpha},\beta) \coloneqq \mathrm{H}_{22}(\bm{\alpha},\beta)-\mathbf{H}_{12}(\bm{\alpha},\beta)^{\top}\mathbf{H}_{11}(\bm{\alpha},\beta)^{-1}\mathbf{H}_{12}(\bm{\alpha},\beta),\label{eq:h_alpha_beta_def}
\end{equation}
is the concentrated Hessian matrix, and
\[
s_{n}(\bm{\alpha},\beta) \coloneqq s_{2}(\bm{\alpha},\beta)-\mathbf{H}_{12}(\bm{\alpha},\beta)^{\top}\mathbf{H}_{11}(\bm{\alpha},\beta)^{-1}\mathbf{s}_{1}(\bm{\alpha},\beta),
\]
is the concentrated score function. To define the pseudo-true target of
$\hat{\beta}_{\mathrm{OS}}$, write
$\bar{\mathbf{H}}_{ab}(\bm{\alpha},\beta)\coloneqq\mathbb{E}[\mathbf{H}_{ab}(\bm{\alpha},\beta)\mid\mathbf{x},\bm{\alpha}_{0}]$
for the expected Hessian blocks, and let
\[
\bar{\mathrm{H}}(\bm{\alpha},\beta)\coloneqq\bar{\mathrm{H}}_{22}-\bar{\mathbf{H}}_{12}^{\top}\bar{\mathbf{H}}_{11}^{-1}\bar{\mathbf{H}}_{12},\qquad
\bar{s}_{n}(\bm{\alpha},\beta)\coloneqq\mathbb{E}[s_{2}]-\bar{\mathbf{H}}_{12}^{\top}\bar{\mathbf{H}}_{11}^{-1}\mathbb{E}[\mathbf{s}_{1}],
\]
denote the corresponding population concentrated Hessian and score,
with arguments $(\bm{\alpha},\beta)$ suppressed for brevity.
Under the link function misspecification,
$\hat{\beta}_{\mathrm{OS}}$ in (\ref{eq:os_miss}) centers on the deterministic sequence
\begin{equation}
\beta_{n\sharp}\coloneqq\beta_{n*}-N^{-1}\mathrm{H}_{*}^{-1}\bar{s}_{n}(\bm{\alpha}_{*},\beta_{n*}),\label{eq:os_pseudo_value_miss}
\end{equation}
where $\mathrm{H}_{*}\coloneqq\lim_{n\to\infty}N^{-1}\bar{\mathrm{H}}(\bm{\alpha}_{*},\beta_{n*})$.
$\beta_{n\sharp}$ can be seen as a projection of $\beta_{n*}$ by concentrating
out the fixed effects. When the link function is correctly specified,
$(\bm{\alpha}_{*},\beta_{n*})\equiv(\bm{\alpha}_{0},\beta_{0})$,
thus $\beta_{n\sharp}\equiv\beta_{n*}\equiv\beta_{0}$ because
$\mathbb{E}[\mathbf{s}_{1,0}]=0$
and $\mathbb{E}[s_{2,0}]=0$
imply $\bar{s}_{n,0}\equiv0$.
Furthermore, our OS and BG estimators in the misspecified case share
similar asymptotic properties as their counterparts under correct specification, except that they now center on
the projected pseudo-true value $\beta_{n\sharp}$ instead of $\b_{0}$.

For the next theorem, define $b_{*}$ as $b_{0}$ but with $(\bm{\alpha}_{*},\beta_{n*})$
and the misspecified link function $q(\cdot)$. The asymptotic covariance
matrix $\Gamma_{*}$ is
\begin{equation}
\Gamma_{*} \coloneqq \lim_{n\to\infty} N^{-1}\mathrm{H}_{*}^{-1}\left[
\begin{aligned}
&\mathrm{I}_{22*}+\mathbf{H}_{12*}^{\top}\mathbf{H}_{11*}^{-1}\mathbf{I}_{11*}(\mathbf{H}_{11*}^{-1}\mathbf{H}_{12*})^{\top}\\
&-\mathbf{H}_{12*}^{\top}\mathbf{H}_{11*}^{-1}\mathbf{I}_{12*}-(\mathbf{H}_{12*}^{\top}\mathbf{H}_{11*}^{-1}\mathbf{I}_{12*})^{\top}
\end{aligned}\right](\mathrm{H}_{*}^{-1})^{\top}.\label{eq:gamma_star}
\end{equation}

\begin{thm}[OS and BG Estimation under Misspecified Link Function]
 \label{thm:os_asymptotic_results_miss}Suppose all the bounds in
Lemma~\ref{prop:assu5_primitive} hold for each element of $\mathbf{H}_{12}^{\top}\mathbf{H}_{11}^{-1}$.
If Assumptions \ref{assu:model}--\ref{assu:ass_f} and \ref{assu:identification_model_misspecification}
are satisfied, then
\begin{equation*}
\sqrt{N}(\hat{\beta}_{\mathrm{OS}}-\beta_{n\sharp})+\mathrm{H}_{*}^{-1}b_{*}  \stackrel{d}{\to} \mathcal{N}(0,\Gamma_{*})\quad\text{and}\quad
\sqrt{N}(\hat{\beta}_{\mathrm{BG}}-\beta_{n\sharp})  \stackrel{d}{\to} \mathcal{N}(0,\Gamma_{*}).
\end{equation*}
\end{thm}
Theorem~\ref{thm:os_asymptotic_results_miss} shows that $\hat{\beta}_{\mathrm{BG}}$ is robust: under correct specification it centers on $\beta_{0}$ and attains the CRLB; otherwise it centers on the projected pseudo-true value $\beta_{n\sharp}$ with no asymptotic bias. The covariances $\Omega_{*}$ and $\Gamma_{*}$ are consistently estimated by plugging $(\hat{\bm{\alpha}},\hat{\beta})$ into the formulas of Theorems~\ref{thm:jmm_model_miss} and~\ref{thm:os_asymptotic_results_miss}.

\subsection{Proofs of Theorems~\ref{thm:jmm_model_miss} and~\ref{thm:os_asymptotic_results_miss}}

\begin{proof}
Both proofs proceed in parallel to their correctly specified counterparts (Theorems~\ref{thm:jmm_alpha_beta_consistency}, \ref{thm:os_norm}, and~\ref{thm:bagging_normality}), with $(\bm{\alpha}_{0},\beta_{0})$ replaced by the pseudo-true parameters $(\bm{\alpha}_{*},\beta_{n*})$ and $(\bm{\alpha}_{*},\beta_{n\sharp})$, respectively. The existence and uniqueness of $\bm\alpha_*$ follows from Lemma~\ref{lem:alpha_conv_rate} applied to the misspecified link function $q(\cdot)$ under Assumption~\ref{assu:identification_model_misspecification}.

For Theorem~\ref{thm:jmm_model_miss}: the consistency argument is identical to that of Theorem~\ref{thm:jmm_alpha_beta_consistency}, with $\bar{S}_n(\beta)$ replaced by $\tilde{S}_n(\beta)$ from Assumption~\ref{assu:identification_model_misspecification}. The asymptotic normality follows the same Taylor expansion, except that the variance of $y_{ij}-q_{ij}(\bm\alpha_*,\beta_{n*})$ under the true DGP is $p_{ij,0}(1-p_{ij,0})$ (not $q_{ij}(1-q_{ij})$), yielding the sandwich-form variance $\Omega_*$ in place of $\Omega_0$.

For Theorem~\ref{thm:os_asymptotic_results_miss}: the one-step expansion parallels Theorem~\ref{thm:os_norm}, with the Hessian $\mathbf{H}$ replacing the information matrix $\mathbf{I}$ (since the information equality $\mathbb{E}[\mathbf{H}]=-\mathbf{I}$ fails under misspecification). The population concentrated score $\bar{s}_n(\bm\alpha_*,\beta_{n*})\neq 0$ shifts the pseudo-true target from $\beta_{n*}$ to $\beta_{n\sharp}$ as defined in (\ref{eq:os_pseudo_value_miss}). The bagging argument is unchanged.\qedhere
\end{proof}

\setcounter{equation}{0}
\setcounter{thm}{0}
\setcounter{lem}{0}
\setcounter{prop}{0}
\setcounter{rem}{0}
\setcounter{figure}{0}
\setcounter{table}{0}
\renewcommand{\theequation}{C.\arabic{equation}}
\renewcommand{\thethm}{C.\arabic{thm}}
\renewcommand{\thelem}{C.\arabic{lem}}
\renewcommand{\theprop}{C.\arabic{prop}}
\renewcommand{\therem}{C.\arabic{rem}}
\renewcommand{\thefigure}{C.\arabic{figure}}
\renewcommand{\thetable}{C.\arabic{table}}

\section{Monte Carlo Simulations}\label{sec:sm_simulations}

This section complements the baseline results in Section~\ref{sec:simulation} of the main text. Section~\ref{subsec:nonconcavity} illustrates the non-concavity challenge faced by direct MLE. Section~\ref{subsec:ntu_extended} presents extended NTU results: fixed-effect recovery, the APE estimator, link-function misspecification, and sparser networks. Section~\ref{subsec:tu_extended} briefly summarizes analogous TU results.

\subsection{Non-Concavity of MLE in the Fixed Effects}\label{subsec:nonconcavity}

The major challenge of directly invoking MLE is the non-concavity
of the log-likelihood function in the high-dimensional fixed effects $\bm{\alpha}$.
To illustrate this, Figure~\ref{fig:alpha_rmse} compares the distributions
of $\lVert\widehat{\bm{\alpha}}-\bm{\alpha}_{0}\rVert_{2}$ obtained by JMM and MLE, both initialized
at the same starting values.\footnote{MLE is implemented using the
L-BFGS-B algorithm in Python's SciPy package.} Both panels draw $\bm{\alpha}_{0}$ from
bounded supports but with structures that expose the non-concavity problem.
The left panel uses a symmetric bimodal design
$\alpha_{i0}\stackrel{i.i.d.}{\sim}0.5\,U(-3,-1)+0.5\,U(1,3)$, with a gap at zero;
the right panel considers a hub--periphery design
$\alpha_{i0}\stackrel{i.i.d.}{\sim}0.9\,U(-1.5,-0.5)+0.1\,U(3,4)$, in which a small
fraction of high-degree hubs coexists with a low-degree periphery.
In both settings the contrast is stark: JMM's RMSE distribution is
tightly concentrated near zero, whereas the bulk of MLE's distribution
shifts to substantially larger values, indicating that MLE often fails
to consistently recover $\bm{\alpha}_{0}$ from the same starting point.
The hub--periphery case is more extreme, with MLE errors reaching
nearly $10$ in some replications. These results highlight that, when
the fixed effects are high-dimensional, the non-concavity of the
log-likelihood poses a serious practical obstacle for MLE even under
bounded support, which our moment-based approach effectively circumvents.

\begin{figure}[!htbp]
\centering
\includegraphics[width=1\textwidth]{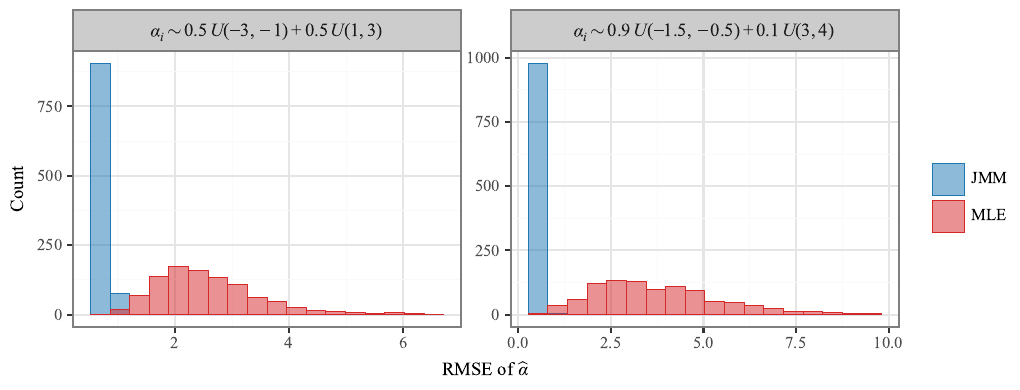}
\caption{Distributions of RMSE of $\widehat{\bm\alpha}$ under two NTU DGPs ($n=100$)}\label{fig:alpha_rmse}
\end{figure}

\subsection{Extended NTU Results}\label{subsec:ntu_extended}

This section reports further finite-sample results under the NTU specification of Section~\ref{sec:simulation}.

\paragraph{Fixed-effect recovery.} To examine estimation performance in the high-dimensional setting, where $\bm{\alpha}_{0}$ is an $n\times 1$ vector, we additionally consider $n=500$ and $1{,}000$. Figure~\ref{fig:alpha_fit} reports histograms of $\widehat{\alpha}_{i}-\alpha_{i0}$ for $i\in\mathcal{I}_{n}$. Across all sample sizes, the histograms are centered around zero. As $n$ increases, the estimation accuracy of $\widehat{\alpha}_{i}$ improves, and the dispersion of the estimation errors contracts toward zero, consistent with the theory.

\begin{figure}[!htbp]
\centering
\includegraphics[width=1\textwidth]{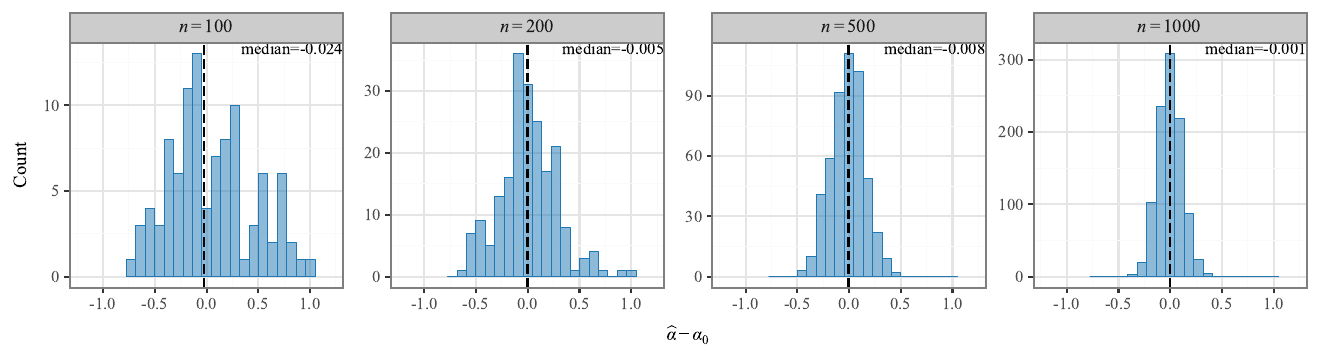}
\caption{Histograms of $\widehat{\bm{\alpha}}-\bm{\alpha}_{0}$ under NTU for different $n$}\label{fig:alpha_fit}
\end{figure}

\paragraph{APEs.} Table~\ref{tab:sim_ape} summarizes the APEs for each coordinate of $X_{ij}$ defined in (\ref{eq:ape}). Our plug-in estimator performs well with respect to RMSE and coverage probabilities. When applied to the estimation of APEs, the split-network jackknife bagging method does not yield meaningful improvement. As predicted by Theorem~\ref{thm:norm_ape}, the asymptotic bias in estimating APEs is $O(n^{-1})$, which is negligible relative to the APE estimator's $O(n^{-1/2})$ standard error.

\begin{table}[!htbp]
\centering
\begin{centering}
\caption{NTU estimation results for the APEs}\label{tab:sim_ape}
\par\end{centering}
\begin{threeparttable}

\begin{tabular}{lrrrrrrrr}
\toprule 
\multirow{3}{*}{} & \multicolumn{4}{c}{\emph{$n=100$}} & \multicolumn{4}{c}{\emph{$n=200$}}\tabularnewline
\cmidrule{2-9}
 & \multicolumn{2}{c}{Plug-in} & \multicolumn{2}{c}{Bagging} & \multicolumn{2}{c}{Plug-in} & \multicolumn{2}{c}{Bagging}\tabularnewline
\cmidrule{2-9}
 & $X_{ij,1}$ & $X_{ij,2}$ & $X_{ij,1}$ & $X_{ij,2}$ & $X_{ij,1}$ & $X_{ij,2}$ & $X_{ij,1}$ & $X_{ij,2}$\tabularnewline
\midrule 
Mean Bias & -0.14 & 0.39 & -0.16 & 0.14 & -0.00 & 0.19 & -0.00 & 0.09\tabularnewline
Median Bias & -0.04 & 0.46 & -0.09 & 0.21 & 0.04 & 0.15 & 0.06 & 0.08\tabularnewline
Standard Deviation & 2.75 & 2.79 & 2.89 & 2.85 & 1.32 & 1.42 & 1.37 & 1.45\tabularnewline
Mean Standard Error & 2.78 & 2.88 & 2.78 & 2.88 & 1.40 & 1.46 & 1.40 & 1.46\tabularnewline
Mean Absolute Bias & 2.18 & 2.26 & 2.28 & 2.29 & 1.04 & 1.15 & 1.08 & 1.16\tabularnewline
Median Absolute Bias & 1.78 & 1.87 & 1.91 & 1.88 & 0.84 & 0.97 & 0.89 & 0.99\tabularnewline
RMSE & 2.75 & 2.82 & 2.90 & 2.86 & 1.32 & 1.43 & 1.37 & 1.45\tabularnewline
90\% Coverage Rate & 90.8 & 90.2 & 89.2 & 89.2 & 91.8 & 90.6 & 90.8 & 90.8\tabularnewline
95\% Coverage Rate & 95.0 & 95.0 & 94.4 & 95.2 & 96.2 & 95.7 & 95.4 & 96.1\tabularnewline
\bottomrule
\end{tabular}

\begin{tablenotes}
\setlength{\itemindent}{-0.5cm}
\item \footnotesize \textit{Note:} All values have been multiplied by 100; true values of APEs are calibrated from a simulation with $n=10,000$ agents.
\end{tablenotes}
\end{threeparttable}

\end{table}

\paragraph{Link-function misspecification.} Table~\ref{tab:sim_beta_misspecify} presents the results for estimating the homophily coefficients under misspecification of the distribution of $\epsilon_{ij}$. We draw $\epsilon_{ij}$ from the standard normal distribution, but ``mistakenly'' specify the logistic link function in the estimation. Following Theorems~\ref{thm:jmm_model_miss} and~\ref{thm:os_asymptotic_results_miss}, we compare JMM to $\beta_{n*}$ and OS/BG to $\beta_{n\sharp}$ defined in (\ref{eq:os_pseudo_value_miss}), and find that the results are satisfactory. The performance of our BG estimator dominates other estimators in terms of bias, variance, and coverage probabilities, highlighting the efficacy and importance of proper bias-correction procedures.

\begin{table}[!htbp]
\centering
\begin{centering}
\caption{NTU estimation results for pseudo-true values under link function misspecification}\label{tab:sim_beta_misspecify}
\par\end{centering}
\begin{threeparttable}

\begin{tabular}{lrrrrrr}
\toprule 
\multirow{2}{*}{\emph{$n=100$}} & \multicolumn{2}{c}{JMM} & \multicolumn{2}{c}{OS} & \multicolumn{2}{c}{BG}\tabularnewline
\cmidrule{2-7}
 & $\b_{1}$ & $\b_{2}$ & $\b_{1}$ & $\b_{2}$ & $\b_{1}$ & $\b_{2}$\tabularnewline
\midrule 
Mean Bias & 5.92 & -5.34 & 5.36 & -5.14 & 0.15 & -0.02\tabularnewline
Median Bias & 5.99 & -4.83 & 5.37 & -4.91 & 0.24 & 0.11\tabularnewline
Standard Deviation & 6.16 & 15.36 & 6.16 & 15.34 & 5.98 & 14.93\tabularnewline
Mean Standard Error & 6.24 & 15.55 & 6.20 & 15.39 & 6.20 & 15.39\tabularnewline
Mean Absolute Bias & 7.03 & 12.81 & 6.66 & 12.76 & 4.78 & 11.89\tabularnewline
Median Absolute Bias & 6.29 & 10.18 & 5.91 & 10.44 & 4.02 & 9.84\tabularnewline
RMSE & 8.55 & 16.26 & 8.16 & 16.17 & 5.99 & 14.93\tabularnewline
90\% Coverage Rate & 75.8 & 88.6 & 78.0 & 88.5 & 91.7 & 90.8\tabularnewline
95\% Coverage Rate & 84.0 & 94.2 & 86.5 & 94.1 & 96.7 & 95.6\tabularnewline
\midrule
\multirow{2}{*}{\emph{$n=200$}} & \multicolumn{2}{c}{JMM} & \multicolumn{2}{c}{OS} & \multicolumn{2}{c}{BG}\tabularnewline
\cmidrule{2-7}
 & $\b_{1}$ & $\b_{2}$ & $\b_{1}$ & $\b_{2}$ & $\b_{1}$ & $\b_{2}$\tabularnewline
\midrule
Mean Bias & 2.67 & -2.69 & 2.49 & -2.54 & -0.08 & 0.01\tabularnewline
Median Bias & 2.66 & -2.62 & 2.47 & -2.32 & -0.11 & 0.24\tabularnewline
Standard Deviation & 3.06 & 7.44 & 3.04 & 7.43 & 2.99 & 7.34\tabularnewline
Mean Standard Error & 3.05 & 7.51 & 3.04 & 7.44 & 3.04 & 7.44\tabularnewline
Mean Absolute Bias & 3.28 & 6.28 & 3.16 & 6.24 & 2.38 & 5.90\tabularnewline
Median Absolute Bias & 2.89 & 5.18 & 2.78 & 5.29 & 1.98 & 5.06\tabularnewline
RMSE & 4.06 & 7.91 & 3.93 & 7.86 & 2.99 & 7.34\tabularnewline
90\% Coverage Rate & 78.5 & 88.5 & 79.3 & 89.0 & 90.0 & 91.2\tabularnewline
95\% Coverage Rate & 86.6 & 94.0 & 86.7 & 94.1 & 94.9 & 95.0\tabularnewline
\bottomrule
\end{tabular}

\begin{tablenotes}
\setlength{\itemindent}{-0.5cm}
\item \footnotesize \textit{Note:} All values have been multiplied by 100.
\end{tablenotes}
\end{threeparttable}

\end{table}

\paragraph{Sparser networks.} We examine the performance of the method in networks with fewer links on average. To this end, we lower all $\alpha_{i}$'s by one, resulting in a network density of 8.6\%. As reported in Table~\ref{tab:sim_beta_sparse}, network sparsity worsens the performance of all estimators. Nevertheless, the BG estimator continues to outperform the others across all metrics.

\begin{table}[!htbp]
\centering
\begin{centering}
\caption{NTU estimation results for $\protect\b_{0}$ under a sparser network}\label{tab:sim_beta_sparse}
\par\end{centering}
\begin{threeparttable}

\begin{tabular}{lrrrrrr}
\toprule 
\multirow{2}{*}{\emph{$n=100$}} & \multicolumn{2}{c}{JMM} & \multicolumn{2}{c}{OS} & \multicolumn{2}{c}{BG}\tabularnewline
\cmidrule{2-7}
 & $\b_{1}$ & $\b_{2}$ & $\b_{1}$ & $\b_{2}$ & $\b_{1}$ & $\b_{2}$\tabularnewline
\midrule 
Mean Bias & 4.25 & -7.68 & 4.76 & -4.24 & -0.46 & 0.70\tabularnewline
Median Bias & 4.22 & -8.26 & 4.78 & -4.77 & -0.35 & 0.01\tabularnewline
Standard Deviation & 7.45 & 18.58 & 7.45 & 18.66 & 7.08 & 17.91\tabularnewline
Mean Standard Error & 7.39 & 17.91 & 7.36 & 17.81 & 7.36 & 17.81\tabularnewline
Mean Absolute Bias & 6.90 & 16.21 & 7.11 & 15.45 & 5.69 & 14.37\tabularnewline
Median Absolute Bias & 5.86 & 14.13 & 6.17 & 13.25 & 4.97 & 11.99\tabularnewline
RMSE & 8.58 & 20.11 & 8.84 & 19.14 & 7.09 & 17.93\tabularnewline
90\% Coverage Rate & 85.8 & 85.4 & 83.5 & 87.6 & 91.5 & 90.4\tabularnewline
95\% Coverage Rate & 91.7 & 92.2 & 90.1 & 93.2 & 96.5 & 95.1\tabularnewline
\midrule
\multirow{2}{*}{\emph{$n=200$}} & \multicolumn{2}{c}{JMM} & \multicolumn{2}{c}{OS} & \multicolumn{2}{c}{BG}\tabularnewline
\cmidrule{2-7}
 & $\b_{1}$ & $\b_{2}$ & $\b_{1}$ & $\b_{2}$ & $\b_{1}$ & $\b_{2}$\tabularnewline
\midrule
Mean Bias & 2.14 & -3.91 & 2.35 & -2.45 & -0.14 & -0.04\tabularnewline
Median Bias & 1.99 & -3.95 & 2.19 & -2.44 & -0.30 & 0.02\tabularnewline
Standard Deviation & 3.65 & 8.74 & 3.65 & 8.72 & 3.55 & 8.53\tabularnewline
Mean Standard Error & 3.59 & 8.72 & 3.58 & 8.68 & 3.58 & 8.68\tabularnewline
Mean Absolute Bias & 3.37 & 7.61 & 3.47 & 7.18 & 2.85 & 6.77\tabularnewline
Median Absolute Bias & 2.79 & 6.19 & 2.95 & 5.90 & 2.48 & 5.68\tabularnewline
RMSE & 4.23 & 9.58 & 4.34 & 9.06 & 3.56 & 8.53\tabularnewline
90\% Coverage Rate & 84.0 & 86.2 & 82.3 & 87.9 & 91.4 & 90.9\tabularnewline
95\% Coverage Rate & 90.3 & 93.0 & 88.8 & 94.0 & 95.2 & 95.2\tabularnewline
\bottomrule
\end{tabular}

\begin{tablenotes}
\setlength{\itemindent}{-0.5cm}
\item \footnotesize \textit{Note:} All values have been multiplied by 100.
\end{tablenotes}
\end{threeparttable}

\end{table}

\subsection{Transferable-Utility Extended Results}\label{subsec:tu_extended}

For the TU specification, we also conducted analogous exercises at $n=100$ and $n=200$ under both logistic and probit links, covering fixed-effect recovery, the APE estimator, link-function misspecification, and sparser networks. The qualitative conclusions match those reported for NTU: the BG estimator substantially reduces the bias of JMM and OS while preserving coverage close to the nominal level, and the APE results are consistent with the $O(n^{-1})$ bias bound of Theorem~\ref{thm:norm_ape}. The detailed tables are omitted for brevity and are available from the authors upon request.

\setcounter{equation}{0}
\setcounter{thm}{0}
\setcounter{lem}{0}
\setcounter{prop}{0}
\setcounter{rem}{0}
\renewcommand{\theequation}{D.\arabic{equation}}
\renewcommand{\thethm}{D.\arabic{thm}}
\renewcommand{\thelem}{D.\arabic{lem}}
\renewcommand{\theprop}{D.\arabic{prop}}
\renewcommand{\therem}{D.\arabic{rem}}

\section{Additional Details for Empirical Applications}\label{sec:sm_data}

This section provides detailed data descriptions and summary statistics for the two empirical applications in Section~\ref{sec:empirical_applications} of the main text.

\subsection{Townsend Thai Village Networks}\label{sec:sm_thai}

The Townsend Thai monthly panel is a long-running household survey covering rural and peri-urban villages in four provinces of Thailand. We use the subsample of $V=16$ villages for which the transaction-level network data have been compiled by \citet{kinnan2024propagation}. Each village has an average of $44$ households (min 24, max 50), and the data record monthly transactions across multiple domains from September 1998 through December 2012. We aggregate monthly transactions to the annual level, yielding a pooled dyad sample of $N=\sum_{v=1}^{16}\binom{n_v}{2}=15{,}641$ observations in each year.

We construct three binary dyadic outcome variables from the transaction records:
\begin{enumerate}[label=(\roman*),itemsep=2pt,topsep=2pt]
\item a \emph{financial} link, equal to $1$ if the two households exchange gifts or informal loans/repayments during the year;
\item an \emph{operations} link, equal to $1$ if they transact in production inputs, intermediate goods, or output sales/purchases; and
\item a \emph{labor} link, equal to $1$ if one household hires labor from the other or they exchange labor.
\end{enumerate}
We treat each resulting binary pairwise indicator as $Y_{ij,v}$ in (\ref{eq:observed_link}) under the TU specification. Because our theory requires a dense-network regime, for each outcome we select the year with the highest pooled density: year $3$ for the financial network, year $2$ for the operations network, and year $3$ for the labor network.

Our covariates are three dyadic variables constructed from baseline household characteristics, together denoted $X_{ij}$:
\begin{enumerate}[label=(\roman*),itemsep=2pt,topsep=2pt]
\item $X_{1,ij}=\ln(\lVert D_i-D_j\rVert_2)$, the (ln) demographic distance, where $D_i$ is a vector of baseline household composition and head characteristics (number of males, females, adults, children, total size, mean age, and head's sex, age, and education);
\item $X_{2,ij}=\left|\ln(W_i)-\ln(W_j)\right|$, the absolute (ln) difference in baseline net worth $W_i$ (aggregated from the household financial accounts); and
\item $X_{3,ij}\in\{0,1\}$, an indicator for kinship based on the reported relationship codes in the household roster.
\end{enumerate}
Summary statistics for the three outcome variables and the covariates, pooled across villages, are reported in Table~\ref{tab:sum_stat_thai}.

\begin{table}[!htbp]
\caption{Summary statistics for the Townsend Thai village networks (pooled across $V=16$ villages)}\label{tab:sum_stat_thai}
\centering
\begin{tabular}{lrrrr}
\toprule
Variable & Mean & Std. Dev. & Min & Max \\
\midrule
Financial link & 0.0146 & 0.1201 & 0 & 1 \\
Operations link & 0.0536 & 0.2252 & 0 & 1 \\
Labor link & 0.0779 & 0.2681 & 0 & 1 \\
(ln) Demographic difference & 5.4840 & 0.6193 & 1.7132 & 7.0546 \\
(ln) Net-worth difference & 2.0698 & 4.0895 & 0 & 32.7114 \\
Kinship & 0.0330 & 0.1786 & 0 & 1 \\
\bottomrule
\end{tabular}
\end{table}

\subsection{Nyakatoke Risk-Sharing Network}\label{sec:sm_nyakatoke}

The network data of Nyakatoke, located in the Kagera Region of Tanzania,
covers a small Haya community of all 119 households. We investigate
how important wealth difference, distance, and blood or religious
ties are in deciding the formation of risk-sharing links among local residents.
The dataset includes the following variables: (i) whether or not two
households are linked in the insurance network, (ii) total USD assets
and religion of each household, (iii) kinship and distance between
households. To define the dependent variable \textit{link}, each household
was asked:

\medskip

``\textit{Can you give a list of people from inside or outside of
Nyakatoke, who you can personally rely on for help and/or that can
rely on you for help in cash, kind or labor?}''

\medskip

The data contains three answers of ``bilaterally mentioned'', ``unilaterally
mentioned'', and ``not mentioned'' between each pair of households.
Considering the question is about whether one can rely on the other
for help, we interpret both ``bilaterally mentioned'' and ``unilaterally
mentioned'' as indicating that they are connected in this undirected network.
This coding is broader than literal bilateral consent and should be
interpreted as a proxy for an underlying risk-sharing relationship.
In the context of village economies, these links are unlikely to be
formed through explicit side-payment transfers, making NTU the more
natural benchmark.

We estimate the coefficients for three regressors: \emph{wealth difference},
\emph{distance} and \emph{tie} between households. \textit{Wealth}
is defined as the total assets in USD owned by each household, including
livestock, durables and land. \textit{Distance} measures how far
away two households are located in meters. \textit{Tie} is a discrete
variable, with the value ``3'' if members of one household are parents,
children and/or siblings of members of the other household, ``2''
if nephews, nieces, aunts, cousins, grandparents and grandchildren,
``1'' if any other blood relation applies or if two households share
the same religion, and ``0'' if no blood or religious tie exists. Following
the literature, we take the natural logarithm on \textit{wealth} and \textit{distance},
and we construct the \textit{wealth difference} variable as the absolute
difference in \emph{wealth}, i.e.,
\[
X_{ij}=\left(\abs{\ln\left(\text{wealth}_{i}\right)-\ln\left(\text{wealth}_{j}\right)},\ \ln\left(\text{distance}_{ij}\right),\ \text{tie}_{ij}\right)^{\top}.
\]

Five households in the data have no information on \textit{wealth}
and/or \textit{distance}. We drop these observations, resulting in
a sample of $n=114$ households and $N=6{,}441$ dyadic observations. Table~\ref{tab:empirical_application_summary} reports
the summary statistics.

\begin{table}[!htbp]
\caption{Summary statistics for the Nyakatoke network}\label{tab:empirical_application_summary}

\vspace{0cm}
\centering{}%
\begin{tabular}{lrrrr}
\toprule 
Variables & Mean & Std. Dev. & Min & Max\tabularnewline
\midrule
link & 0.0733 & 0.2606 & 0.0000 & 1.0000\tabularnewline
(ln) wealth difference & 1.0365 & 0.8227 & 0.0004 & 5.8898\tabularnewline
(ln) distance & 6.0553 & 0.7092 & 2.6672 & 7.4603\tabularnewline
tie & 0.4260 & 0.6123 & 0.0000 & 3.0000\tabularnewline
\bottomrule
\end{tabular}
\end{table}


\singlespacing

\renewcommand{\refname}{References for Supplemental Material}
\putbib
\end{bibunit}


\begin{thebibliography}{}

\bibitem[\protect\citeauthoryear{Bonhomme}{Bonhomme}{2012}]{bonhomme2012functional}
Bonhomme, S. (2012).
\newblock Functional differencing.
\newblock {\em Econometrica\/}~{\em 80\/}(4), 1337--1385.

\bibitem[\protect\citeauthoryear{Bonhomme, Jochmans, and Weidner}{Bonhomme
  et~al.}{2024}]{bonhomme2024neyman}
Bonhomme, S., K.~Jochmans, and M.~Weidner (2024).
\newblock A {Neyman}-orthogonalization approach to the incidental parameter
  problem.
\newblock arXiv preprint arXiv:2412.10304.

\bibitem[\protect\citeauthoryear{Breiman}{Breiman}{1996}]{breiman1996bagging}
Breiman, L. (1996).
\newblock Bagging predictors.
\newblock {\em Machine Learning\/}~{\em 24\/}(2), 123--140.

\bibitem[\protect\citeauthoryear{Chamberlain}{Chamberlain}{1984}]{chamberlain1984panel}
Chamberlain, G. (1984).
\newblock Panel data.
\newblock In Z.~Griliches and M.~D. Intriligator (Eds.), {\em Handbook of
  Econometrics}, Volume~2, pp.\  1247--1318. North-Holland.

\bibitem[\protect\citeauthoryear{Chatterjee, Diaconis, and Sly}{Chatterjee
  et~al.}{2011}]{chatterjee2011random}
Chatterjee, S., P.~Diaconis, and A.~Sly (2011).
\newblock Random graphs with given degree sequence.
\newblock {\em Annals of Applied Probability\/}~{\em 21\/}(4), 1400--1435.

\bibitem[\protect\citeauthoryear{Chen, Fern{\'a}ndez-Val, and Weidner}{Chen
  et~al.}{2021}]{chen2021nonlinear}
Chen, M., I.~Fern{\'a}ndez-Val, and M.~Weidner (2021).
\newblock Nonlinear factor models for network and panel data.
\newblock {\em Journal of Econometrics\/}~{\em 220\/}(2), 296--324.

\bibitem[\protect\citeauthoryear{Chen, Chernozhukov, Lee, and Newey}{Chen
  et~al.}{2014}]{chen2014local}
Chen, X., V.~Chernozhukov, S.~Lee, and W.~K. Newey (2014).
\newblock Local identification of nonparametric and semiparametric models.
\newblock {\em Econometrica\/}~{\em 82\/}(2), 785--809.

\bibitem[\protect\citeauthoryear{{de Paula}}{{de
  Paula}}{2020a}]{de2020econometric}
{de Paula}, {\'A}. (2020a).
\newblock Econometric models of network formation.
\newblock {\em Annual Review of Economics\/}~{\em 12}, 775--799.

\bibitem[\protect\citeauthoryear{{de Paula}}{{de Paula}}{2020b}]{de2020trade}
{de Paula}, {\'A}. (2020b).
\newblock Trade networks and the strength of strong ties.
\newblock {\em Advances in Econometrics\/}~{\em 42}.

\bibitem[\protect\citeauthoryear{{de Paula}, Richards-Shubik, and Tamer}{{de
  Paula} et~al.}{2018}]{de2018identifying}
{de Paula}, {\'A}., S.~Richards-Shubik, and E.~Tamer (2018).
\newblock Identifying preferences in networks with bounded degree.
\newblock {\em Econometrica\/}~{\em 86\/}(1), 263--288.

\bibitem[\protect\citeauthoryear{De~Weerdt}{De~Weerdt}{2004}]{deweerdt2004risk}
De~Weerdt, J. (2004).
\newblock Risk-sharing and endogenous network formation.
\newblock In {\em Insurance Against Poverty}. Oxford University Press.

\bibitem[\protect\citeauthoryear{Dhaene and Jochmans}{Dhaene and
  Jochmans}{2015}]{dhaene2015split}
Dhaene, G. and K.~Jochmans (2015).
\newblock Split-panel jackknife estimation of fixed-effect models.
\newblock {\em Review of Economic Studies\/}~{\em 82\/}(3), 991--1030.

\bibitem[\protect\citeauthoryear{Dzemski}{Dzemski}{2019}]{dzemski2019empirical}
Dzemski, A. (2019).
\newblock An empirical model of dyadic link formation in network with
  unobserved heterogeneity.
\newblock {\em Review of Economics and Statistics\/}~{\em 101\/}(5), 763--776.

\bibitem[\protect\citeauthoryear{Fern{\'a}ndez-Val and
  Weidner}{Fern{\'a}ndez-Val and Weidner}{2016}]{fernandez2016individual}
Fern{\'a}ndez-Val, I. and M.~Weidner (2016).
\newblock Individual and time effects in nonlinear panel models with large {N,
  T}.
\newblock {\em Journal of Econometrics\/}~{\em 192\/}(1), 291--312.

\bibitem[\protect\citeauthoryear{Fern{\'a}ndez-Val and
  Weidner}{Fern{\'a}ndez-Val and Weidner}{2018}]{fernandez2018fixed}
Fern{\'a}ndez-Val, I. and M.~Weidner (2018).
\newblock Fixed effects estimation of large-{T} panel data models.
\newblock {\em Annual Review of Economics\/}~{\em 10\/}(1), 109--138.

\bibitem[\protect\citeauthoryear{Gao}{Gao}{2020}]{gao2020nonparametric}
Gao, W.~Y. (2020).
\newblock Nonparametric identification in index models of link formation.
\newblock {\em Journal of Econometrics\/}~{\em 215\/}(2), 399--413.

\bibitem[\protect\citeauthoryear{Gao, Li, and Xu}{Gao
  et~al.}{2023}]{gao2023logical}
Gao, W.~Y., M.~Li, and S.~Xu (2023).
\newblock Logical differencing in dyadic network formation models with
  nontransferable utilities.
\newblock {\em Journal of Econometrics\/}~{\em 235\/}(1), 302--324.

\bibitem[\protect\citeauthoryear{Gao, Li, and Xu}{Gao
  et~al.}{2026}]{GaoLiXu2026TIS}
Gao, W.~Y., M.~Li, and Z.~Xu (2026).
\newblock Tractable identification of strategic network formation models with
  unobserved heterogeneity.
\newblock Working paper.

\bibitem[\protect\citeauthoryear{Graham}{Graham}{2017}]{graham2017econometric}
Graham, B.~S. (2017).
\newblock An econometric model of network formation with degree heterogeneity.
\newblock {\em Econometrica\/}~{\em 85\/}(4), 1033--1063.

\bibitem[\protect\citeauthoryear{Graham}{Graham}{2020}]{graham2020network}
Graham, B.~S. (2020).
\newblock Network data.
\newblock In {\em Handbook of Econometrics}, Volume~7, pp.\  111--218.
  Elsevier.

\bibitem[\protect\citeauthoryear{Hahn and Newey}{Hahn and
  Newey}{2004}]{hahn2004jackknife}
Hahn, J. and W.~Newey (2004).
\newblock Jackknife and analytical bias reduction for nonlinear panel models.
\newblock {\em Econometrica\/}~{\em 72\/}(4), 1295--1319.

\bibitem[\protect\citeauthoryear{Hirano and Wright}{Hirano and
  Wright}{2017}]{hirano2017forecasting}
Hirano, K. and J.~H. Wright (2017).
\newblock Forecasting with model uncertainty: Representations and risk
  reduction.
\newblock {\em Econometrica\/}~{\em 85\/}(2), 617--643.

\bibitem[\protect\citeauthoryear{Honor{\'e} and {de Paula}}{Honor{\'e} and {de
  Paula}}{2021}]{honore2021identification}
Honor{\'e}, B.~E. and {\'A}.~{de Paula} (2021).
\newblock Identification in simple binary outcome panel data models.
\newblock {\em Econometrics Journal\/}~{\em 24\/}(2).

\bibitem[\protect\citeauthoryear{Hughes}{Hughes}{2026}]{hughes2026estimating}
Hughes, D.~W. (2026).
\newblock A jackknife bias correction for nonlinear network data models with
  fixed effects.
\newblock {\em Journal of Econometrics\/}~{\em 253}, 106130.

\bibitem[\protect\citeauthoryear{Jackson and Wolinsky}{Jackson and
  Wolinsky}{1996}]{jackson1996strategic}
Jackson, M.~O. and A.~Wolinsky (1996).
\newblock A strategic model of social and economic networks.
\newblock {\em Journal of Economic Theory\/}~{\em 71\/}(1), 44--74.

\bibitem[\protect\citeauthoryear{Kinnan, Samphantharak, Townsend, and
  Vera-Cossio}{Kinnan et~al.}{2024}]{kinnan2024propagation}
Kinnan, C., K.~Samphantharak, R.~Townsend, and D.~Vera-Cossio (2024).
\newblock Propagation and insurance in village networks.
\newblock {\em American Economic Review\/}~{\em 114\/}(1), 252--284.

\bibitem[\protect\citeauthoryear{Le~Cam}{Le~Cam}{1969}]{le1969theorie}
Le~Cam, L.~M. (1969).
\newblock {\em Th{\'e}orie Asymptotique de la D{\'e}cision Statistique}.
\newblock Presses de l'Universit{\'e} de Montr{\'e}al.

\bibitem[\protect\citeauthoryear{Liao, Mei, and Shi}{Liao
  et~al.}{2024}]{liao2024nickell}
Liao, C., Z.~Mei, and Z.~Shi (2024).
\newblock {Nickell} meets {Stambaugh}: A tale of two biases in panel predictive
  regressions.
\newblock {\em arXiv preprint arXiv:2410.09825\/}.

\bibitem[\protect\citeauthoryear{Mei, Sheng, and Shi}{Mei
  et~al.}{2026}]{mei2026nickell}
Mei, Z., L.~Sheng, and Z.~Shi (2026).
\newblock Nickell bias in panel local projection: Financial crises are worse
  than you think.
\newblock {\em Journal of International Economics\/}, 104210.

\bibitem[\protect\citeauthoryear{Qu, Chen, Yan, and Chen}{Qu
  et~al.}{2025}]{Qu2025}
Qu, L., L.~Chen, T.~Yan, and Y.~Chen (2025).
\newblock Inference in semiparametric formation models for directed networks.
\newblock {\em Journal of Business \& Economic Statistics\/}.
\newblock Forthcoming.

\bibitem[\protect\citeauthoryear{Rao}{Rao}{1992}]{rao1992information}
Rao, C.~R. (1992).
\newblock Information and the accuracy attainable in the estimation of
  statistical parameters.
\newblock In {\em Breakthroughs in Statistics: Foundations and Basic Theory},
  pp.\  235--247. Springer.

\bibitem[\protect\citeauthoryear{Shi and Chen}{Shi and
  Chen}{2016}]{shi2016structural}
Shi, Z. and X.~Chen (2016).
\newblock A structural pairwise network model with individual heterogeneity.
\newblock Available at SSRN 2848375.

\bibitem[\protect\citeauthoryear{{v}an~der Vaart}{{v}an~der
  Vaart}{2000}]{van2000asymptotic}
{v}an~der Vaart, A.~W. (2000).
\newblock {\em Asymptotic Statistics}, Volume~3.
\newblock Cambridge University Press.

\bibitem[\protect\citeauthoryear{Yan}{Yan}{2019}]{yan2019approximating}
Yan, T. (2019).
\newblock Approximating the inverse of a diagonally dominant matrix with
  positive elements.
\newblock arXiv Preprint arXiv:1902.00668.

\bibitem[\protect\citeauthoryear{Yan, Jiang, Fienberg, and Leng}{Yan
  et~al.}{2019}]{yan2019statistical}
Yan, T., B.~Jiang, S.~E. Fienberg, and C.~Leng (2019).
\newblock Statistical inference in a directed network model with covariates.
\newblock {\em Journal of the American Statistical Association\/}~{\em
  114\/}(526), 857--868.

\bibitem[\protect\citeauthoryear{Yan, Qin, and Wang}{Yan
  et~al.}{2016}]{yan2016asymptotics}
Yan, T., H.~Qin, and H.~Wang (2016).
\newblock Asymptotics in undirected random graph models parameterized by the
  strengths of vertices.
\newblock {\em Statistica Sinica\/}~{\em 26\/}(1), 273--293.

\end{thebibliography}


\begin{thebibliography}{}

\bibitem[\protect\citeauthoryear{Boucheron, Lugosi, and Massart}{Boucheron
  et~al.}{2013}]{boucheron2013concentration}
Boucheron, S., G.~Lugosi, and P.~Massart (2013).
\newblock {\em Concentration Inequalities: A Nonasymptotic Theory of
  Independence}.
\newblock Oxford University Press.

\bibitem[\protect\citeauthoryear{Chatterjee, Diaconis, and Sly}{Chatterjee
  et~al.}{2011}]{chatterjee2011random}
Chatterjee, S., P.~Diaconis, and A.~Sly (2011).
\newblock Random graphs with given degree sequence.
\newblock {\em Annals of Applied Probability\/}~{\em 21\/}(4), 1400--1435.

\bibitem[\protect\citeauthoryear{Graham}{Graham}{2017}]{graham2017econometric}
Graham, B.~S. (2017).
\newblock An econometric model of network formation with degree heterogeneity.
\newblock {\em Econometrica\/}~{\em 85\/}(4), 1033--1063.

\bibitem[\protect\citeauthoryear{Kinnan, Samphantharak, Townsend, and
  Vera-Cossio}{Kinnan et~al.}{2024}]{kinnan2024propagation}
Kinnan, C., K.~Samphantharak, R.~Townsend, and D.~Vera-Cossio (2024).
\newblock Propagation and insurance in village networks.
\newblock {\em American Economic Review\/}~{\em 114\/}(1), 252--284.

\bibitem[\protect\citeauthoryear{{v}an~der Vaart}{{v}an~der
  Vaart}{2000}]{van2000asymptotic}
{v}an~der Vaart, A.~W. (2000).
\newblock {\em Asymptotic Statistics}, Volume~3.
\newblock Cambridge University Press.

\bibitem[\protect\citeauthoryear{Yan}{Yan}{2019}]{yan2019approximating}
Yan, T. (2019).
\newblock Approximating the inverse of a diagonally dominant matrix with
  positive elements.
\newblock arXiv Preprint arXiv:1902.00668.

\bibitem[\protect\citeauthoryear{Yan, Jiang, Fienberg, and Leng}{Yan
  et~al.}{2019}]{yan2019statistical}
Yan, T., B.~Jiang, S.~E. Fienberg, and C.~Leng (2019).
\newblock Statistical inference in a directed network model with covariates.
\newblock {\em Journal of the American Statistical Association\/}~{\em
  114\/}(526), 857--868.

\bibitem[\protect\citeauthoryear{Yan and Xu}{Yan and Xu}{2013}]{yan2013central}
Yan, T. and J.~Xu (2013).
\newblock A central limit theorem in the $\beta$-model for undirected random
  graphs with a diverging number of vertices.
\newblock {\em Biometrika\/}~{\em 100\/}(2), 519--524.

\end{thebibliography}
\end{document}